\PassOptionsToPackage{table,dvipsnames}{xcolor}
\documentclass[sigconf,nonacm,screen]{acmart}

\AtBeginDocument{%
  \providecommand\BibTeX{{%
    Bib\TeX}}}

\usepackage{amsthm}
\usepackage{graphicx}
\usepackage{textcomp}

\usepackage{multirow}
\usepackage{booktabs}
\usepackage{graphicx}
\usepackage[dvipsnames,table]{xcolor}
\usepackage{tikz}
\usepackage[english]{babel}

\usepackage{rotating}
\usepackage{tablefootnote}
\usepackage[para]{threeparttable}
\usepackage[vlined,ruled,linesnumbered]{algorithm2e}

\usepackage{soul}

\usepackage{tikz}
\usepackage{amsmath}

\usepackage[normalem]{ulem}

\usepackage{comment}

\usepackage{balance}

\usepackage{pifont}%

\usepackage[para]{threeparttable}

\def\BibTeX{{\rm B\kern-.05em{\sc i\kern-.025em b}\kern-.08em
    T\kern-.1667em\lower.7ex\hbox{E}\kern-.125emX}}

\usepackage[varqu]{zi4}%
\AtBeginDocument{%
}

\definecolor{charcoal}{rgb}{0.21, 0.27, 0.31}
\definecolor{Gray}{rgb}{0.5, 0.5, 0.5}
\definecolor{slateGray}{rgb}{0.44, 0.5, 0.56}
\definecolor{smoke}{rgb}{0.52, 0.53, 0.52}
\definecolor{steelGray}{rgb}{0.44, 0.47, 0.49}

\definecolor{Maroon}{cmyk}{0, 0.87, 0.68, 0.32}
\definecolor{darkgreen}{rgb}{0.0, 0.2, 0.13}
\definecolor{JungleGreen}{rgb}{0.16, 0.67, 0.53}

\usepackage{listings}
\usepackage[dvipsnames]{xcolor}
\definecolor{codegreen}{rgb}{0,0.56,0.56}
\definecolor{codegray}{rgb}{0.3,0.3,0.3}
\definecolor{codepurple}{rgb}{0.7,0,0.82}
\definecolor{codeblack}{rgb}{0,0,0}
\definecolor{backcolor}{rgb}{0.99,0.99,0.99}
\definecolor{keywordblue}{rgb}{0.1,0.1,1}
\definecolor{commentgray}{rgb}{0.65,0.65,0.65}

\newcommand{\ignore}[1]{}%

\lstdefinestyle{mystyle}{
    frame=single,
    backgroundcolor=\color{backcolor},   
    commentstyle=\sffamily\it\color{commentgray},
    keywordstyle=\color{keywordblue},
    numberstyle=\tiny\color{codeblack},
    stringstyle=\color{codepurple},
    basicstyle=\rmfamily\footnotesize\color{black},
    breakatwhitespace=false,         
    breaklines=true,                 
    captionpos=b,                    
    keepspaces=true,                 
    numbers=left,                    
    numbersep=4pt,                  
    showspaces=false,                
    showstringspaces=false,
    showtabs=false,                  
    tabsize=4,
    float=tp,
    floatplacement=tbp,
    abovecaptionskip=0pt
}

\usepackage[frozencache,cachedir=.]{minted}
\definecolor{codegreen}{rgb}{0,0.56,0.56}
\definecolor{codegray}{rgb}{0.3,0.3,0.3}
\definecolor{codepurple}{rgb}{0.7,0,0.82}
\definecolor{codeblack}{rgb}{0,0,0}
\definecolor{backcolor}{rgb}{0.99,0.99,0.99}
\definecolor{keywordblue}{rgb}{0.1,0.1,1}
\definecolor{commentgray}{rgb}{0.65,0.65,0.65}

\newlength{\mintednumbersep}
\AtBeginDocument{%
  \sbox0{\tiny00}%
  \setlength\mintednumbersep{\parindent}%
  \addtolength\mintednumbersep{-\wd0}%
}

\setminted[c]{ %
    linenos=true,             %
    autogobble=true,          %
    breaklines=true,
    bgcolor=backcolor,
    frame=lines,
    framesep=1mm,
    fontsize=\footnotesize,
    xleftmargin=5pt,
    numbersep=1pt
}

\usepackage{multicol}
\usepackage{amsmath}

\usepackage{pgf}
\usetikzlibrary{arrows,automata}
\usetikzlibrary{positioning}
\usetikzlibrary{shapes.geometric}
\usetikzlibrary{plotmarks}
\usetikzlibrary{patterns}
\usetikzlibrary{calc}
\usetikzlibrary{fit}
\usetikzlibrary{shapes.misc}
\usetikzlibrary{backgrounds}
\usetikzlibrary{decorations.markings}
\usetikzlibrary{fit}
\usetikzlibrary{arrows.meta}
\usetikzlibrary{cd}

\usetikzlibrary{pgfplots.groupplots}

\makeatletter

\usepackage{mathpartir}
\usepackage{thmtools}
\usepackage{thm-restate}

\usepackage{paralist}
\usepackage{listings}
\usepackage{color}
\usepackage{graphicx}
\usepackage{float}
\usepackage{textcomp}
\usepackage{xspace}
\usepackage{pifont}
\usepackage{multirow}
\usepackage[noend]{algpseudocode}
\usepackage{pgfplots,booktabs}
\usepackage{dirtree}
\usepackage{syntax}
\usepackage{bm}
\usepackage{footnote}
\usepackage{stmaryrd}
\usepackage{mathtools}
\usepackage{diagbox}
\usepackage{tabularx}
\usepackage{xifthen}

\usepackage{enumitem}
\setlistdepth{15} 
\newlist{longitemize}{itemize}{15}
\setlist[longitemize,1]{label=\textbullet}
\setlist[longitemize,2]{label=\textbullet}
\setlist[longitemize,3]{label=\textbullet}
\setlist[longitemize,4]{label=\textbullet}
\setlist[longitemize,5]{label=\textbullet}
\setlist[longitemize,6]{label=\textbullet}
\setlist[longitemize,7]{label=\textbullet}
\setlist[longitemize,8]{label=\textbullet}
\setlist[longitemize,9]{label=\textbullet}
\setlist[longitemize,10]{label=\textbullet}
\setlist[longitemize,11]{label=\textbullet}
\setlist[longitemize,12]{label=\textbullet}
\setlist[longitemize,13]{label=\textbullet}
\setlist[longitemize,14]{label=\textbullet}
\setlist[longitemize,15]{label=\textbullet}

\usepackage[htt]{hyphenat}

\usepackage{multirow}

\usepackage{subcaption}

\usepackage{url} 

\usepackage[framemethod=TikZ]{mdframed}

\usepackage{fancyhdr}

\definecolor{ao(english)}{rgb}{0.0, 0.5, 0.0}
\definecolor{royalblue(web)}{rgb}{0.25, 0.41, 0.88}

\newcommand{\setCommentColor}[1]{%
	\ifthenelse{\equal{#1}{bk}}%
		{\colorlet{colorVar}{red!50}}%
		{\ifthenelse{\equal{#1}{pv}}%
			{\colorlet{colorVar}{blue}}%
			{\ifthenelse{\equal{#1}{mg}}%
				{\colorlet{colorVar}{ao(english)}}%
			{\ifthenelse{\equal{#1}{jr}}%
				{\colorlet{colorVar}{magenta}}%
				{}%
			}%
		}%
	}%
}

\definecolor{codegray}{gray}{0.95}

\newtheorem{definition}{Definition}

\newcommand{\tup}[1]{\langle #1 \rangle}

\newcommand{\pc}{\mathit{pc}}
\newcommand{\buf}{\mathit{buf}}

\newcommand{\exprEval}[2]{\llbracket #1 \rrbracket(#2)}

\newcommand{\srclang}{$\mu$\textsc{Asm}}

\newcommand{\concat}{\cdot}

\renewcommand{\paragraph}[1]{\smallskip\noindent\textbf{#1:\ }}

\newcommand{\highlightBox}[1]{\colorbox{cassandraTheme!18}{$#1$}}

\newcommand{\Var}{\mathit{Regs}}
\newcommand{\Val}{\mathit{Vals}}
\newcommand{\Nat}{\mathbb{N}}

\newcommand{\Obs}{\mathit{Obs}}

\newcommand{\cfObs}{\mathit{CfObs}}
\newcommand{\memObs}{\mathit{MemObs}}
\newcommand{\Prg}{\mathit{p}}

\newcommand{\lbl}{\ell}
\newcommand{\kywd}[1]{\mathbf{#1}}

\newcommand{\storeKywd}{\kywd{store}}
\newcommand{\loadKywd}{\kywd{load}}

\newcommand{\jzKywd}{\kywd{beqz}}
\newcommand{\callKywd}{\kywd{call}}
\newcommand{\retKywd}{\kywd{ret}}
\newcommand{\pcKywd}{\kywd{pc}}

\renewcommand{\pc}{\pcKywd}
\definecolor{obsColor}{rgb}{0.0, 0.34, 0.25}
\newcommand{\obsKywd}[1]{\textcolor{obsColor}{\mathtt{#1}}}
\newcommand{\loadObsKywd}{\obsKywd{load}}
\newcommand{\storeObsKywd}{\obsKywd{store}}
\newcommand{\pcObsKywd}{\obsKywd{pc}}
\newcommand{\callObsKywd}{\obsKywd{call}}
\newcommand{\retObsKywd}{\obsKywd{ret}}
\newcommand{\loadObs}[1]{\loadObsKywd\ #1}
\newcommand{\storeObs}[1]{\storeObsKywd\ #1}
\newcommand{\pcObs}[1]{\pcObsKywd\ #1}
\newcommand{\callObs}[1]{\callObsKywd\ #1}
\newcommand{\retObs}[1]{\retObsKywd\ #1}
\newcommand{\startObsKywd}[1]{\obsKywd{start}}
\newcommand{\commitObsKywd}[1]{\obsKywd{commit}} 
\newcommand{\rollbackObsKywd}[1]{\obsKywd{rollback}}

\newcommand{\commitObs}[1]{\commitObsKywd{}} %
\newcommand{\rollbackObs}[1]{\rollbackObsKywd{}} %
\newcommand{\unaryOp}[1]{\ominus #1}
\newcommand{\binaryOp}[2]{#1 \otimes #2}

\newcommand{\pseq}[2]{#1 ; #2}

\newcommand{\passign}[2]{#1 \leftarrow #2}

\newcommand{\pload}[2]{\loadKywd\ #1, #2}
\newcommand{\pstore}[2]{\storeKywd\ #1, #2}

\newcommand{\pcall}[1]{\callKywd\ #1}
\newcommand{\pret}{\retKywd{}}
\newcommand{\pjz}[2]{\jzKywd\ #1, #2}

\newcommand{\select}[2]{#1(#2)}

\newcommand{\fetch}[1]{
\ifthenelse{\equal{#1}{}}{\kywd{fetch}}{\kywd{fetch}}%
}

\newcommand{\apply}[2]{\mathsf{apl}(#1,#2)}

\newcommand{\mi}[1]{\ensuremath{\mathit{#1}}}
\newcommand{\mr}[1]{\ensuremath{\mathrm{#1}}}

\newcommand{\mf}[1]{\ensuremath{\mathbf{#1}}}

\newcommand{\ms}[1]{\ensuremath{\mathsf{#1}}}

\newcommand{\neutcol}[0]{black}
\newcommand{\stlccol}[0]{contractColor}
\newcommand{\ulccol}[0]{cassandraTheme}

\newcommand{\idecol}[0]{contractColor}

\newcommand{\col}[2]{\ensuremath{{\color{#1}{#2}}}}

\newcommand{\archStyle}[1]{\ms{\col{\stlccol}{#1}}}
\newcommand{\muarchStyle}[1]{{\mf{\col{\ulccol }{#1}}}}
\newcommand{\interfStyle}[1]{{\mr{\col{\idecol }{#1}}}}
\newcommand{\neut}[1]{{\mi{\col{\neutcol }{#1}}}}

\newcommand{\archStep}[2]{\archStyle{\xrightarrow[\neut{#2}]{\neut{#1}}}}

\newcommand{\muarchSem}[1]{\muarchStyle{\{\!\!|} #1\muarchStyle{ |\!\!\} } }

\newcommand{\interfSem}[1]{\interfStyle{\llbracket} #1\interfStyle{\rrbracket} }

\newcommand{\ControlAux}[2]{\mathcal{C}_{\interfStyle{#1}}^{\interfStyle{#2}}}
\newcommand{\MemAux}[2]{\mathcal{M}_{\interfStyle{#1}}^{\interfStyle{#2}}}

\newcommand{\CustomInterf}[3]{\interfStyle{\llbracket} #1\interfStyle{\rrbracket}_{\interfStyle{#2}}^{\interfStyle{#3}}}
\newcommand{\CtSeqInterf}[1]{\interfStyle{\llbracket} #1\interfStyle{\rrbracket}_{\interfStyle{ct}}^{\interfStyle{seq}}}
 
\newcommand{\ArchSeqInterf}[1]{\interfStyle{\llbracket} #1\interfStyle{\rrbracket}_{\interfStyle{arch}}^{\interfStyle{seq}}}

\newcommand{\hsni}[2]{#2 \vdash #1}

\newcommand{\CtxMuarchSem}[2]{ {\muarchSem{#1}}_{\muarchStyle{#2}} }

\newcommand{\ProposedMuarchSem}[1]{ \CtxMuarchSem{#1}{\nameAbbr{}} }

\newcommand{\CacheStates}{\mathit{CacheStates}}
\newcommand{\CacheState}{\mathit{cs}}
\newcommand{\CacheAccess}{\mathsf{access}}
\newcommand{\CacheUpdate}{\mathsf{update}}
\newcommand{\CacheHit}{\mathtt{Hit}}
\newcommand{\CacheMiss}{\mathtt{Miss}}

\newcommand{\BpStates}{\mathit{BpStates}}
\newcommand{\BpState}{\mathit{bp}}
\newcommand{\BpUpdate}{\mathsf{update}}
\newcommand{\BpPredict}{\mathsf{predict}}

\newcommand{\BtStates}{\mathit{TcStates}}
\newcommand{\BtState}{\mathit{tc}}
\newcommand{\BtAccess}{\mathsf{access}}
\newcommand{\BtUpdate}{\mathsf{update}}

\newcommand{\BtHit}{\mathtt{Hit}}
\newcommand{\BtMiss}{\mathtt{Miss}}

\newcommand{\ReorderBuffers}{\mathit{Bufs}}
\newcommand{\BufProject}[1]{ {#1}\!\!\downarrow }
\newcommand{\resolved}{\mathtt{R}}
\newcommand{\unresolved}{\mathtt{UR}}

\newcommand{\SchedStates}{\mathit{ScStates}}
\newcommand{\SchedState}{\mathit{sc}}
\newcommand{\SchedUpdate}{\mathsf{update}}
\newcommand{\SchedNext}{\mathsf{next}}
\newcommand{\FetchDir}{\mathtt{Fetch}}
\newcommand{\ExeDir}{\mathtt{Execute}}
\newcommand{\CommitDir}{\mathtt{Commit}}

\newcommand{\wMuarch}{\textbf{B}}

\newcommand{\cMuarch}{\textbf{C}}

\definecolor{Blue3}{HTML}{0000CD}
\definecolor{Green4}{HTML}{008B00}
\definecolor{Red3}{HTML}{CD0000}
\definecolor{orange}{rgb}{0.8, 0.47, 0.196}

\newcommand{\hrun}{\muarchStyle{hr}}
\newcommand{\crun}{\interfStyle{cr}}
\newcommand{\hrunp}{\muarchStyle{hr'}}
\newcommand{\crunp}{\interfStyle{cr'}}

\newcommand{\contractTrans}[1]{\overset{#1}{\textcolor{contractColor}{\rightsquigarrow}}}

\newcommand{\hwTrans}[1]{\overset{#1}{\textcolor{cassandraTheme}{\rightarrowtail}}}

\newcommand{\cryptotag}{\kappa}

\newcommand{\name}{\textsc{Cassandra}}
\newcommand{\nameAbbr}{csd}

\newcommand{\secref}[1]{\S\ref{#1}}

\newcommand{\BtUnit}{\textit{Branch Trace Unit}}
\newcommand{\BTU}{\textit{BTU}}
\newcommand{\TraceCache}{\textit{Trace Cache}}
\newcommand{\TC}{\textit{TRC}}
\newcommand{\PatternTable}{\textit{Pattern Table}}
\newcommand{\PT}{\textit{PAT}}
\newcommand{\CkptTable}{\textit{Checkpoint Table}}
\newcommand{\CT}{\textit{CPT}}

\newcommand{\baseline}{\textit{Unsafe Baseline}}
\newcommand{\stl}{\name{}+\textit{STL}}
\newcommand{\spt}{\textit{SPT}}
\newcommand{\prospect}{\textit{ProSpeCT}}
\newcommand{\lightVersion}{\name{}-\textit{lite}}

\newcommand{\raw}{\textit{raw}}

\newcommand{\vanilla}{\textit{vanilla}}
\newcommand{\Vanilla}{\textit{Vanilla}}
\newcommand{\tandem}{\textit{$k$-mers}}
\newcommand{\Tandem}{\textit{$k$-mers}}

\newcommand{\toyAES}{\texttt{Toy-AES-2}}

\newcommand{\targetBlue}[1]{\textcolor{blue}{#1}}

\newcommand{\targetRed}[1]{\textcolor{BrickRed}{#1}}

\newcommand{\pattern}[2]{#1\times#2}
\newcommand{\patternBlue}[2]{\textcolor{blue}{#1}\times#2}
\newcommand{\patternGreen}[2]{\textcolor{PineGreen}{#1}\times#2}
\newcommand{\patternRed}[2]{\textcolor{BrickRed}{#1}\times#2}

\newcommand{\ntBlue}[1]{\textcolor{blue}{#1}}
\newcommand{\ntGreen}[1]{\textcolor{PineGreen}{#1}}
\newcommand{\ntRed}[1]{\textcolor{BrickRed}{#1}}

\definecolor{cassandraTheme}{rgb}{0.5, 0.5, 0.0}
\definecolor{highlightColor}{rgb}{0.94,0.96,0.79}
\definecolor{contractColor}{rgb}{0.60, 0.09, 0.32}

\definecolor{challengeColor}{rgb}{0.36,0.0,0.0}
\newcommand{\challenge}[1]{\textcolor{challengeColor}{\textbf{Challenge #1}}}

\definecolor{insightColor}{rgb}{0.016,0.32,0.27}
\newcommand{\insight}[1]{\textcolor{insightColor}{\textbf{Insight #1}}}

\newcommand{\question}{\textcolor{BrickRed}{\textsc{Question}}}

\definecolor{veryDarkRed}{rgb}{0.25,0.0,0.0}
\newcommand{\questionN}[1]{\textcolor{veryDarkRed}{\textsc{Q#1}}}

\newcommand*\bluecircled[1]{%
   \tikz[baseline=(C.base)]\node[draw,circle,fill=black,inner sep=0.3pt](C) {#1};\!
}
\newcommand{\flowstep}[1]{$\bluecircled{\small \textcolor{white}{#1}}$\hspace{0.06cm}}
\newcommand{\flowstepletter}[1]{$\bluecircled{\footnotesize \textcolor{white}{#1}}$\hspace{0.07cm}}

\definecolor{cryptocurvecolor}{rgb}{0,0.33,0.57}
\usetikzlibrary{shapes.misc, positioning}
\newcommand*\croundrectangle[1]{%
   \tikz[baseline=(C.base)]\node[draw,rounded rectangle,fill=cryptocurvecolor,inner sep=0.4pt](C) {#1};\!
}
\newcommand{\cryptocode}[1]{$\croundrectangle{\small \textcolor{white}{#1}}$\hspace{0.07cm}}

\newcommand{\Mathcryptocode}[1]{\croundrectangle{\small \textcolor{white}{#1}}\hspace{0.07cm}}

\definecolor{noncryptocurvecolor}{rgb}{1,0.83,0.47}
\newcommand*\nroundrectangle[1]{%
   \tikz[baseline=(C.base)]\node[draw,gray,rounded rectangle,fill=noncryptocurvecolor,inner sep=0.4pt](C) {#1};\!
}
\newcommand{\noncryptocode}[1]{$\nroundrectangle{\small \textcolor{black}{#1}}$\hspace{0.07cm}}%

\newcommand{\Mathnoncryptocode}[1]{\nroundrectangle{\small \textcolor{black}{#1}}\hspace{0.07cm}}%

\newcommand*\captioncroundrectangle[1]{%
   \protect\tikz[baseline=(C.base)]\protect\node[draw,rounded rectangle,fill=cryptocurvecolor,inner sep=0.4pt](C) {#1};\!
}
\newcommand{\captioncryptocode}[1]{$\captioncroundrectangle{\small \textcolor{white}{#1}}$\hspace{0.07cm}}

\definecolor{backcolor}{rgb}{0.99,0.99,0.99}

\setminted[c]{ %
    style=colorful,
    linenos=true,             %
    autogobble=true,          %
    breaklines=true,
    bgcolor=backcolor,
    frame=lines,
    framesep=1mm,
    fontsize=\footnotesize,
    xleftmargin=5pt,
    numbersep=1pt
}

\newcommand{\InsightBox}[2][]{%
  \vspace{-7pt}
  \begin{center}
    \begin{tikzpicture}
      \node[draw=black,
            fill=gray!5!white,
            rounded corners=7pt,
            line width=0.6pt,
            inner xsep=6pt,
            inner ysep=3pt,
            text width=\dimexpr\linewidth - 12pt\relax, %
            align=justify,
            anchor=west,
            #1] 
      (box) {#2};
    \end{tikzpicture}
  \end{center}
}

\newcommand{\QuestionBox}[2][]{%
  \vspace{-6pt}
  \begin{center}
    \begin{tikzpicture}
      \node[draw=black,
            fill=gray!5!white,
            rounded corners=7pt,
            line width=0.6pt,
            inner xsep=6pt,
            inner ysep=3pt,
            text width=\dimexpr\linewidth - 12pt\relax, %
            align=justify,
            anchor=west,
            #1] 
      (box) {#2};
    \end{tikzpicture}
  \end{center}
}

\newcommand{\stepMap}[1]{~\textcolor{cassandraTheme}{\rightarrowtail_{\text{#1}}}~}
\newcommand{\stepMapN}[1]{\textcolor{cassandraTheme}{\rightarrowtail_{\text{#1}}}}
\newcommand{\stepArrow}[1]{\textcolor{cassandraTheme}{\xrightarrow{\textcolor{black}{\text{#1}}}}}

\newcommand{\vanillaAvgSize}{637,425}
\newcommand{\vanillaMaxSize}{90,110,880}

\newcommand{\tandemAvgSize}{19.9}
\newcommand{\tandemMaxSize}{2,312}

\newcommand{\avgCompressionRate}{163,371$\times$}
\newcommand{\maxCompressionRate}{30,036,960$\times$}

\newcommand{\storageKB}{1.74 KiB}

\newcommand{\ourPerf}{$1.85\%$}
\newcommand{\stlPerf}{$1.14\%$}
\newcommand{\sptPerf}{$12.07\%$}
\newcommand{\sptPerftoOurs}{$14.21\%$}

\newcommand{\sptChaChaPerf}{$59.8\%$}

\newcommand{\oursShaPerf}{$14.7\%$}

\newcommand{\ourPerfTwoFiftyHzFlush}{$1.80\%$}

\newcommand{\ourPerfChachaAllCrypto}{$2.8\%$}
\newcommand{\ourPerfCurveNinetyS}{$0.6\%$}

\newcommand{\ourPerfCurveAllCrypto}{$6.7\%$}

\newcommand{\ProspectChachaAllCrypto}{$0.8\%$}
\newcommand{\ProspectCurveNinetyS}{$2.5\%$}

\newcommand{\ProspectCurveAllCrypto}{$15.0\%$}

\newcommand{\ourPower}{$2.73\%$}
\newcommand{\ourArea}{$1.26\%$}

\newcommand{\revision}[1]{\textcolor{black}{#1}}
\newcommand{\removeRevision}[1]{}%

\newcommand{\inlineasm}[1]{\textcolor{blue}{\texttt{#1}}}
\newcommand{\inlineasmOneOperand}[2]{\textcolor{blue}{\texttt{#1}}\ #2}
\newcommand{\inlineasmTwoOperand}[3]{\textcolor{blue}{\texttt{#1}}\ #2,#3}

\definecolor{greentransitcolor}{rgb}{0, 0.57, 0.57}
\definecolor{redtransitcolor}{rgb}{0.75, 0, 0}
\definecolor{purpletransitcolor}{rgb}{0.75, 0, 0}
\definecolor{greytransitcolor}{rgb}{0.5, 0.5, 0.5}

\tikzset{SEQ/.style = {draw=greentransitcolor,
                      line width=#1, -{Straight Barb[length=3pt]}},
         SEQ/.default=1pt
}
\tikzset{SPEC/.style = {draw=redtransitcolor,
                      line width=#1, -{Straight Barb[length=3pt]}},
         SPEC/.default=1pt
}
\tikzset{NOTCARE/.style = {draw=black, 
                      line width=#1, -{Straight Barb[length=3pt]}},
         NOTCARE/.default=1pt
}

\def\genbox#1#2#3#4#5#6{%
    \leavevmode\raise#4bp\hbox to#5bp{\vrule height#5bp depth0bp width0bp
    \pdfliteral{q .5 w \csname #2COLOR\endcsname\space RG
                       \csname #3PDF\endcsname{#5}{#6} S Q
             \ifx1#1 q \csname #2COLOR\endcsname\space rg 
                       \csname #3PDF\endcsname{#5}{#6} f Q\fi}\hss}}

\def\sqbox      #1#2{\genbox{#1}{#2}  {sq}       {0}   {4.5}  {2.25}}
\def\trianbox   #1#2{\genbox{#1}{#2}  {trian}    {0}   {5}    {2.5}}

\def\circbox    #1#2{\genbox{#1}{#2}  {circ}     {0}   {5}    {2.5}}

\AtBeginDocument{%
  \providecommand\BibTeX{{%
    Bib\TeX}}}

\setcopyright{none} 
\copyrightyear{2025}
\acmYear{2025}
\acmDOI{XXXXXXX.XXXXXXX}

\acmConference['ISCA 2025']{The 52nd IEEE/ACM International Symposium on Computer Architecture}{June 21--25, 2025}{Tokyo, Japan}

\begin{document}

\title{\name{}: Efficient Enforcement of Sequential Execution for Cryptographic Programs}
\subtitle{(Extended Version)}

\author{Ali Hajiabadi}
\affiliation{%
  \institution{ETH Z\"{u}rich}
  \country{}
  \vspace{-0.3cm}
}
\authornote{This work was done while the author was at the National University of Singapore.}

\author{Trevor E. Carlson}
\affiliation{%
  \institution{National University of Singapore}
  \country{}
}

\renewcommand{\shortauthors}{Ali Hajiabadi and Trevor E. Carlson}

\begin{abstract}

Constant-time programming is a widely deployed approach to harden cryptographic programs against side channel attacks. However, modern processors often violate the underlying assumptions of standard constant-time policies by transiently executing unintended paths of the program. Despite many solutions proposed, addressing control flow misspeculations in an efficient way without losing performance is an open problem.

In this work, we propose \name{}, a novel hardware/software mechanism to enforce sequential execution for constant-time cryptographic code in a highly efficient manner. \name{} explores the radical design point of disabling the branch predictor and recording-and-replaying sequential control flow of the program. Two key insights that enable our design are that (1) the sequential control flow of a constant-time program is mostly static over different runs, and (2) cryptographic programs are loop-intensive and their control flow patterns repeat in a highly compressible way. 
These insights allow us to perform an upfront branch analysis that significantly compresses control flow traces. We add a small component to a typical processor design, the \BtUnit{}, to store compressed traces and determine fetch redirections according to the sequential model of the program.
Despite providing a strong security guarantee, \name{} counterintuitively
provides an average \ourPerf{} speedup compared to an unsafe baseline processor, mainly due to enforcing near-perfect fetch redirections\footnote{\name{} is published and appeared in Proceedings of 52nd International Symposium on Computer Architecture (ISCA 2025). DOI: \href{https://doi.org/10.1145/3695053.3731048}{10.1145/3695053.3731048}}.
\end{abstract}

\keywords{Cryptography, constant-time programming, speculative execution, hardware/software co-design}

\maketitle

\section{Introduction}

Protecting cryptographic programs has always been a major concern since they are the primary programs that process secrets. While the underlying cryptographic schemes provide strong levels of security to prevent secret extraction through cryptanalysis, their implementations can still be vulnerable to various side channel attacks. 
\textit{Constant-time programming} is a widely deployed approach to protect cryptographic programs against timing and memory side channels and it is the de facto coding discipline to write high-assurance cryptographic code~\cite{bearssl,zinzindohoue2017hacl,almeida2017jasmin,bernstein2012security}. 
Constant-time principles mandate the absence of secret dependent control flow and data flow. In other words, the attacker-visible observations of the execution 
must be independent of the confidential inputs of the program~\cite{almeida2016verifying}.

Unfortunately, speculative execution of modern processors violates standard constant-time principles that assume instructions are executed sequentially. 
After the advent of Spectre~\cite{Spectre2019Kocher}, several speculative execution attacks have demonstrated the ability to leak secrets from verified constant-time programs by transiently declassifying and leaking confidential states~\cite{shivakumar2023spectre,yavarzadeh2024pathfinder}.
For example, recent attacks have demonstrated powerful adversaries that can manipulate the Branch Prediction Unit (BPU) and
precisely control the paths executed by the victim 
and leak secrets from a constant-time AES implementation~\cite{yavarzadeh2024pathfinder}.
More extensive studies on the constant-time implementation of cryptographic programs demonstrate that most popular libraries leak secrets under speculative execution semantics~\cite{barthe2024testing}.
Hence, it is essential to find a solution that fundamentally eliminates the speculative execution attack surface in cryptographic programs. 

Existing defenses to protect constant-time programs, both on the hardware level~\cite{choudhary2021speculative,loughlin2021dolma,schwarz2020context,daniel2023prospect,hajiabadi2024levioso} and the software level~\cite{cauligi2020constant,guarnieri2020spectector,daniel2021hunting,barthe2021high,vassena2021blade,mosier2024serberus,pescosta2021bounded,wu2019abstract,oleksenko2020specfuzz,wang2020kleespectre,guo2020specusym,qi2021spectaint,wang2019oo7}, deploy a restrictive approach to prevent or limit speculative execution of instructions, diminishing the benefits of speculative, high-performance processors.

While naively disabling data flow speculation shows negligible performance impact for cryptographic programs, 
addressing control flow speculation is still a major issue for high-performance processors.
In this work, we investigate a new, radical design point to strictly enforce sequential execution for cryptographic programs, namely \textit{recording-and-replaying}. 
This mechanism disables branch prediction altogether, and instead,
redirects fetch based on the upfront recorded sequential control flow traces.
This design ensures that instruction fetch is always redirected according to the sequential execution model of the program, as assumed by standard constant-time policies. However, this idea has two major challenges:

\challenge{1}: Dynamic control flow traces change based on the program input; pre-computing control flow traces for all possible inputs in general-purpose applications is challenging, if not infeasible.

\challenge{2}: Control flow traces can be huge and storing/loading these traces in the processor would incur high overheads.
In the worst case, it can show a similar slowdown as a processor without a branch predictor which stalls fetch until the branch is resolved.

In this work, we discuss two key insights from constant-time cryptographic programs that overcome these challenges:

\insight{1}: Sequential control flow of constant-time programs are constant with respect to confidential inputs. In addition, public parameters of cryptographic programs are specified by standards or determined by the algorithm (e.g., the key length, number of encryption rounds, etc.). Hence, reusing just a single control flow trace over different runs of a program can be sufficient. However, control flow traces can still be extremely large (up to millions of decisions per static branch in our evaluated programs). As mentioned in \challenge{2}, storing and communicating a huge number of decisions per branch is not efficient, and a solution is needed.

\insight{2}: Most operations in cryptographic programs are in loops and they repeat the same operations over time. Detecting the repeating patterns of branch decisions would help to allow the storage of smaller, compressed patterns, and once loaded, the processor can replay the same pattern in the future.

Leveraging these insights, we propose \name{}, a hardware/software mechanism to enforce sequential execution for cryptographic programs and remove the control flow speculation attack surface within these programs.
To the best of our knowledge, 
\name{} is the first mechanism that takes advantage of the key characteristics of cryptographic applications, 
and counterintuitively, \textit{improves} performance.
The main artifacts of \name{} are twofold:

\textbf{(1) Branch analysis} (\secref{sec:motivation-and-analysis}). We perform an extensive branch analysis of cryptographic programs and devise a trace compression technique that significantly compresses branch traces. Our approach is inspired by DNA sequencing techniques that detect frequent and unknown patterns of nucleotides in large DNA sequences~\cite{marccais2011fast}. 
The average size of our new compressed traces is just 20 entries in BearSSL, OpenSSL, and post-quantum crypto primitives.

\textbf{(2) Microarchitecture} (\secref{sec:design}). We propose a new processor design that (1) communicates compressed branch traces to the processor, and (2) uses branch traces for fetch redirections while avoiding accessing and updating the branch predictor. We add a small, new component to the frontend, called the \BtUnit{} (\BTU{}), that efficiently stores and decompresses dynamic branch information.  

Additionally, we provide a detailed security analysis and 
discussion on how to deploy \name{} in conjunction with other defenses for a comprehensive Spectre mitigation (\secref{sec:security-analysis}). \name{} guarantees sequential execution for cryptographic programs that adhere to a constant-time policy and
can be easily integrated with other solutions that block Spectre attacks that violate software isolation (i.e., provide secure speculation for a sandboxing policy~\cite{guarnieri2021hardware}).

The main contributions of \name{} are as follows:
\begin{itemize}[leftmargin=0.5cm]
    \item Introducing a novel recording-and-replaying mechanism to strictly enforce sequential execution for constant-time cryptographic programs;
    \item Performing a detailed branch analysis and trace compression technique, inspired by DNA sequencing methods, that significantly compresses branch traces;
    \item Proposing an efficient design of \name{} that communicates branch traces with the hardware and enforces branch directions of a sequential execution model;
    \item Achieving a \ourPerf{} speedup over an unsafe baseline processor—delivering performance gains instead of slowdowns—while reducing power consumption by \ourPower{} and incurring only a \ourArea{} area overhead.
\end{itemize}

\section{Background}
\label{sec:background}

\subsection{Constant-Time Programming}
\label{sec:ct-definition}

Modern implementations of cryptographic applications deploy constant-time principles to harden programs against traditional side channels that exploit secret dependent behaviors of the program.
Constant-time principles satisfy confidential input indistinguishability to remove timing, cache, and memory side channels~\cite{almeida2016verifying}.
In other words, constant-time principles assume an adversary can observe the program counter, memory access patterns, and operands of variable-time instructions, and they guarantee all attacker-visible traces of a program are independent from the confidential inputs of the program~\cite{almeida2016verifying,cauligi2022sok}.

Standard constant-time policies provide security for a \textit{sequential execution} model, i.e., all instructions are executed in a sequential order specified by the architectural states of the program.
However, Spectre attacks have demonstrated the ability to leak secrets from constant-time programs in modern processors that use a \textit{speculative execution} model~\cite{yu2019data,shivakumar2023spectre,shivakumar2023typing,yavarzadeh2024pathfinder,barthe2024testing}. For example, Listing~\ref{listing:ct-spec-leak} shows a constant-time decryption of confidential input \texttt{m}. Sequential execution of the code dictates that the secret state is declassified (line 6) only after all decryption rounds are completed, after which any subsequent leak (line 7) is allowed.
However, in a speculative execution model, the \texttt{for} loop can be skipped due to misspeculation and directly leak the confidential input \texttt{m} before executing all decryption rounds, hence, violate constant-time policies of the program.

\begin{listing}[t]
\begin{minted}{c}
uint8 decrypt(uint8 m, uint8 *skey) 
{
    uint8 state = m; //m and state are secret
    for (int i = 0; i < num_rounds; i++)
        state = decrypt_ct(state, skey[i]); 
    uint8 d = declassify(state); //d is public
    return leak(d);
}
\end{minted}
\vspace{-0.75cm}
\caption{Constant-time decryption of \texttt{m}. Misspeculation and skipping the \texttt{for} loop can directly leak the secret \texttt{m}.}
\label{listing:ct-spec-leak}
\end{listing}

\subsection{Speculation Primitives}
\label{sec:spec-primitives}

Speculative execution can be triggered through different sources in modern processors, referred to as \textit{speculation primitives}. Speculation primitives can be categorized into control flow and data flow primitives~\cite{canella2019systematic,cauligi2022sok}.

\textbf{Control flow speculation}. The Branch Prediction Unit (BPU) in modern processors predicts the next PC after control flow instructions and fetches instructions speculatively from the predicted path. Control flow prediction allows the processor to avoid frontend stalls for cases where resolving control flow conditions depends on long latency operations. Prior attacks have demonstrated leaks via three general primitives in the BPU: 
\begin{enumerate}[leftmargin=0.9cm]
    \item[\textbf{PHT}] The Pattern History Table (PHT) predicts \textit{conditional direct branches} (e.g., \texttt{cmp [reg],0; je L}) with two possible outcomes of Taken and Not-Taken (e.g., Spectre-v1~\cite{Spectre2019Kocher}). 

    \item[\textbf{BTB}] The Branch Target Buffer (BTB) predicts \textit{indirect branches} (e.g., \texttt{jmp [reg]}) to determine the target address of next instruction (e.g., Spectre-v2~\cite{Spectre2019Kocher}). 

    \item[\textbf{RSB}] The Return Stack Buffer (RSB) predicts the target address of return instructions. While returns can considered as indirect branches, processors use the RSB to determine return addresses (e.g. Spectre-RSB~\cite{koruyeh2018spectre} and RetBleed~\cite{wikner2022retbleed}). 
\end{enumerate}
Note, that commodity microarchitectures might have multiple components that speculate on a specific type of branch. For example, GadgetSpinner~\cite{chen2024gadgetspinner} demonstrates that the Loop Stream Detector (LSD) in Intel CPUs also speculates on loop conditional branches. However, we use the aforementioned primitives to represent general classes of primitives (i.e., the LSD speculation falls into the PHT primitive).
Throughout this paper, we refer to all control flow instructions (direct, indirect, and return) as \textit{branches}. 

\textbf{Data flow speculation}. Modern processors deploy mechanisms for speculative execution of loads. Prior attacks demonstrated two primitives that can leak:
\begin{enumerate}[leftmargin=0.9cm]
    \item[\textbf{STL}] Store-to-load forwarding (STL) allows a load to forward data from a prior same-address store before all prior stores are resolved, without sending a request to the memory (e.g., Spectre-v4~\cite{horn2018speculative}). 

    \item[\textbf{PSF}] Predictive store forwarding (PSF) allows a younger load to forward data from an unresolved store before the load and store addresses are resolved (e.g., Spectre-PSF~\cite{cauligi2020constant}). 
\end{enumerate}

Mitigating control flow speculation poses higher overheads compared to data flow. 
Our experiments in \secref{sec:perf-results} show that naively addressing data flow speculation in cryptograhic programs incurs negligible performance overhead (less than 1\%).
Hence, we only focus on addressing control flow speculation in an efficient way.

\subsection{Evolution of Hardware Defenses for Spectre}
\label{sec:spectre-defenses}

Early defenses for speculative execution attacks focus only on data caches as the transmission channel, similar to the original Spectre-v1~\cite{yan2018invisispec,saileshwar2019cleanupspec,DOM,pashrashid2023hidfix,khasawneh2019safespec,pashrashid2022fast}. More comprehensive defenses, like STT~\cite{yu2019speculative} and NDA~\cite{weisse2019nda}, propose mechanisms to prevent leaks from a more comprehensive list of transmission channels. 
These solutions implement dynamic taint tracking to restrict the execution or data propagation for instructions that are tainted by speculatively loaded data. While this approach protects \textit{sandboxed} programs~\cite{guarnieri2021hardware}, they fail to protect constant-time programs, where secrets are loaded \textit{non-speculatively} 
(see line 3 in Listing~\ref{listing:ct-spec-leak}).

Recent Spectre defenses for constant-time programs extend prior solutions to protect non-speculative secrets as well~\cite{choudhary2021speculative,loughlin2021dolma,schwarz2020context,daniel2023prospect,hajiabadi2024levioso}. 
\removeRevision{For example, SPT~\cite{choudhary2021speculative} extends the taint tracking mechanism of STT and assumes all data in registers and memory are tainted unless they leak during the non-speculative, sequential execution of the program which means they are declassified intentionally and can be untainted.}
Most hardware-only defenses for constant-time programs introduce additional slowdown compared to the sandboxed cases. This is mainly because they must protect not only speculatively loaded data but also all values that are already loaded in the registers, as any of them can potentially be secret.
In this paper, our goal is to strictly enforce sequential execution for cryptographic code and avoid the additional overhead of prior solutions. To the best of our knowledge, our approach is the first that exploits the key characteristics of cryptographic code to improve performance compared to an unprotected baseline, while providing a sequential security guarantee.

\textbf{Motivating example}: DOLMA~\cite{loughlin2021dolma} shows that protecting non-cryptographic programs under a sandboxing policy incurs a 10.2\% performance overhead, rising to 22.3\% across \textit{all} applications when extended to a constant-time policy; this trend holds for all hardware-based defenses. \name{}, however, efficiently protects constant-time programs and allows the CPU to select more efficient defenses for other applications which better fit their threat model.

\section{Threat Model}
\label{sec:threat-model}

\name{} eliminates the possibility of transient execution exclusively for cryptographic code that adheres to the sequential constant-time policy.
\name{} does not provide protection for software isolation (i.e., sandboxing policy~\cite{guarnieri2021hardware}). Existing lightweight isolation techniques~\cite{schwarzl2022robust,reis2019site,hertogh2023quarantine} or secure speculation mechanisms for sandboxing~\cite{yu2019speculative,loughlin2021dolma,weisse2019nda,hajiabadi2024levioso} can be integrated with \name{} to prevent transient leaks of non-crypto code as well.

We consider Meltdown-type attacks~\cite{Lipp2018meltdown,van2018foreshadow,canella2019fallout,van2019ridl,schwarz2019zombieload} out of scope. These attacks exploit the transient execution upon exceptions and CPU faults, which are efficiently mitigated in recent CPUs via microcode updates~\cite{intel-affected-cpus}.
Additionally, non-speculative control flow attacks~\cite{evtyushkin2018branchscope,puddu2021frontal,hajiabadi2024conjuring} are out of scope; constant-time programs are inherently safe against such attacks.

\section{Branch Analysis of Constant-Time Cryptographic Programs}
\label{sec:motivation-and-analysis}

In this section, we investigate the practicality of a \textit{recording-and-replaying} solution for cryptographic programs to enforce sequential execution.
In \secref{sec:key-insights}, we discuss the key insights that enable our proposed solution, and in \secref{sec:branch-analysis}, we detail our branch analysis.

\subsection{Key Insights}
\label{sec:key-insights}

We discuss two key insights that are directly derived from fundamental characteristics of constant-time cryptographic programs.

\InsightBox{\insight{1}: \textit{Sequential control flow of constant-time programs is independent of confidential inputs and is determined by the algorithm and its implementation, which are known before execution.}}

As we discussed in \secref{sec:ct-definition}, constant-time principles assume that the entire control flow trace and accessed memory addresses are leaked~\cite{almeida2016verifying}. Hence, the dynamic control flow of the program is required to be independent from confidential inputs. On the other hand, public parameters of the cryptographic programs are specified by standards or determined by the underlying scheme and its implementation, e.g., the key length, array sizes, number of encryption rounds, etc. As a result, the sequential and dynamic control flow of these programs is known before execution and does not change during runtime. This enables us to pre-compute sequential branch traces and enforce them during runtime, instead of using the BPU to predict the branch directions.

While branch traces of cryptographic programs can be computed before execution, they can still be prohibitively large and incur penalties to load them in the CPU. Our \insight{2} enables us to significantly compress the branch traces; fitting the entire trace of most branches into a single entry of a small structure in the CPU.

\InsightBox{\insight{2}: \textit{Sequential control flow of cryptographic programs is highly regular and loop-intensive, allowing for significant compression of control flow traces.}}

Most operations and transformations of constant-time cryptographic programs occur in loops (like Listing~\ref{listing:ct-spec-leak}); standard constant-time policies
allow one to wrap the operations in loops if the loop count is public. 
Hence, this insight enables us to detect the repeating patterns of each branch and only communicate this pattern with the CPU to repeatedly replay. 

\textbf{Example: ChaCha20}~\cite{nir2018chacha20} is a stream cipher that starts with an internal state $S$ as a $4 \times 4$ matrix of 32-bit values, consisting of the secret key (256-bit), the nonce (96-bit), a counter (32-bit), and a 128-bit constant, totaling 64 bytes. The encryption of the plaintext $P$ proceeds in four steps: (1) The internal state $S$ is initialized, and state $K$ is initialized with a copy of $S$. (2) State $K$ is transformed 10 times using the $DoubleRound$ function (two rounds of additions, XORs, and rotations); totaling 20 rounds of transformation. (3) Each element of $K$ is added to the corresponding element of $S$. (4) The first 64 bytes of plaintext $P$ are XORed with $K$. These steps are repeated until the entire plaintext is encrypted. Notably, all control flow decisions, such as loop counts, calls, and returns, are static and determined by the algorithm, with all operations wrapped in loops.

\begin{figure}
    \centering
    \includegraphics[trim=0cm 1cm 0cm 0cm,width=\linewidth]{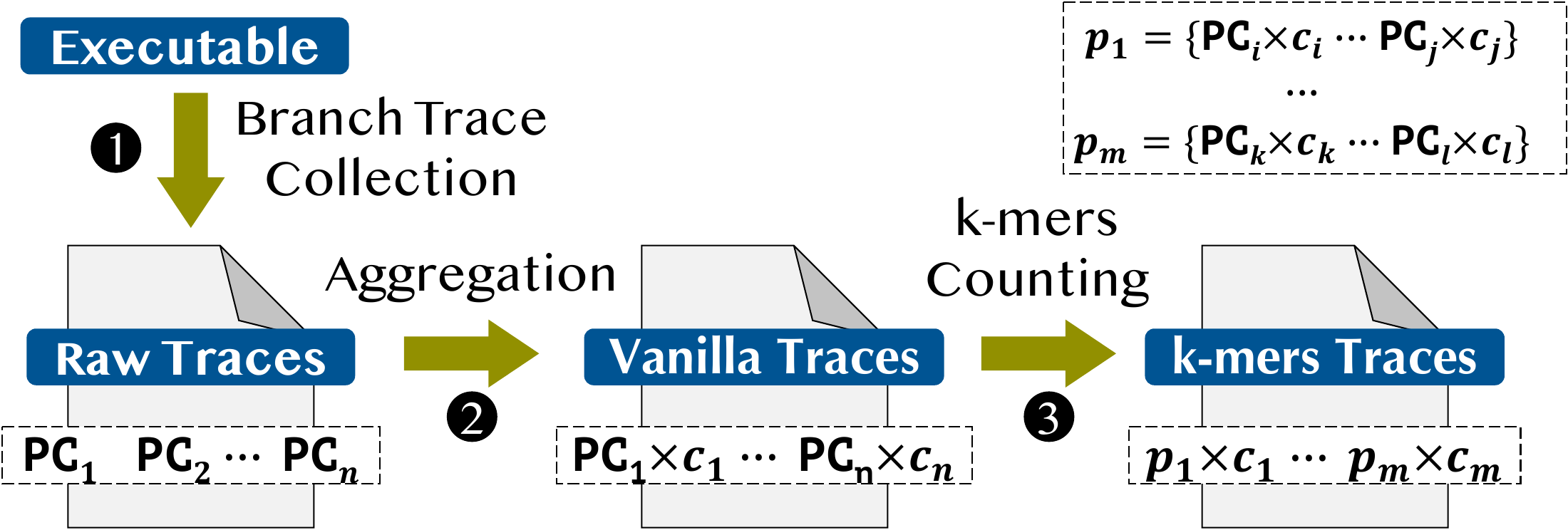}
    \caption{Branch analysis overview in \name{}. Traces are per static branch.}
    \label{fig:branch-analysis-flow}
\end{figure}

\subsection{Detailed Branch Analysis}
\label{sec:branch-analysis}

In this section, we investigate branches in different constant-time cryptographic programs from BearSSL~\cite{bearssl}, OpenSSL~\cite{openssl} and post-quantum crypto (PQC) programs: Kyber~\cite{kyber} and SPHINCS+~\cite{sphincs}. We consider all types of branches: conditional direct branches, unconditional indirect branches, returns, etc. To collect branch traces, we use Intel Pin~\cite{luk2005pin} and record the branch target at each execution of a branch. Figure~\ref{fig:branch-analysis-flow} shows an overview of our branch analysis steps. In the first step (step~\flowstep{1} in Figure~\ref{fig:branch-analysis-flow}), we collect the \raw{} traces for each static branch. In this trace, we capture all the target PCs of a branch (i.e., the branch outcome) in the order they are executed (we log the next PC for not-taken cases). Here is an example of \raw{} trace of a loop branch $BR_0$ with loop count of four:
\begin{center}
    $\targetBlue{PC_1} \concat \targetBlue{PC_1} \concat \targetBlue{PC_1} \concat \targetBlue{PC_1} \concat \targetRed{PC_0}$
\end{center}
where $\targetBlue{PC_1}$ is the taken path of the branch and $\targetRed{PC_0}$ is the next PC after $BR_0$ (i.e., the not-taken path).

\begin{table}[t]
  \centering
  \caption{Branch analysis of cryptographic programs.  
  \tandem{} trace size is the sum of trace size and its pattern set size. 
  }
  \vspace{-0.3cm}
  \resizebox{\linewidth}{!}{%
  \begin{tabular}{cl|rr|rr|rr}
  \toprule
  \multicolumn{2}{c|}{\multirow{3}{*}{\textbf{\Large Program}}} & \multicolumn{2}{c|}{\multirow{2}{*}{\Large \Vanilla{} trace size}} & \multicolumn{2}{c|}{\Large \Tandem{}} & \multicolumn{2}{c}{\Large \Tandem{}}\\
  
  \multicolumn{2}{c|}{} & & & \multicolumn{2}{c|}{\Large trace size} &  \multicolumn{2}{c}{\Large compression rate}\\
  
  \cline{3-8}
  
  \multicolumn{2}{c|}{} & \multicolumn{1}{c|}{\Large Avg} & \multicolumn{1}{c|}{\Large Max} & \multicolumn{1}{c|}{\Large Avg} & \multicolumn{1}{c|}{\Large Max} & \multicolumn{1}{c|}{\Large Avg} & \multicolumn{1}{c}{\Large Max}\\
  
  \midrule
  
  \multirow{7}{*}{\rotatebox[origin=c]{0}{\large \sqbox1{cred}}} & \Large RSA-2048  & \Large 221,619.8 &\Large 24,340,548	& \Large 35.0 & \Large 2,312	&\Large 18,677.0	&\Large 1,622,703.2\\

  & \Large EC\_c25519 & \Large 965,261.6 & \Large 51,538,410 & \Large 7.9 &\Large 134 &\Large 321,607.7&\Large  17,179,470.0\\

  & \Large DES      &\Large 1,483,319.9 &\Large 24,000,000	&\Large 7.9 &\Large 34	&\Large 494,420.7 &\Large 8,000,000.0\\
  
  & \Large AES-128 & \Large 163.2 & \Large 1,530	&\Large 7.6	&\Large 50	&\Large 43.8 &	\Large 510.0 \\

  & \Large ChaCha20 &\Large 175.8 &\Large 752 &\Large 35.5 &\Large 561 &\Large 40.9&\Large 250.7\\

  & \Large Poly1305 & \Large 45.0& \Large 600 &\Large 14.9&\Large 134&\Large 8.7 &\Large 200.0\\

   & \Large SHA-256 &\Large 3,350.5 &\Large 31,736 &\Large 10.7 &\Large 70 &\Large 1,077.6 &\Large 10,578.7\\

  \midrule

  \multirow{3}{*}{\rotatebox[origin=c]{0}{\large \trianbox1{cgreen}}} & \Large curve25519  & \Large 19,375.0 &\Large 128,700	& \Large 4.3 & \Large 18 & \Large 3,479.2 & \Large 17,000.0 \\

  & \Large chacha20  & \Large 24,500.0 &\Large 32,000 & \Large 3.0 & \Large 3	& \Large 8,166.7 & \Large 10,666.7 \\

  & \Large sha256  & \Large 440.0 &\Large 	44,316 & \Large 25.8 & \Large 803  & \Large 42.2 & \Large 5,539.5 \\

  \midrule

  \multirow{5}{*}{\rotatebox[origin=c]{0}{\large \circbox1{cblue}}} & \Large kyber512  & \Large 738,074.1 &\Large 34,620,000	& \Large 5.3 & \Large 24 & \Large 89,705.1	& \Large 2,304,000.0\\

  & \Large kyber768  & \Large 1,177,127.1 &\Large 69,195,000 & \Large 5.6 & \Large 54	& \Large 143,300.9 & \Large 4,608,000.0\\

  & \Large sphincs-shake-128s  & \Large  3,097,903.5 &\Large 90,110,880	& \Large  20.5 & \Large 	348 & \Large 	1,019,536.1 & \Large 30,036,960.0 \\

  & \Large sphincs-haraka-128s  & \Large 1,863,707.0  &\Large 59,244,320 & \Large 24.5 & \Large 544 & \Large 	599,537.1 & \Large 19,748,106.7\\

  & \Large sphincs-sha2-128s  & \Large 298,160.1 &\Large 5,300,746	& \Large  24.6 & \Large 	389 & \Large 	42,948.7 & \Large 1,766,834.0 \\

  \midrule
  
  \multicolumn{2}{c|}{\Large All} &\Large 637,425.5 &\Large 90,110,880	& \Large 19.9 &\Large 2,312 &\Large 163,370.7
  &\Large 30,036,960.0\\
  
  \bottomrule

  \multicolumn{8}{c}{\multirow{2}{*}{\Large \sqbox1{cred} BearSSL \quad
    \trianbox1{cgreen} OpenSSL \quad
    \circbox1{cblue} Post-Quantum Crypto (PQC)}}
  \end{tabular}
  }
  \label{table:branch-analysis-bearssl}
\end{table}

The next step of the analysis builds the \vanilla{} traces that are a more compact format of the \raw{} traces (step~\flowstep{2}). In this format, we aggregate the branch outcomes that are repeating and replace them with the repeated outcome PC and number of repetitions (this is also known as a run-length encoding). Here is the \vanilla{} trace of branch $BR_0$ discussed earlier:
\begin{center}
    $\patternBlue{PC_1}{4} \concat \patternRed{PC_0}{1}$
\end{center}

\Vanilla{} traces are the baseline traces that we use for analysis and compression.
Table~\ref{table:branch-analysis-bearssl} shows 
that the average size of \vanilla{} traces per branch is \vanillaAvgSize{}
in our evaluated programs,
and the maximum size is \vanillaMaxSize{}\footnote{Here, size refers to the number of elements in a trace, not storage size.}. 
Communicating these large traces with the hardware can incur high efficiency overheads. However, we expect these traces to be represented by fewer elements according to \insight{2}; we only need to detect the repeating outcome patterns of each static branch. We aim to devise a generic approach that can detect the repeating patterns in a given \vanilla{} trace. 

\QuestionBox{\question{}: How does one detect the repeating patterns and their frequency in a \vanilla{} trace?}

Detecting repeating, unknown patterns in large traces has been the focus of many domains, like database mining~\cite{agrawal1995mining} and DNA sequencing~\cite{benson1999tandem,marccais2011fast}.
For example, two problems in DNA sequencing that can be useful are
finding tandem repeats~\cite{benson1999tandem} and $k$-mers counting~\cite{marccais2011fast}. 
A tandem repeat in a DNA sequence is two or more contiguous copies of a pattern of nucleotides. Finding tandem repeats has many applications, like individual identification and tracing the root of an outbreak. $k$-mers also refer to a substring of size $k$ of a given DNA sequence. Counting the frequency of $k$-mers is useful in genome assembly and sequence alignment.

\begin{algorithm}[t]
\footnotesize
\DontPrintSemicolon
\SetAlgoLined
\SetKw{KwInit}{Initialization}
\KwIn{DNA sequence $seq$}
\KwOut{\Tandem{} trace $K$ and pattern set $P$}
$unused\_letters \gets alphabet \setminus \mathsf{unique\_letters}(seq)$\\
$current\_len \gets \infty$\\
\While{$\mathsf{len}(seq) < current\_len$}{
    $current\_len = \mathsf{len}(seq)$\\
    $coverage.\mathsf{clear}()$\\
    \For{$k\leftarrow 2$ \KwTo $max\_k$}{
        $freqs \gets \mathsf{count\_kmers}(seq, k)$\\
        \ForEach{$kmer \in freqs$}{
             \If{$fres[kmer] > 1 \land \mathsf{Size}(kmer) \leq max\_k$}{
                $coverage[kmer] \gets (k \times freqs[kmer])/len(seq)$\\
             }
        }
    }
    $most\_frequent\_kmer \gets \mathsf{max}(coverage)$\\
    $frequent\_kmers.\mathsf{insert}(most\_frequent\_kmer)$\\
    $letter \gets unused\_letters.\mathsf{pop}()$\\
    $seq.\mathsf{replace\_and\_merge}(most\_frequent\_kmer, letter)$\\
}
$K \gets seq$\\
$P \gets frequent\_kmers$
\caption{\Tandem{} Branch Compression}
\label{algo:kmers-compression}
\end{algorithm}

\subsubsection{$k$-mers Counting and Traces}
\label{sec:kmers-counting}

In this work, we deploy the $k$-mers counting technique for pattern repeat detection (step~\flowstep{3}). The reason for this choice is that our experiments with the state-of-the-art tools show that $k$-mers counting tools are much faster to analyze large traces (up to millions) compared to others (e.g., the TRF tool~\cite{benson1999tandem} for tandem repeat finding) and also they are more configurable. We use scikit-bio Python library~\cite{scikit} in our analysis which allows us to define a custom alphabet for DNA sequences, while most other tools only consider four letters A, C, G, T; some branches can have more than four outcomes (e.g., a return can jump to more than four callsites). Additionally, $k$-mers counting tools allow configuring the algorithm parameters which is useful to enforce starting with smaller and more frequent patterns and then continuing to larger patterns if necessary. This is beneficial to reduce the storage requirement as much as possible. However, note that we use the $k$-mers counting just as a demonstration and our compression results do not depend on a specific tool.

Before the $k$-mers counting step of our analysis, we transform \vanilla{} traces to their equivalent DNA sequences. For example, \vanilla{} trace of branch $BR_1$ of this form:
\begin{center}
    $\patternBlue{PC_0}{2} \concat \patternRed{PC_1}{5} \concat \patternBlue{PC_0}{2} \concat \patternRed{PC_1}{5} \concat \patternGreen{PC_2}{3}$
\end{center}
is transformed to this DNA sequence: $\ntBlue{A}\ntRed{C}\ntBlue{A}\ntRed{C}\ntGreen{G}$.

Algorithm~\ref{algo:kmers-compression} shows a simplified version of the technique that we use to build \tandem{} traces.
The input of the algorithm is the equivalent DNA sequence of a \vanilla{} trace.
The core of the algorithm is the $count\_kmers$ procedure (line 7) that takes $k$ and DNA sequence $seq$ as input and builds a frequency map of all the existing $k$-mers and their frequency. Algorithm~\ref{algo:kmers-compression} continues compressing the sequence with the most frequent pattern (i.e., has the highest coverage in the sequence, lines 14-17) 
until the length of the compressed sequence stops reducing (line 3).
Finally, the output of the algorithm is the compressed DNA sequence $K$ and the set of detected patterns $P$ (lines 19-20).

As the final step, we re-transform the DNA \tandem{} patterns back to the PC traces. We refer to the result as the \tandem{} representation; \tandem{} representation consists of the \tandem{} trace $K$ and its transformed pattern set $P$. 
For example, here is the \tandem{} trace of branch $BR_1$ that we discussed earlier:
\begin{center}
    $\patternBlue{p_0}{2} \concat \patternRed{p_1}{1}$
\end{center}
where the pattern set is: 
\begin{center}
    $P =\{\ntBlue{p_0}: \pattern{PC_0}{2} \concat \pattern{PC_1}{5}, \ntRed{p_1}: \pattern{PC_2}{3}\}$
\end{center}

Table~\ref{table:branch-analysis-bearssl} shows the average and maximum size of \tandem{} representation (sum of trace $K$ size and pattern set $P$ size). The average \tandem{} size per static branch is \tandemAvgSize{} and the maximum size is \tandemMaxSize{}. Compared to \vanilla{} trace sizes, our compression leads to an average compression rate of \avgCompressionRate{} and a maximum rate of \maxCompressionRate{}.
Note, that the results presented in Table~\ref{table:branch-analysis-bearssl} exclude the branches that always have a single target (i.e., their \vanilla{} trace size is already 1).

\textbf{Example: \toyAES{}}.
Figure~\ref{fig:toy-example} illustrates the \name{} branch analysis for a toy example that encrypts data in three encryption rounds
with key and plaintext length of two. 
In the first step, \raw{} traces are collected per static branch (step~\flowstep{1}). For instance, \texttt{BR6} is a loop branch with a loop count of two: it executes \texttt{BR7} twice and then executes the fall-through path, \texttt{PC7}.
In the next step, \vanilla{} traces are generated (step~\flowstep{2}). After transforming \vanilla{} traces into equivalent DNA sequences (step~\flowstep{3}), we perform our $k$-mers branch compression technique and generate the \tandem{} traces and pattern sets (step~\flowstep{4}).

\begin{algorithm}[t!]
\footnotesize
\DontPrintSemicolon
\SetAlgoLined
\SetKwFunction{proc}{generate\_kmers\_traces}
\SetKw{KwInit}{Initialization}
\KwIn{Input binary $bin\_in$, $inp1$, $inp2$}
\KwOut{Updated binary $bin\_out$ with traces and hint information}
$traces.\mathsf{clear}()$\\
$unique\_branches \gets \mathsf{detect\_static\_branches}(bin\_in)$ \flowstepletter{A}\\
\ForEach{$branch \in unique\_branches$}{
    $[K_1, P_1] \gets \proc(branch, bin\_in, inp1)$\\
    $[K_2, P_2] \gets \proc(branch, bin\_in, inp2)$\\
    $is\_input\_dependent \gets \mathsf{diff}(K_1, K_2)$\\
    \If{$\neg is\_input\_dependent$}{
        $traces.\mathsf{insert}([branch, K_1, P_1])$\\
    }
}
$bin\_out \gets \mathsf{embed\_information}(bin\_in, traces)$\\
\SetKwProg{myproc}{Procedure}{}{}
\myproc{\proc{$branch$, $bin$, $inp$}}{
  $R \gets \mathsf{collect\_raw\_traces}(branch, bin, inp)$  \flowstepletter{B}\\
  $V \gets \mathsf{transform\_to\_vanilla\_traces}(R)$ \flowstepletter{C}\\
  $DNA\_seq \gets \mathsf{transform\_to\_DNA}(V)$ \flowstepletter{D}\\
  $[K, P] \gets \mathsf{kmers\_compression}(DNA\_seq)$ \flowstepletter{E}\\
  \KwRet $[K, P]$\\
}
\caption{Trace Generation Procedure}
\label{algo:trace-generation}
\end{algorithm}

\begin{figure*}[t]
    \centering
    \includegraphics[trim=0cm 1.8cm 0cm 0cm,width=0.95\linewidth]{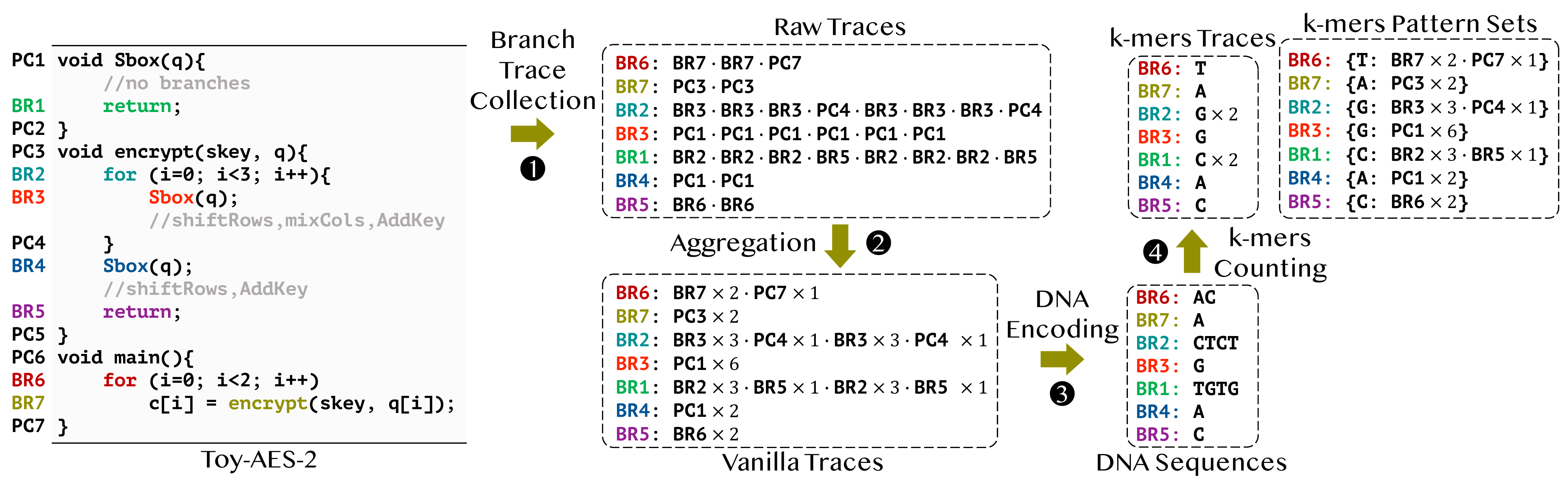}
    \caption{\name{} branch analysis workflow example. Note, that the branches are analyzed separately and traces are generated per static branch; DNA sequences of branches are independent from each other.}
    \label{fig:toy-example}
\end{figure*}

\subsection{Automatic Trace Generation Procedure}
\label{sec:trace-generation}

We provide an automatic procedure to generate branch traces for a given binary of a constant-time cryptographic application. Algorithm~\ref{algo:trace-generation} shows the steps of this procedure (steps~\flowstepletter{A}-\flowstepletter{E}).

Step~\flowstepletter{A} identifies all static branches that appear during the execution (line 2) and stores them in the $unique\_branches$ set. Steps~\flowstepletter{B}-\flowstepletter{E} generate \tandem{} traces for each branch, as we explained in \secref{sec:branch-analysis}.

Note, that in lines 4 and 5, we generate \tandem{} traces twice with two different inputs to detect branches that their traces change depending on the input. For example, \textit{stream loops} in stream ciphers, like ChaCha20, accept input plaintexts of an arbitrary length.  
The program processes each block of the plaintext in a loop (i.e., the stream loop). The \vanilla{} trace of the stream loop is in the form of $\patternGreen{PC_1}{n} \concat \patternRed{PC_0}{1}$, where $n$ is the length of the input\footnote{In addition, some branches in post-quantum crypto primitives have random traces that change in different runs, e.g., two branches in rejection sampling of Kyber.}.
However, all the other branches are wrapped inside this loop and repeat. Hence, they have valid \tandem{} traces.
For branches that their trace depends on the input, we stall the fetch until the branch resolves\footnote{In general, if traces are not available for a crypto branch,
we redirect fetch only if the branch direction is resolved.}; this incurs a negligible penalty since they are not frequent and quickly resolve.

Finally, once all branches are analyzed, the input binary is instrumented with the \tandem{} traces and their \textit{hint information} to facilitate their access during execution (line 11, see \secref{sec:trace-communication} for the details of trace representations and their communication with the hardware). We discuss the runtime overhead of the one-time trace generation procedure for all applications in \secref{sec:trace-gen-results}.

\section{Design of \name{}}
\label{sec:design}

To efficiently implement \name{} in hardware, we need to (1) communicate the branch traces prepared by our analysis with the hardware on demand, and (2) design a specialized unit, called \BtUnit{} (\BTU{}), in the fetch stage to determine the branch directions based on the sequential branch traces. The \BTU{} is designed similarly to Trace Caches~\cite{rotenberg1996trace,rotenberg1997trace} and Schedule Caches~\cite{padmanabha2017mirage} in prior work, with two key differences: (1) traces are determined before execution in \name{} and no dynamic trace selection methodology is used. (2) In case of a trace miss in the \BTU{}, the frontend stalls until the trace becomes available, while prior works would switch to a normal, speculative fetch procedure. 

In \secref{sec:design-overview}, we present an overview of \name{} design, and in \secref{sec:trace-communication} and \secref{sec:uarch}, we provide the details for \name{} implementation.

\subsection{Overview}
\label{sec:design-overview}

Figure~\ref{fig:uarch} shows an overview of the \name{} microarchitecture. When a branch is fetched, two possible scenarios occur depending on whether the branch
belongs to a cryptographic program, or it is a non-crypto branch.
In the former scenario, the fetch unit queries the \BTU{} to determine the next PC (step~\flowstep{1}), and in the latter scenario, the BPU predicts the next PC (step~\flowstep{2}). The \PatternTable{} and the \TraceCache{} are the two sub-components of the \BTU{} that (1) determine the next PC for each branch and (2) keep track of the progress within the trace. In cases that a trace fits in one entry of the \TraceCache{}, it will rotate to keep replaying the trace. However, if the trace does not fit in one entry then the head element of the entry is removed when the branch commits, and the entire entry shifts and prefetches the upcoming parts of the trace at the back of the entry (step~\flowstep{3}). Finally, when a branch misses in the \BTU{}, one of the entries is evicted and a checkpoint of its progress is taken in the \CkptTable{}. This checkpoint allows to resume the execution of the evicted branch when it reappears in the future. In \secref{sec:uarch}, we discuss the details of our microarchitecture. 

\begin{figure}
    \centering
    \includegraphics[trim=0cm 2cm 0cm 0cm,width=0.9\linewidth]{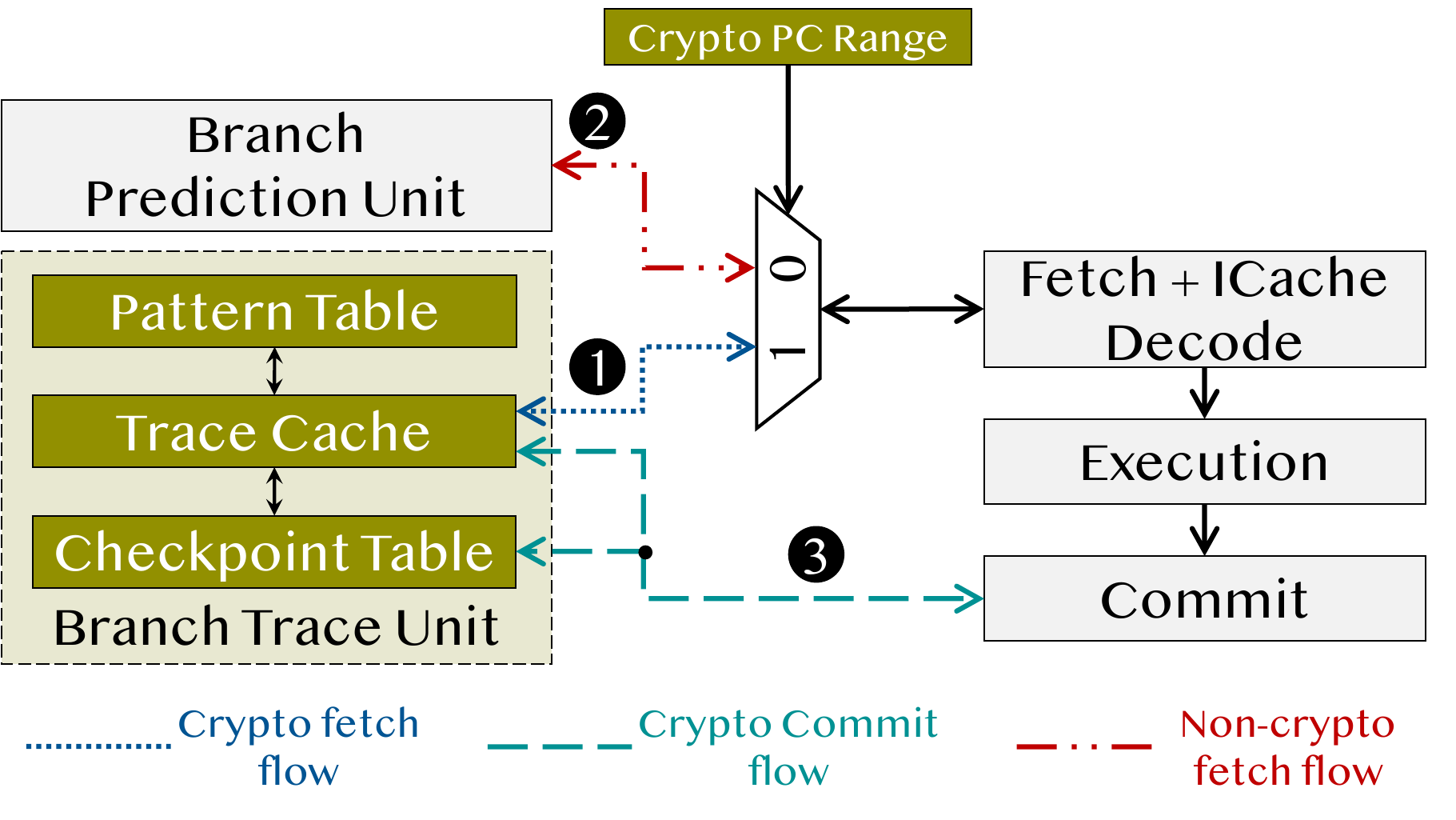}
    \caption{Overview of \name{} microarchitecture. Crypto branches do not access or update the BPU.}
    \label{fig:uarch}
\end{figure}

\subsection{Trace Representation and Communication}
\label{sec:trace-communication}

We use the output of Algorithm~\ref{algo:kmers-compression} to prepare the branch traces. 
Traces consist of two parts per static branch:
(1) the pattern set built from the \tandem{} patterns $P$,
which stores all the possible branch outcomes, and (2) the branch trace built from the \tandem{} trace $K$.
Figure~\ref{fig:elements}(a) shows the structure of each element in the pattern set. Each pattern element has a 12-bit target offset (the signed difference between the branch PC and the target PC) and the number of its repetitions (8-bit). 
In cases where the number of repetitions exceeds 8 bits, the element is duplicated in a way that the sum of the two elements is equal to the original number: 
\begin{center}
    $\patternRed{\delta(BR_0)}{300} \rightarrow \patternGreen{\delta(BR_0)}{255} \concat \patternGreen{\delta(BR_0)}{45}$
\end{center}
We use a compact form to store the patterns in cases where patterns overlap. For example, if two patterns in a trace are $\ntBlue{A}\ntRed{C}\ntGreen{T}$ and $\ntRed{C}\ntGreen{T}\ntBlue{A}$, then the output pattern set is $\ntBlue{A}\ntRed{C}\ntGreen{T}\ntBlue{A}$.

Figure~\ref{fig:elements}(b) shows the structure of each element in the branch trace. The first two fields, \textit{pattern index} and \textit{pattern size}, specify the corresponding pattern from the pattern set. For example, if the corresponding pattern of a trace element is $\ntRed{C}\ntGreen{T}$ and the entire pattern set is $\ntBlue{A}\ntRed{C}\ntGreen{T}\ntBlue{A}$, then the \textit{pattern index} is $1$ (indices start from $0$) and the \textit{pattern size} is $2$. \textit{Pattern counter} is equal to the sum of the repetitions of the corresponding pattern elements and the \textit{trace counter} specifies the total number of times that the pattern needs to be repeated before advancing to the next trace element.

A special End of Trace marker is used to denote the end of each trace. This allows the processor to repeat the trace whenever it reaches the end of the trace.
We store traces in data pages and embed hints for each static branch:

\textbf{(1) \textit{Single-target} mark}. A significant portion of branches always jump to a single target (e.g., \texttt{"call sbox <pc>"}), and we mark such branches as \textit{single-target} and do not need to store and communicate traces for them (e.g., $79\%$ of static branches in RSA are \textit{single-target}); we only need to embed its single target within the hint information (i.e., a \textit{PC offset} pointing to the branch's target). This implementation ensures that no \BTU{} resources are used for \textit{single-target} branches and no trace miss would occur as well.

\textbf{(2) Traces Virtual Address offset ($\Delta$)}. If the branch is \textit{multiple-target}, then $\Delta$ points to the data page address holding branch traces.

\textbf{(3) \textit{Short-trace} mark}. We mark the branches when their traces are smaller than 16 (i.e., they fit in one entry of the \BTU{}). This will allow us to avoid additional accesses to bring traces to the \BTU{} and only repeat the trace once loaded.

\begin{figure}
    \centering
    \includegraphics[trim=0cm 2cm 0cm 0cm,width=0.85\linewidth]{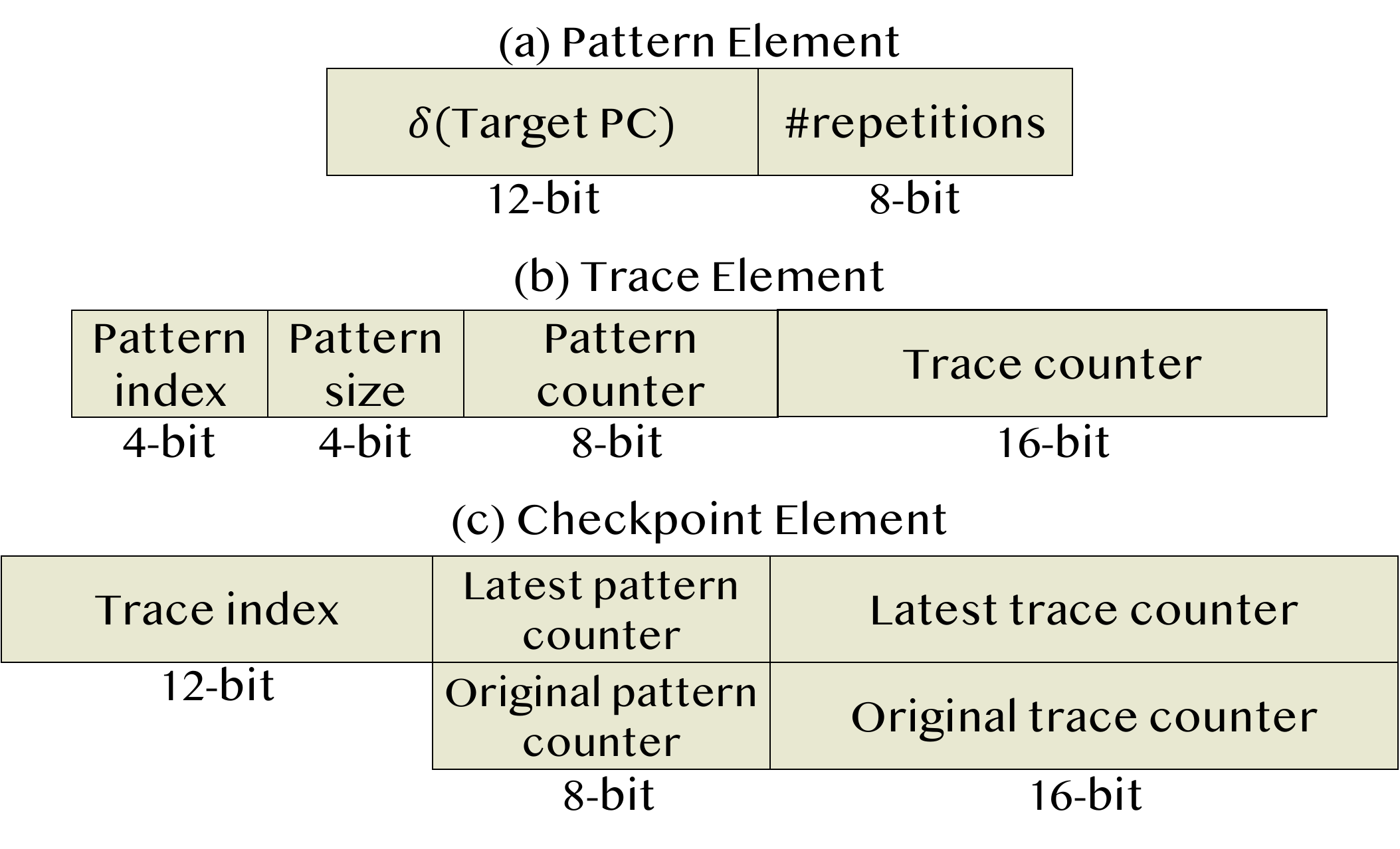}
    \caption{Elements in the \BtUnit{} (\BTU{}). Each entry of the \PatternTable{}, \TraceCache{}, and \CkptTable{}, consisting of 16 elements and corresponds to a static branch.}
    \label{fig:elements}
\end{figure}

\textbf{Embedding hint information}. 
A general approach to inform the hardware about the hints is to insert a special hint instruction before each branch. Hint instructions are only decoded and do not use the ALUs; prior work has used hint instructions for x86~\cite{khan2022whisper} and RISC-V ISAs~\cite{hajiabadi2021noreba}. 
An alternative solution is to re-purpose some of the previously-ignored prefix bytes in x86, in the same way that XRELEASE~\cite{intelLockElision} was implemented, to embed the hint information for each branch (similar to prior work~\cite{invarspec}).
\revision{Fourteen bits can be sufficient per static branch to embed single-target mark (1 bit), address offset (12 bits), and short-trace mark (1 bit).}
We opt to use the latter solution in this work because hint instructions still consume critical frontend resources, even though not executed. Moreover, inserting hint instructions might not provide backward compatibility with older processors.

\textbf{Crypto PC range}. 
We also use a new status register that specify PC ranges for crypto code, called the \textit{Crypto PC Ranges} register, to avoid the penalties of waiting until hint information is decoded.
Note, that crypto branches that hit in the BTU do not require hint information; only rare cases where traces miss in the BTU require decoding hint information to load traces.

\subsection{Details of the Microarchitecture}
\label{sec:uarch}

The \BTU{} consists of three main components:
\begin{itemize}[leftmargin=0.5cm]
    \item \PatternTable{} (\PT{}) holds the pattern sets of branches and each entry consists of 16 pattern elements (see Figure~\ref{fig:elements}(a));
    \item \TraceCache{} (\TC{}) holds the branch traces and each entry consists of 16 trace elements (see Figure~\ref{fig:elements}(b));
    \item \CkptTable{} (\CT{}) always holds the latest valid position of the branch trace, i.e., the committed progress of the trace. Each entry is only one checkpoint element (see Figure~\ref{fig:elements}(c)). \CT{} is stored in data pages which keeps the checkpoints for all branches to handle the \BTU{} evictions and interrupts.
\end{itemize}
In addition, the \CT{} keeps the original counts of the first element of the \TC{} (head of the trace); this helps the \BTU{} to insert a refreshed version of the element at the back of the \TC{} entry for repetition
(see the \textbf{commit flow} for the details of the \CT{} updates).

All three tables are direct-mapped tables, indexed with the branch PC, 
and they are fully inclusive of each other. 
The \BTU{} uses an LRU replacement policy to evict an entry.

\textbf{Crypto fetch flow}. Once a crypto branch is fetched, the fetch unit queries the \BTU{} to determine the next PC (step \flowstep{1} in Figure~\ref{fig:uarch}). If the branch is marked as \textit{single-target}, then the next PC is already known by the hint information and there is no need for a \BTU{} lookup.
For \textit{multiple-target} branches, \BTU{} looks up the first element of the \TC{} to find the appropriate pattern element in the \PT{} which provides the next PC. 
Upon each \BTU{} lookup, the \textit{pattern counter} of the first element in the \TC{} is decremented. 
Whenever the \textit{pattern counter} reaches zero, we advance to the next pattern element by decrementing the \textit{trace counter} and updating the \textit{pattern counter} based on the new pattern element.
As we will explain in the \textbf{crypto commit flow}, the first element of the trace is removed only when the enforced branch direction is committed. Hence, there is a possibility that the \textit{trace counter} of the first element is zero (i.e., we need to advance to the next element) but the branch is not committed yet. In this case, the \BTU{} needs to lookup the next element in the \TC{} entry.
In the worst case that all 16 elements of the \TC{} are looked up and not committed (i.e., \textit{trace counter} is zero in all of them), then the \BTU{} waits until the first element is removed. We did not encounter this scenario in our simulations since crypto branches resolve before all elements are looked up.

\textbf{Non-crypto fetch flow}. For non-crypto branches, 
we use the BPU to determine the next PC (step \flowstep{2}). However, to prevent speculative fetch redirections to the crypto code we perform an integrity check to prevent fetch redirection if the predicted target is part of the crypto code (using the Crypto PC Ranges register). In this case, we wait for the branch to resolve before taking the branch.

\textbf{Crypto commit flow}. Once a crypto branch commits (step \flowstep{3}), if the \textit{trace counter} of the first element in the corresponding \TC{} entry is zero, then the first element is removed and all the other elements are shifted. 
To fill the last element of the \TC{} entry, two cases can happen:
\begin{enumerate}[leftmargin=0.6cm]
    \item if the branch is marked as \textit{short-trace}, a refreshed version of the removed element is inserted at the back of the entry;
    \item if the trace is larger than the \TC{} entry, we prefetch the upcoming elements and insert at the back of the entry. If the last element is an End of Trace marker, we restart from the beginning of the trace.
\end{enumerate}

Additionally, when a crypto branch commits, the latest \textit{pattern counter} and \textit{trace counter} are checkpointed in the \CT{}. This allows the processor to resume the execution when it is interrupted (e.g., in context switches). \textit{Trace index} in the checkpoint element (see Figure~\ref{fig:elements}(c)) points to the latest trace element that the execution needs to resume from.

\textbf{Trace evictions in the \BTU{}}. Once a trace is evicted from the \TC{}, the corresponding entries in the \PT{} and \CT{} are evicted as well. 
Before evicting the \CT{} entry, the checkpoint element is updated with the latest counters and \textit{trace index} and is stored in the memory. This allows the CPU to resume the execution when the evicted branch reappears.

\textbf{Recovery for ROB Squashes}. While \name{} guarantees no branch mispredictions for crypto branches, ROB squashes can still occur due to other reasons (e.g., non-crypto mispredictions or interrupts), and \name{} needs to recover in cases where the crypto branches are squashed. Whenever a crypto branch is squashed, we undo the actions of the \textbf{crypto fetch flow}; the \textit{pattern counter} and \textit{trace counter} of the first elements are incremented according to the checkpointed counters in the \CT{}.

\section{Security Analysis}
\label{sec:security-analysis}

In this section, we aim to provide a detailed analysis of \name{}'s security and precisely identify its protection scope.
We discuss practical deployment and additional considerations to comprehensively block Spectre-type leaks.
We explain the preliminaries in \secref{sec:analysis-preliminaries} before discussing \name{}'s security in \secref{sec:cassandra-analysis}.

\subsection{Preliminaries}
\label{sec:analysis-preliminaries}

Below, we provide the required background and the terminology used for the security analysis of \name{}.

\subsubsection{Security Policies}

Traditionally, developers relied on sequential program execution to enforce security policies in two main application domains: (1) high-assurance cryptography and (2) isolation of untrusted code~\cite{cauligi2022sok}. We briefly discuss the security policies required in these domains.

\textbf{Security policy for high-assurance cryptography}. 
As we discussed in \secref{sec:ct-definition}, cryptographic programs deploy \textit{constant-time programming} to ensure that attacker-visible observations (also referred to as \textit{leakage model}) of the program do not depend on secrets.
The leakage model of constant-time programs captures the control flow, accessed memory addresses, and operands of variable-time instructions (notated as $\CustomInterf{\cdot}{ct}{}$ leakage model\footnote{We borrow the notation and terminology from prior work~\cite{cauligi2022sok,guarnieri2021hardware}.}).

\textbf{Security policy for software isolation}. For software isolation (also referred to as \textit{sandboxing}), 
a host application needs to ensure that untrusted guest code cannot access the host's memory outside an authorized range (e.g., eBPF in Linux kernel).
A common leakage model assumes an adversary observing all architectural computation and accesses, including register file contents ($\CustomInterf{\cdot}{arch}{}$ leakage model). Thus, this security policy requires preventing out-of-bounds memory accesses during execution.

\begin{table*}[t]
  \centering
  \caption{Security analysis of all possible control flow scenarios in \name{} (see Figure~\ref{fig:cassandra-scenarios}).}
  \vspace{-0.3cm}
  \resizebox{\linewidth}{!}{%
  \begin{tabular}{c|l|l|l}
  \toprule
  
   Scenario & \multicolumn{1}{c|}{Transition~$\ddagger$} & \multicolumn{1}{c|}{Execution flow} & \multicolumn{1}{c}{\name{} mechanism}\\
   
  \midrule
  
  \flowstep{1} & 
  \begin{tikzcd}[cramped,sep=small,ampersand replacement=\&]\Mathcryptocode{\footnotesize BR1}\arrow[r,SEQ]\&\Mathcryptocode{\footnotesize R1}\end{tikzcd} &
  $\inlineasm{branch1} \concat \inlineasmOneOperand{leak}{r1}$ & Encforcing sequential flow via looking up pre-computed sequential branch traces (\BTU{})\\ 

  \hline 
  
  \flowstep{2} & 
  \begin{tikzcd}[cramped,sep=small,ampersand replacement=\&]\Mathcryptocode{\footnotesize BR1}\arrow[r,SEQ]\&\Mathcryptocode{\footnotesize M1}\end{tikzcd} &
  $\inlineasm{branch1} \concat \inlineasmTwoOperand{load}{r2}{addr_A} \concat \inlineasmOneOperand{leak}{r2}$ & Encforcing sequential flow via looking up pre-computed sequential branch traces (\BTU{})\\ 

  \hline

  \flowstep{3} & 
  \begin{tikzcd}[cramped,sep=small,ampersand replacement=\&]\Mathcryptocode{\footnotesize BR1}\arrow[r,SEQ]\&\Mathnoncryptocode{\footnotesize R2}\end{tikzcd} &
  $\inlineasm{branch1} \concat \inlineasmOneOperand{leak}{r4}$ & Encforcing sequential flow via looking up pre-computed sequential branch traces (\BTU{})\\ 

  \hline

  \flowstep{4} & 
  \begin{tikzcd}[cramped,sep=small,ampersand replacement=\&]\Mathcryptocode{\footnotesize BR1}\arrow[r,SEQ]\&\Mathnoncryptocode{\footnotesize M2}\end{tikzcd} & 
  $\inlineasm{branch1} \concat \inlineasmTwoOperand{load}{r3}{addr_B} \concat \inlineasmOneOperand{leak}{r3}$ & Encforcing sequential flow via looking up pre-computed sequential branch traces (\BTU{})\\

  \hline

  \flowstep{5} & 
  \begin{tikzcd}[cramped,sep=small,ampersand replacement=\&]\Mathnoncryptocode{\footnotesize BR2}\arrow[r,SEQ]\&\Mathcryptocode{\footnotesize M1}\end{tikzcd} &
  $\inlineasm{branch2} \concat \inlineasmTwoOperand{load}{r2}{addr_A} \concat \inlineasmOneOperand{leak}{r2}$ & Encforcing sequential flow via integrity checks upon non-crypto branches\\

  \hline

  \flowstep{6} & 
  \begin{tikzcd}[cramped,sep=small,ampersand replacement=\&]\Mathnoncryptocode{\footnotesize BR2}\arrow[r,SEQ]\&\Mathcryptocode{\footnotesize R1}\end{tikzcd} &
  $\inlineasm{branch2} \concat \inlineasmOneOperand{leak}{r1}$ & Sequential flow via integrity checks; however, $r1$ is already declassified by the crypto code\\

  \hline

  \flowstep{7} & 
  \begin{tikzcd}[cramped,sep=small,ampersand replacement=\&]\Mathnoncryptocode{\footnotesize BR2}\arrow[r,dotted,NOTCARE]\&\Mathnoncryptocode{\footnotesize R2}\end{tikzcd} &
  $\inlineasm{branch2} \concat \inlineasmOneOperand{leak}{r4}$ & Speculative flow; this is allowed in non-crypto code and {\scriptsize $\CustomInterf{\cdot}{arch}{seq}$} contract\\

  \hline

  \flowstep{8} & 
  \begin{tikzcd}[cramped,sep=small,ampersand replacement=\&]\Mathnoncryptocode{\footnotesize BR2}\arrow[r,dashed,SPEC]\&\Mathnoncryptocode{\footnotesize M2}\end{tikzcd} &
  $\inlineasm{branch2} \concat \inlineasmTwoOperand{load}{r3}{addr_B} \concat \inlineasmOneOperand{leak}{r3}$ & Speculative flow; out of scope (i.e., software isolation)\\
  
  \bottomrule

  \multicolumn{4}{c}{\multirow{2}{*}{$\ddagger$ \cryptocode{\footnotesize A}: crypto code, \noncryptocode{\footnotesize B}: non-crypto code.
    \quad \begin{tikzcd}[cramped,sep=small,ampersand replacement=\&]A\arrow[r,SEQ]\&B\end{tikzcd}: sequential flow,
    \begin{tikzcd}[cramped,sep=small,ampersand replacement=\&]A\arrow[r,dashed,SPEC]\&B\end{tikzcd}: speculative flow, 
    \begin{tikzcd}[cramped,sep=small,ampersand replacement=\&]A\arrow[r,dotted,NOTCARE]\&B\end{tikzcd}: don't care flow.}}
  \end{tabular}
  }
  \label{table:security-analysis}
  \vspace{-0.3cm}
\end{table*}

\subsubsection{Spectre Vulnerabilities}

Before the advent of Spectre, software-level tools assumed a sequential (architectural) execution model (notated as $\CustomInterf{\cdot}{}{seq}$ execution model) to enforce either constant-time or isolation policies. However, modern CPUs follow a speculative execution model (notated as $\CustomInterf{\cdot}{}{spec}$ execution model) that can transiently execute instructions from unintended paths of the program.
Figure~\ref{fig:spectre-leaks} shows transient leaks of the programs that are secure with respect to a sequential execution model.

Sequential execution of both examples in Figure~\ref{fig:spectre-leaks} are constant-time (i.e., they are secure according to a $\CtSeqInterf{\cdot}$ contract\footnote{The combination of a leakage model $\interfStyle{\beta}$ and an execution model $\interfStyle{\alpha}$ expresses a $\CustomInterf{\cdot}{\beta}{\alpha}$ contract that governs the attacker-visible observations of a program during execution.}).
Nevertheless, the code in Figure~\ref{fig:spectre-leaks}(a) presents a transient leak of register $r1$ 
(through a register leak gadget \cryptocode{R2}) which contains secret data that never leaks during sequential execution. Similarly, Figure~\ref{fig:spectre-leaks}(b) transiently accesses a secret memory region and leaks the value loaded in register $r2$ (through memory leak gadget \cryptocode{M1}). 

Moreover, Figure~\ref{fig:spectre-leaks}(b) demonstrates that the speculative execution model of CPUs can violate software isolation policies as well through transient execution of memory leak gadgets; it leaks a memory location through \cryptocode{M1} gadget which is not accessed during sequential execution (i.e., it is secure against a $\ArchSeqInterf{\cdot}$ contract). 

After the advent of Spectre, many software tools were developed to capture speculative execution semantics and extend constant-time policies to potentially transient program paths~\cite{cauligi2020constant,guarnieri2020spectector,daniel2021hunting,barthe2021high,vassena2021blade,mosier2024serberus,pescosta2021bounded,wu2019abstract,oleksenko2020specfuzz,wang2020kleespectre,guo2020specusym,qi2021spectaint,wang2019oo7}. However, we argue that software should only consider sequential (non-speculative) semantics, with hardware ensuring no additional leaks beyond those intended by the software. The rationale for this argument is that reasoning about transient execution at the software level requires complete microarchitectural knowledge, which is not public for widely deployed CPUs, making such protections vulnerable to unknown speculation mechanisms. Moreover, software protections lack the performance flexibility of hardware mechanisms, relying on costly measures like fences to block potential transient leaks. 
Hence, \textit{we present \name{} as a CPU enhancement that allows developers to write high-assurance cryptographic code using the standard set of constant-time programming rules, without requiring additional programming effort and software-level considerations regarding speculation}.

\begin{figure}
    \centering
    \includegraphics[trim=0 1.2cm 0 0,width=\linewidth]{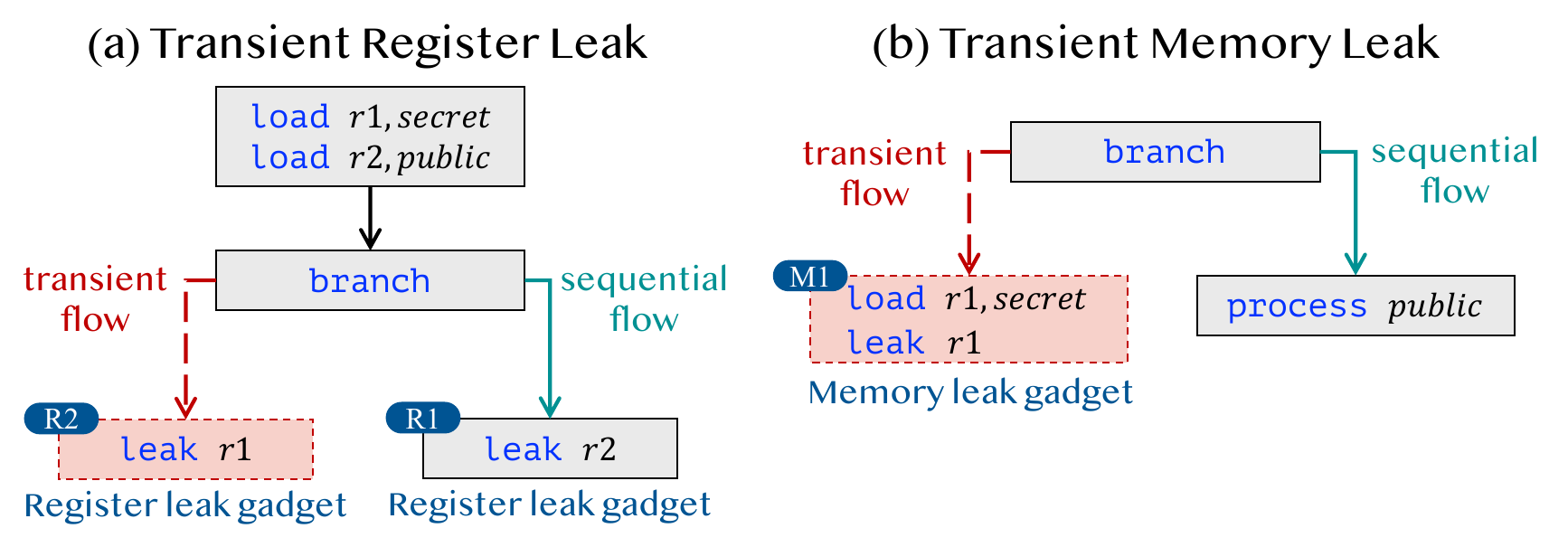}
    \caption{
    (a) Transient register leak, (b) transient memory leak. Both cases are constant-time during sequential execution, but violated during transient execution. 
    }
    \label{fig:spectre-leaks}
    \vspace{-0.55cm}
\end{figure}

\subsection{Security Analysis of \name{}}
\label{sec:cassandra-analysis}

The security goal of \name{} is to guarantee that (1) all executed paths after crypto branches are on the sequential path, and (2) all crypto leak gadgets execute on the sequential path. In other words, (1) all outgoing edges from \cryptocode{BR1} in Figure~\ref{fig:cassandra-scenarios} must follow the sequential (non-speculative) flow of the program, and (2) all incoming edges to \cryptocode{M1} and \cryptocode{R1} are on the sequential path.

Figure~\ref{fig:cassandra-scenarios} illustrates all possible scenarios that capture the execution of both crypto code (secure against the $\CustomInterf{\cdot}{ct}{seq}$ contract) and non-crypto code (secure against the $\CustomInterf{\cdot}{arch}{seq}$ contract) and their interaction in a \name{}-enabled processor. 
We will discuss these scenarios to explain \name{}'s security and protection scope (crypto gadgets are highlighted as \cryptocode{blue} and non-crypto gadgets are highlighted as \noncryptocode{orange}). 
Our security analysis is summarized in Table~\ref{table:security-analysis}, indicating \name{} mechanisms for each scenario.

\textbf{Scenarios \flowstep{1} and \flowstep{2}}
(\begin{tikzcd}[cramped,sep=small]\Mathcryptocode{BR1}\arrow[r,SEQ]&\Mathcryptocode{M1}~\Mathcryptocode{R1}\end{tikzcd}):
These are the scenarios where crypto leak gadgets execute after a crypto branch (i.e., leakage of register $r1$ and memory $addr_A$ after $\inlineasm{branch1}$).
\name{} guarantees sequential execution for these scenarios by looking up the sequential control flow trace of $\inlineasm{branch1}$.

\textbf{Scenario~\flowstep{3}} 
(\begin{tikzcd}[cramped,sep=small]\Mathcryptocode{BR1}\arrow[r,SEQ]&\Mathnoncryptocode{R2}\end{tikzcd}):
This is the scenario where a non-crypto register leak gadget executes after a crypto branch (i.e., leakage of register $r4$ after \inlineasm{branch1}). 
\name{} guarantees sequential execution for this scenario by looking up the sequential control flow trace of $\inlineasm{branch1}$.
Note, that the leakage of $r4$ is intentional, as its content should already have been declassified (i.e., made public) before transitioning to unsafe, non-crypto code.

\textbf{Scenario~\flowstep{4}} 
(\begin{tikzcd}[cramped,sep=small]\Mathcryptocode{BR1}\arrow[r,SEQ]&\Mathnoncryptocode{M2}\end{tikzcd}):
This is the scenario where a non-crypto memory leak gadget executes after a crypto branch (i.e., leakage of memory $addr_B$ after \inlineasm{branch1}). \name{} guarantees sequential execution for this scenario by looking up the sequential trace of $\inlineasm{branch1}$.

\textbf{Scenario~\flowstep{5}} 
(\begin{tikzcd}[cramped,sep=small]\Mathnoncryptocode{BR2}\arrow[r,SEQ]&\Mathcryptocode{M1}\end{tikzcd}):
This is the scenario where a crypto memory leak gadget executes after a non-crypto branch (i.e., leakage of memory $addr_A$ after \inlineasm{branch2}).
While a \name{}-enabled processor predicts the outcome of non-crypto branches, we perform an integrity check to not speculatively redirect fetch to crypto code (discussed in \secref{sec:uarch} in the \textbf{non-crypto fetch flow}). 
\name{} guarantees sequential execution for this scenario by stalling fetch until the non-crypto branch resolves.

\textbf{Scenario~\flowstep{6}} 
(\begin{tikzcd}[cramped,sep=small]\Mathnoncryptocode{BR2}\arrow[r,SEQ]&\Mathcryptocode{R1}\end{tikzcd}):
This is the scenario where a crypto register leak gadget executes after a non-crypto branch (i.e., leakage of register $r1$ after \inlineasm{branch2}).
\name{} guarantees sequential execution for this scenario similar to scenario~\flowstep{5} through integrity checks.
Note, that the content of register $r1$ is public since crypto programs declassify registers before transitioning to unsafe, non-crypto code. In other words, transient execution of this scenario would not have leaked any secrets as well.

\textbf{Scenario~\flowstep{7}} 
(\begin{tikzcd}[cramped,sep=small]\Mathnoncryptocode{BR2}\arrow[r,dotted,NOTCARE]&\Mathnoncryptocode{R2}\end{tikzcd}):
This is the scenario where a non-crypto register leak gadget executes after a non-crypto branch (i.e., leakage of register $r4$ after \inlineasm{branch2}).
\name{} allows speculative execution for this scenario.
Note, that transient execution of this scenario does not violate software isolation guarantees of the non-crypto code (i.e., it still satisfies the $\CustomInterf{\cdot}{arch}{seq}$ contract).

\begin{figure}
    \centering
    \includegraphics[trim=0 1.4cm 0 0,width=0.85\linewidth]{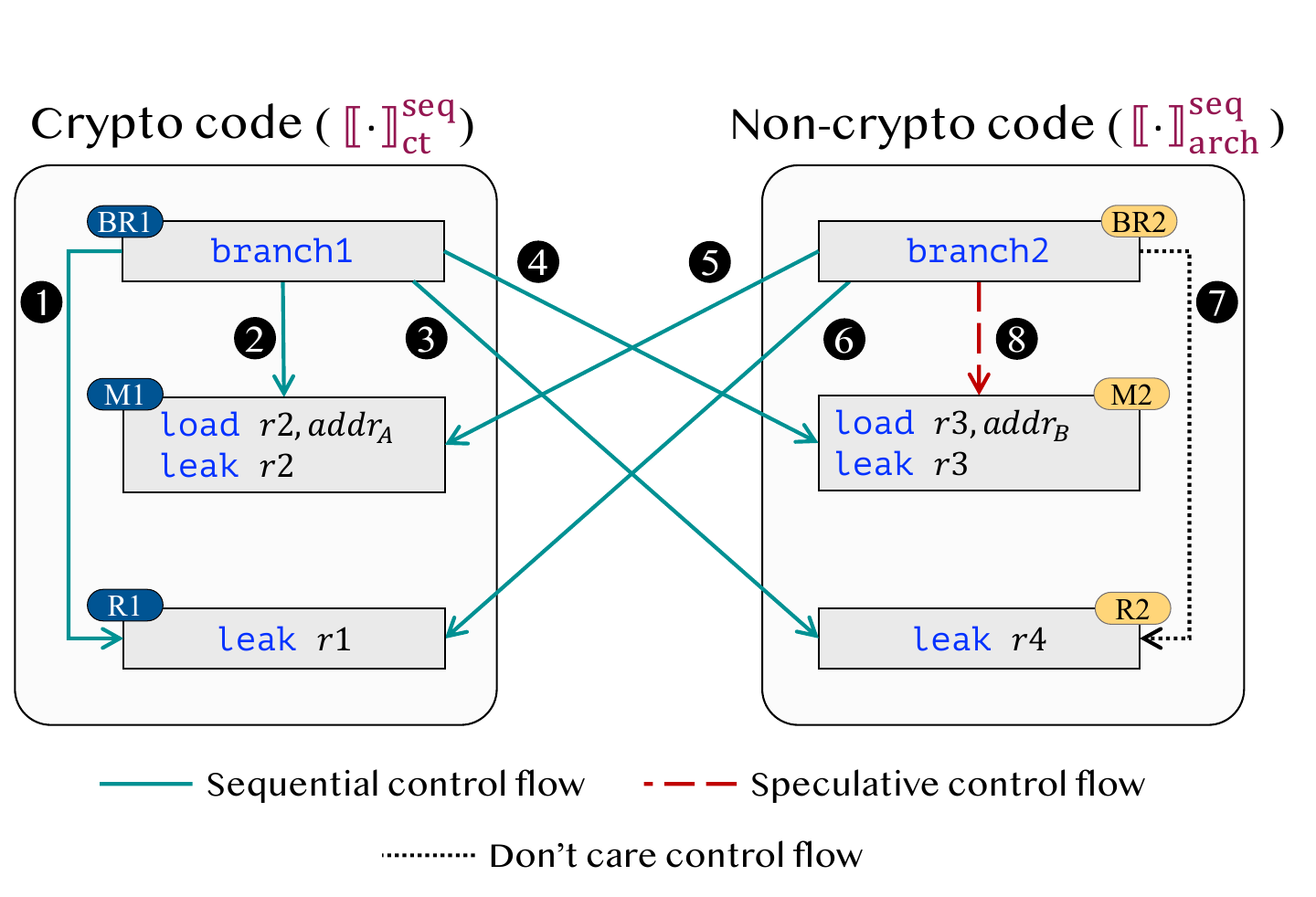}
    \caption{All possible control flows in a \name{}-enabled processor. \name{} guarantees that (1) all outgoing edges from \captioncryptocode{BR1} follow the sequential flow, and (2) all incoming edges to \captioncryptocode{M1} and \captioncryptocode{R1} are on the sequential path.}
    \label{fig:cassandra-scenarios}
\end{figure}

\textbf{Scenario~\flowstep{8}} 
(\begin{tikzcd}[cramped,sep=small]\Mathnoncryptocode{BR2}\arrow[r,dashed,SPEC]&\Mathnoncryptocode{M2}\end{tikzcd}):
This is the scenario where a non-crypto memory leak gadget executes after a non-crypto branch (i.e., leakage of memory $addr_B$ after \inlineasm{branch2}).
\name{} allows speculative execution for this scenario.
However, transient execution of this scenario can violate software isolation guarantees of the non-crypto code. Preventing this leakage is out of the scope of \name{}.
Since violating software isolation can potentially leak arbitrary memory locations (e.g., secret keys),
we expect a \name{}-enabled system to provide a level of isolation for crypto applications (e.g., through lightweight isolation techniques that prevent Spectre~\cite{schwarzl2022robust,reis2019site}).
Ideally, \name{} can be integrated with a defense that provides secure speculation for sandboxing policy (e.g., STT~\cite{yu2019speculative}, DOLMA~\cite{loughlin2021dolma}, and Levioso~\cite{hajiabadi2024levioso}) to comprehensively block Spectre-type attacks for both constant-time and software isolation. Integration of \name{} with other defenses is straightforward; the only consideration is that \textit{crypto branches do not induce a speculation window}, and only speculatively loaded data under the speculation window of non-crypto branches need protection (i.e., scenario~\flowstep{8}).

\begin{figure*}
    \centering
    \includegraphics[trim=0 1cm 0 0,width=0.9\linewidth]{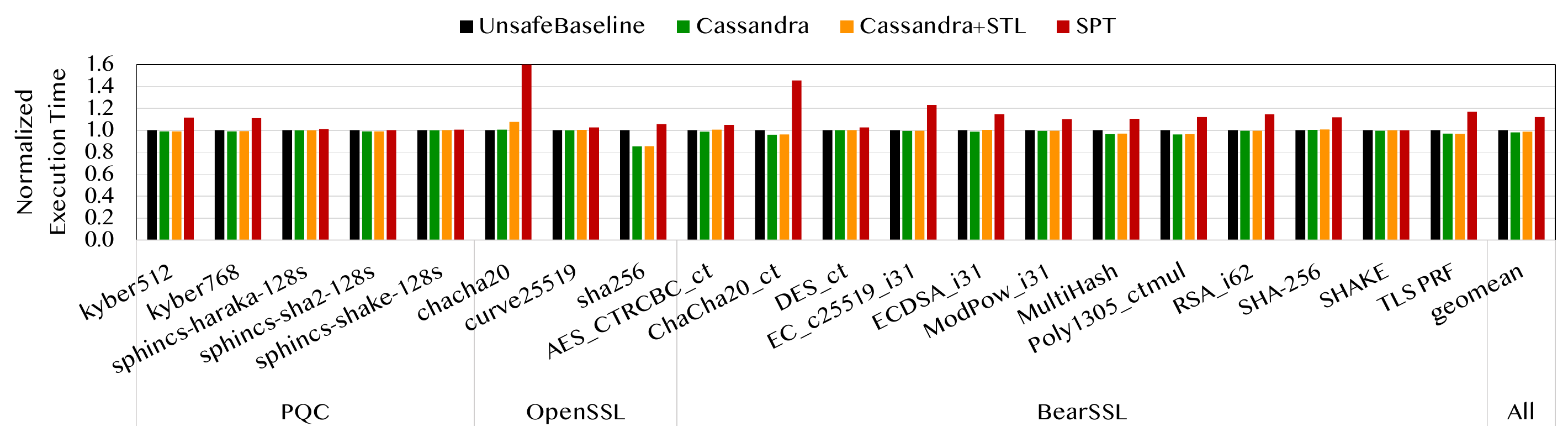}
    \caption{Execution time of different designs normalized to the \baseline{}. A higher bar means higher slowdown.
    }
    \label{fig:perf-bearssl}
\end{figure*}

\textbf{Formal security analysis}.
We provide a formalization of \name{} in \secref{sec:formal-analysis}, where we demonstrate an interesting use of hardware-software contracts~\cite{guarnieri2021hardware}. 
While hardware-software contracts are mostly used to infer contracts for existing defenses, we provide \textit{contract-informed hardware semantics} with a $\CtSeqInterf{\cdot}$ contract in mind as a clean-slate design (i.e., all fetch directions strictly follow the contract trace). Hence, the contract satisfaction proof for $\CtSeqInterf{\cdot}$ is a direct result of our contract-informed semantics. We show that our key innovations in trace compression and microarchitecture enable a performant implementation of \name{}'s semantics for cryptographic applications.

\section{Evaluation}
\label{sec:evaluation}

\subsection{Experimental Setup}
\label{sec:setup}

\textbf{Simulation}. We implement \name{} on top of the gem5 OoO core implementation and evaluate the design using Syscall Emulation (SE) mode and x86 ISA.
Table~\ref{table:gem5-config} shows the system configuration (a Golden-Cove-like microarchitecture~\cite{golden-cove}).
We use McPAT 1.3~\cite{li2013mcpat} and CACTI 6.5~\cite{li2011cacti} to investigate the power and area impacts.

\textbf{Workloads}. We evaluate test applications from the BearSSL~\cite{bearssl} and OpenSSL~\cite{openssl} libraries, alongside \textit{reference} implementations of two post-quantum crypto programs: Kyber~\cite{kyber} and SPHINCS+~\cite{sphincs}.
For the applications with more than 1B instructions, we used SimPoint~\cite{sherwood2002automatically} to generate representative regions for practical simulation time-frames (an average of 6 SimPoints per application and 50M instructions per region). 
In \secref{sec:synthetic-bench}, we evaluate the SpectreGuard synthetic benchmarks~\cite{fustos2019spectreguard} as a mix of crypto and non-crypto code.
Moreover, we used gem5 itself to collect branch traces for \name{}, however, other tools can be used as well (e.g., Intel Pin~\cite{luk2005pin} and DynamoRIO~\cite{bruening2004dynamorio}).

\begin{table}[t]
    \centering
    \caption{gem5 configuration for simulation.}
    \vspace{-0.3cm}
    \resizebox{0.95\linewidth}{!}{%
    \begin{tabular}{l|l}
    \toprule
     Pipeline & 8 F/D/I/C width, 192/114 LQ/SQ entries, 512 ROB\\ 
              & entries, 96 IQ entries, 280/332 RF (INT/FP), LTAGE BPU\\
     \BTU{}   & 16 \PT{}/\TC{}/\CT{} entries (\storageKB{} storage)\\
     \midrule
     L1 DCache    & 48 KB, 64 B line, 12-way, 5-cycle latency \\
     L1 ICache    & 32 KB, 64 B line, 8-way, 5-cycle latency \\
     L2 Cache     & 1280 KB, 64 B line, 16-way, 14-cycle latency \\
     L3 Cache     & 30 MB, 64 B line, 16-way, 40-cycle latency \\
    \bottomrule
    \end{tabular}
    }
    \label{table:gem5-config}
\end{table}

\subsection{Cryptographic Benchmarks Performance}
\label{sec:perf-results}

We evaluate four different designs in this section:
\begin{itemize}[leftmargin=0.5cm]
    \item \baseline{}: unprotected baseline OoO processor, vulnerable to control flow and data flow speculation;
    \item \name{}: our design; addressing control flow speculation;
    \item \stl{}: an extension of \name{} that addresses data flow speculation as well; it always sends a request to memory even if there is a load-store address match, and also restricts the dependents of bypassing loads until prior stores resolve, similar to prior work~\cite{loughlin2021dolma,choudhary2021speculative};
    \item \spt{}: a prior hardware-only defense~\cite{choudhary2021speculative}.
    We use their proposed settings for the Spectre attack model.
\end{itemize}

Figure~\ref{fig:perf-bearssl} shows normalized execution time of the evaluated applications with different designs.
\name{} \textit{improves} performance compared to the \baseline{} by \ourPerf{} on average. This is mainly because of the elimination of prediction for crypto branches, and as a result, no ROB squashes and penalties occur due for mispredicting crypto branches.

In addition, the results show that extending \name{} to protect data flow speculation (i.e., \stl{}) achieves a speedup of \stlPerf{}, which demonstrates that \name{} can still improve performance due to easy-to-resolve address computations in crypto primitives.

Finally, \spt{} shows a \sptPerf{} slowdown compared to the \baseline{}, and a \sptPerftoOurs{} slowdown compared to the \name{}. \spt{} has low overheads for some applications,
but the overheads can be significantly higher, up to a \sptChaChaPerf{} slowdown for OpenSSL \texttt{chacha20}, while \name{} improves performance by \oursShaPerf{} for OpenSSL \texttt{sha256} compared to the \baseline{}.

\subsection{Synthetic Benchmarks Performance}
\label{sec:synthetic-bench}

In this section, we aim to compare \name{} with \prospect{}, the state-of-the-art secure speculation for constant-time programs.

\textit{Implementation}. \prospect{} annotates secret memory regions, and only blocks speculative execution of an instruction if it is about to leak a secret. 
We implement \prospect{} in gem5 and block execution under two conditions: (1) the instruction is speculative (i.e., there is an older, unresolved control inducer), and (2) the instruction is about to process a secret (i.e., one or more operands are tainted). Destination registers of loads from secret memory regions are taint sources that are propagated during execution.
Also, we declassify all registers at the end of crypto primitives.

\textit{Workloads}. We evaluate the synthetic benchmark from SpectreGuard~\cite{fustos2019spectreguard},
which is a mix of non-crypto, (\texttt{s})andboxed code, and (\texttt{c})rypto code (\texttt{s}/\texttt{c} indicates the fraction of each part).
We evaluate two crypto primitives: (1) HACL* \texttt{chacha20}~\cite{zinzindohoue2017hacl} (similar to \prospect{}), and (2) \texttt{curve25519}-donna~\cite{curve25519-donna}\footnote{We use the secret annotations provided in \url{https://github.com/proteus-core/prospect/}}. 
The main difference between these two crypto primitives is that \texttt{chacha20} does not spill secret variables to the stack, while \texttt{curve25519} spills both secret and public variables to the stack which means we need to label the stack as a secret memory region\footnote{Note that we use a different setup and compiler compared to \prospect{} which impacts the stack spills; we use the Clang v14.0.4 compiler for an x86 target, while \prospect{} uses riscv-gnu-toolchain for a RISC-V target.}. 

We evaluate two designs: (1) \prospect{}~\cite{daniel2023prospect}, 
and (2) \name{}+\prospect{}. Since \name{} only enforces sequential execution for crypto branches, we still leave \prospect{} enabled alongside \name{} to prevent transient memory leaks of annotated secret regions during the non-crypto component (i.e., senario~\flowstep{8} in Figure~\ref{fig:cassandra-scenarios}). 
Although \prospect{} is not specifically designed as a general solution for software isolation since non-crypto applications do not have a clear notion of secret annotation similar to crypto applications, and all architecturally out-of-bounds memory accesses are confidential.
Figure~\ref{fig:perf-synthetic-bench} shows the performance impacts of \prospect{} and \name{}+\prospect{}.

\begin{figure}
    \centering
    \includegraphics[trim=0 1cm 0 0,width=\linewidth]{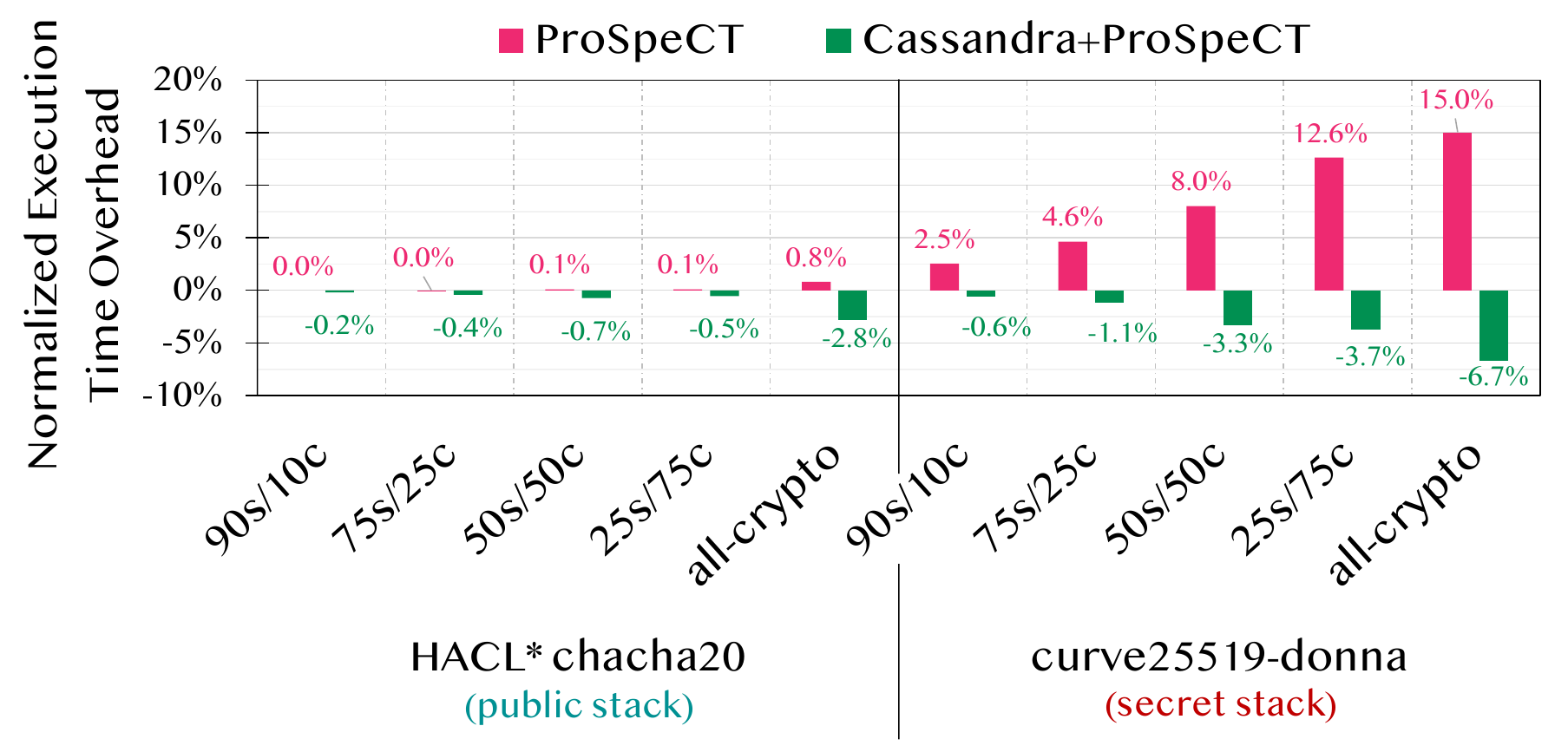}
    \caption{Execution time of \prospect{} and \name{}, normalized to the respective \baseline{} of each benchmark. 
    Negative numbers mean speedup. 
    The stack is marked as public in \texttt{chacha20}, and secret in \texttt{curve25519}.
    }
    \label{fig:perf-synthetic-bench}
    \vspace{-0.5cm}
\end{figure}

Both \prospect{} and \name{} have marginal impact on the performance for all benchmark configurations of \texttt{chacha20} (\name{} shows \ourPerfChachaAllCrypto{} improvement and \prospect{} has only \ProspectChachaAllCrypto{} slowdown for the \texttt{all-crypto} case).
However, \prospect{} experiences a significant slowdown for \texttt{curve25519};
\prospect{} incurs a slowdown between \ProspectCurveNinetyS{} and \ProspectCurveAllCrypto{} when increasing the crypto fraction of the workload.
Interestingly, \name{} provides more speedup when increasing the crypto fraction of the workload, from \ourPerfCurveNinetyS{} to \ourPerfCurveAllCrypto{}.
The main reasons for \name{}'s improvements are: 
(a) \name{} does not induce any control flow speculation for the crypto component that benefits from always-correct fetch redirections (i.e., no penalties for crypto branch mispredictions and squashing cycles), and
(b) it does not restrict the execution for the crypto component, unlike \prospect{} that needs to label the stack as secret for more complex primitives like \texttt{curve25519} and restrict speculative execution of the instructions that potentially process secrets.

\textit{Summary}. Our experiments demonstrate that the worst-case for \name{} is marginal/no performance improvement while \prospect{} can experience significant slowdown for complex crypto applications. 
In addition, \prospect{}'s performance relies on (i) manual programming efforts to precisely annotate secret memory regions, and (ii) how variables are handled with respect to stack spills which is not trivial to manually isolate secret and public variables. 
\prospect{} requires a compiler pass to be aware of secret/public annotations when spilling values to the stack to limit the performance overheads.
In addition, the compiler might need to co-locate secrets in memory to avoid slowdown due to caching effects~\cite{daniel2023prospect}.
On the other hand, \name{} does not need any programming effort or new compiler passes to achieve its full performance potential.

\subsection{Power and Area Impacts}
\label{sec:power-impacts}

Figure~\ref{fig:power} shows the power and area of \name{} compared to \baseline{}. 
The results show that \name{} can reduce the power consumption compared to the \baseline{} by \ourPower{}. The reason is that crypto branches avoid accessing and updating the BPU and access \BTU{} as a smaller and simpler unit (see the reduction in \textit{Instruction Fetch Unit}). 
Our results confirm that \name{} will not add power overheads, nevertheless, the benefits might not be as high when combined with non-crypto workloads.
Finally, \BTU{} has an area overhead of \ourArea{}.

\subsection{Upfront Trace Generation Runtime}
\label{sec:trace-gen-results}

We have evaluated the analysis time for each step of the trace generation procedure (steps~\flowstepletter{A}-\flowstepletter{E} in Algorithm~\ref{algo:trace-generation}). 
We use Intel Pin~\cite{luk2005pin} for dynamic analysis and gathering \raw{} traces.
Branch detection~\flowstepletter{A} is executed once per application and it takes 388 seconds on average. 
Steps~\flowstepletter{B}-\flowstepletter{E} are executed per branch. 
The results show that collecting \raw{} traces (step~\flowstepletter{B}) takes 14 seconds on average per branch and \tandem{} compression (step~\flowstepletter{E}) takes only 3 seconds.

\begin{figure}
    \centering
    \includegraphics[trim=0 1cm 0 0,width=0.8\linewidth]{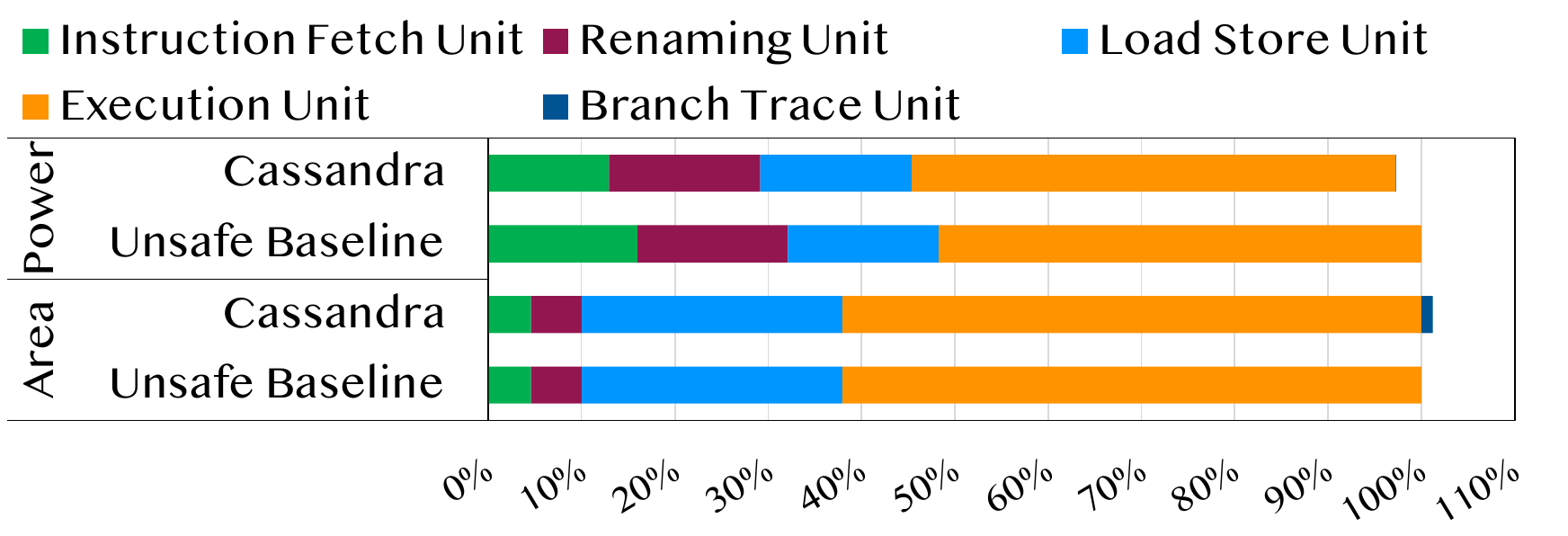}
    \caption{Power and area of \name{}, normalized to the total power and area of the \baseline{}.}
    \label{fig:power}
    \vspace{-0.4cm}
\end{figure}

\section{Discussion}

\textbf{\questionN{1}: Does \name{} handle branches influenced by public parameters?}
Constant-time rules mandate all branches to be independent of confidential inputs. However, they can depend on public information.
In \secref{sec:trace-generation}, we discussed how to handle branches that their traces change in different runs (e.g., stream loops in stream ciphers).
In addition, some branches are influenced by public parameters of the primitive that are specified 
by standards and underlying algorithms which do not change during execution. Hence, \name{} would still generate traces for such branches. However, in some cases, different recommended modes exist for the same application (e.g., key sizes of 128, 192, and 256 for AES). One possible solution for these cases is to generate separate traces for each mode and embedding all of them in the binary. A status register is set to specify the mode before execution and when combined with the hint information it allows accessing the proper traces.
An alternative solution is to generate separate binaries for each mode.

\textbf{\questionN{2}: Who will provide branch traces and when?}
Traces need to be re-generated after each compilation only if PCs change. We believe developers can generate traces for recommended modes (e.g., AES-128/192/256) 
and embed hint information in binaries using our automated tool (\secref{sec:trace-generation}). However, users can also generate traces for legacy binaries of cryptographic programs they intend to 
run on a \name{}-enabled processor with the same procedure.

\textbf{\questionN{3}: What are the benefits and limitations of \name{} only handling single-target branches?}
As we discussed in \secref{sec:trace-communication}, many branches in cryptographic programs are \textit{single-target} (they always jump to the same target). \name{}'s branch analysis marks single-target branches and does not consume the \BTU{} entries for such branches. A lightweight version of \name{}, called \lightVersion{}, can only support single-target branches and stall fetch for multi-target branches until they resolve (i.e., \lightVersion{} does not need to implement a \BTU{}).
Our evaluations show that \lightVersion{} incurs 2.7\%, 6.7\%, and 4.7\% slowdown over \name{} in BearSSL, OpenSSL, and PQC programs, respectively. Some workloads see significant slowdowns (22\% for OpenSSL \texttt{sha256}, 8\% for \texttt{kyber512}) since limiting \name{} to single-target branches diminishes benefits for conditional branches and returns, where even BPUs struggle. However, \name{} ensures equal or better performance across complex, widely used applications via near-perfect control-flow (e.g., 14.5\% speedup for OpenSSL \texttt{sha256} over an \baseline{} using the BPU). While \lightVersion{} performs better than \spt{}, it lacks \name{}’s performance reliability and improvement and it introduces high overhead for key applications like OpenSSL \texttt{sha256} and \texttt{kyber512}.

\textbf{\questionN{4}: How are interrupts handled?}
In \name{}, the OS does not need to store/reload the \BTU{} content during timer interrupts, however, it will flush the BTU if there is a context switch between two different crypto applications (note, that the BTU content is not secret and it is not required to flush it if there is context switch with a non-crypto application). To assess the upper-bound performance impact of flushing the \BTU{}, we flush the \BTU{} at the frequency of 250Hz~\cite{linuxKconfig}; this reduces the performance improvement of \name{} from \ourPerf{} to only \ourPerfTwoFiftyHzFlush{}.

\textbf{\questionN{5}: Does the \BtUnit{} (\BTU{}) introduce a new timing side channel?}
\BTU{} only contains and caches the sequential control flow trace of the program; constant-time principles assume that this trace is already leaked (e.g., ICache has the same leakage), and they guarantee it is completely independent from secrets.

\section{Related Work}
\label{sec:related-work}

In this section, we summarize prior work that provide protection for constant-time code against speculative execution attacks. We categorize these solutions into three classes: (1) hardware-only defenses, (2) software-level defenses, and (3) hardware/software co-designed defenses.

\textbf{(1) Hardware-only defenses}. 
\removeRevision{
Early defenses to protect against Spectre were either protecting a specific transmission channel (e.g., data caches)~\cite{yan2018invisispec,saileshwar2019cleanupspec} or only preventing speculative secrets~\cite{yu2019speculative,weisse2019nda}.
Recent works propose more restrictive defenses to protect non-speculative secrets as well~\cite{choudhary2021speculative,loughlin2021dolma}.
}
Prior works have investigated hardware defenses
to protect constant-time programs~\cite{choudhary2021speculative,loughlin2021dolma}. 
These defenses are complex to design as they need to track speculation taints in all potential microarchitectural components which can also incur high performance overheads due to limited knowledge about the running applications and their security policy. \name{} only adds a small structure (i.e., \BTU{}), which has better performance and consumes less power compared to the baseline, with modest modifications in the microarchitecture.

\textbf{(2) Software-level defenses}. To harden programs on existing CPUs, compiler passes have been developed that take the speculative execution model of CPUs into account and insert protections as needed~\cite{zhang2023ultimate,shivakumar2023typing,choudhary2023declassiflow}. However, software-level defenses usually result in prohibitive slowdowns.
Serberus~\cite{mosier2024serberus} is a compiler mitigation that relies on existing enforcement mechanisms in the CPUs to
mitigate all known speculation primitives.
Serberus introduces a 21\% slowdown on average, and a 8\% slowdown on large-buffer benchmarks (8KB buffers). 
In our evaluation results, a buffer size of 4KB is used in the synthetic benchmark, and the default buffers in BearSSL tests (e.g., \texttt{ChaCha20} uses a 400B buffer)~\cite{bearssl}.

\textbf{(3) Hardware/software co-designed defenses}. Similar to \name{}, some prior defenses propose hardware/software modifications for efficient Spectre defenses~\cite{schwarz2020context,daniel2023prospect,taram2019context,fustos2019spectreguard,yu2019data,hajiabadi2024levioso}. For example,
ProSpeCT~\cite{daniel2023prospect} manually annotates secrets in the program and blocks speculative execution for instructions that process secrets. We provide a detailed comparison with \name{} in \secref{sec:synthetic-bench}. 

\textbf{Profile-guided branch analysis}. 
There have been proposals to use runtime profiles of applications 
to eliminate branch mispredictions~\cite{khan2022whisper,zangeneh2020branchnet,jimenez2001boolean}. 
These techniques mainly target datacenter applications because they have large code footprints and frequent branch mispredictions.
For example, Whisper~\cite{khan2022whisper} proposes a profile-guided approach that provides hints per static branch to help the branch predictor avoid mispredictions.
However, the goal of these works is to build \textit{approximately} accurate branch histories for better prediction, 
but still rely on speculation.

\section{Conclusion}
\label{sec:conclusion}

In this work, we propose \name{}, a novel hardware/software mechanism to enforce sequential execution for constant-time cryptographic programs.
To achieve this, we perform upfront branch analysis to significantly compress sequential branch traces 
which allows an efficient communication with the hardware.
During execution, the processor uses the sequential branch traces to determine fetch directions and to avoid accessing the branch predictor.
Moreover, \name{} can be easily integrated with other mechanisms for software isolation and guaranteeing secure speculation for sandboxed programs.
Despite providing a strong security guarantee, \name{} counterintuitively improves performance by \ourPerf{}.

\begin{acks}
We thank the anonymous reviewers for their insightful comments, which contributed to enhancing the final version of this paper. We also thank Arash Pashrashid and Kaveh Razavi for their valuable discussions and feedback.
\end{acks}

\appendix

\section{Appendix: A Formalization of \name{}}
\label{sec:formal-analysis}

We provide a formalization of \name{} on top of prior work~\cite{guarnieri2021hardware} and express its formal security guarantees. Informally, we first choose a strong security contract and then ensure that the hardware semantics govern that all produced observations agree with the contract. We refer to this approach as \textit{contract-informed hardware semantics}. While many works try to infer contracts for an existing microarchitecture, we use contracts for a clean-slate design of our microarchitecture, starting with a strong contract. Our key observations from cryptographic programs and innovations in trace compression and microarchitectural design enable an efficient implementation of \name{}'s semantics.

\secref{sec:formal-background} provides the background on hardware-software contracts~\cite{guarnieri2021hardware} as our baseline framework.
We specify \name{} semantics in \secref{sec:semantics} and prove its security in \secref{sec:theorems}.

\subsection{Hardware-Software Contracts Preliminaries}
\label{sec:formal-background}

\subsubsection{ISA Language}
\label{sec:isa}

We rely on the \srclang{} language, a small assembly-like language, with the following syntax:
\begin{center}
	\begin{tabular}{llcl}
	(Expressions) 	&  $e$		& $:=$ & $n \mid x \mid \unaryOp{e} \mid \binaryOp{e_1}{e_2}$ \\
	(Instructions)	&  $i$ 		& $:=$ & $\passign{x}{e} \mid \pload{x}{e} \mid \pstore{x}{e}$ \\
	& & &  $\mid \pcall{f} \mid  \pjz{x}{\lbl} \mid \pret{} $ \\ 
    (Functions) 	&  $\mathcal{F}$		& $:=$ & $\varnothing \mid \pseq{\mathcal{F}}{f \mapsto n}$\\
    (Crypto Tags) & $t$ & $:=$ & $\cryptotag \mid \varepsilon$\\
	(Programs)	&  $p$		& $:=$ & $i @ t \mid \pseq{p_1}{p_2} $%
	\end{tabular}
\end{center}
where $x \in \Var$ and $n, \lbl \in \Val = \Nat \cup \{\bot\}$. $\pc$ refers to a special register that contains the program counter.
In addition, an architectural state $\sigma = \tup{m,a}$ consists of the memory $m: \Val \to \Val$, and register assignment $a: \Var \to \Val$.
Each instruction has a tag $t$ that specifies if it is a crypto instruction and analyzed by \name{}; i.e., crypto instructions are tagged as $\cryptotag$ and the rest are considered to be untagged (specified by $\varepsilon$).

\subsubsection{Contracts}
\label{sec:contracts}

A contract governs the attacker-visible observations of a given program. A contract $\CustomInterf{\cdot}{\beta}{\alpha}$ has two main components:
\begin{itemize}[leftmargin=0.5cm]
    \item \textit{Execution model} \textcolor{contractColor}{$\alpha$} specifies how state transitions occur. For example, the sequential model (denoted as $\interfStyle{seq}$) evaluates the branch condition before transitioning to the next state, while a speculative model (denoted as $\interfStyle{spec}$) predicts the target.
    \item \textit{Leakage model} \textcolor{contractColor}{$\beta$} specifies the leakages that are observable by an attacker. For example, the constant-time leakage model (denoted as  $\interfStyle{ct}$) leaks the control flow and memory addresses.
\end{itemize}
Contract semantics
$\contractTrans{\tau_n}$ is labeled with the \textit{observations} $\tau_n$ when transiting between two architectural states. Observations $\tau_i$ capture leaks via cache and control flow:
\begin{center}
\begin{tabular}{rcl}
    $\cfObs$ & $:=$ & $\pcObs{n} \mid \callObs{f} \mid \retObs{n}$\\
    $\memObs$ & $:=$ & $\loadObs{n} \mid \storeObs{n}$ \\
    $\Obs$ & $:=$ & $\memObs \mid \cfObs$\\
    $\tau:=\varepsilon \mid \Obs$ & &
    $\overline{\tau}:=$ $\emptyset \mid \overline{\tau} \concat \tau @ t$
\end{tabular}
\end{center}
The $\pcObs{n}$, $\callObs{f}$, and $\retObs{n}$ observations record the control flow of the program (denoted as $\cfObs$).
The $\loadObs{n}$ and $\storeObs{n}$ observations record the memory addresses to capture cache leakage (denoted as $\memObs$). 
In addition, observations are tagged with the same crypto tag of the instruction that generates the observation.

For a given program $p$ and initial architectural state $\sigma_0$, the labels of the transitions in run $\sigma_0 \contractTrans{\tau_1 @ t_1} \sigma_1 \contractTrans{\tau_2 @ t_2} \dots \contractTrans{\tau_n @ t_n} \sigma_n$ produce the contract trace $\interfSem{p}(\sigma_0) = [\tau_1@ t_1 \concat \dots \concat \tau_n@ t_n]$.

\textbf{Contract $\CtSeqInterf{\cdot}$}. This contract specifies the strongest security guarantee for secure speculation mechanisms (i.e., sequential execution model for constant-time leakages). For example, two rules of $\CtSeqInterf{\cdot}$ contract are as follows:
\begin{mathpar}[\small]
	\inferrule[(Beqz-Sat)]
	{
	\select{p}{a(\pc)} = \pjz{x}{\lbl}@t  \\
	\tup{m, a} \archStep{}{} \tup{m',a'}
	}
	{
	\tup{m, a} \contractTrans{ \pcObs{\lbl}@t  }{} \tup{ m', a'}
	}
\end{mathpar}
\begin{mathpar}[\small]
	\inferrule[(Load)]
	{
	\select{p}{a(\pc)} = \pload{x}{e}@t \\
	\tup{m, a} \archStep{}{} \tup{m',a'}
	}
	{
	\tup{m, a} \contractTrans{ \loadObs{n}@t  }{} \tup{m', a'}
	}
\end{mathpar}
where $n = \exprEval{e}{a}$ is the result of expression $e$ given register assignment $a$. $\CtSeqInterf{\cdot}$ exposes the control flow ($\pcObs{n}$, $\callObs{f}$, and $\retObs{n}$) and memory addresses ($\loadObs{n}$ and $\storeObs{n}$) in a sequential execution model. Note, that the values of loads and stores are not leaked.

\subsection{\name{} Semantics}
\label{sec:semantics}

In this section, we formalize a contract-informed semantics for the \name{} methodology (denoted as $\ProposedMuarchSem{\cdot}$).

As the first step, we define auxiliary contract traces to enable our contract-informed hardware semantics:

\begin{definition}[Crypto control flow trace $\mathcal{C}$]\label{def:cflookup}
For a given program $p$, initial architectural state $\sigma_0$ and contract $\CustomInterf{\cdot}{\beta}{\alpha}$, $\ControlAux{\beta}{\alpha}(p, \sigma_0)$ is a subtrace of contract trace $\CustomInterf{p}{\beta}{\alpha}(\sigma_0) = [\tau_1@ t_1 \concat \dots \concat \tau_n@ t_n]$, consisting of all crypto control flow observations:
\begin{center}
    $\ControlAux{\beta}{\alpha}(p, \sigma_0) = [\tau_i @ \cryptotag | 1 \leq i \leq n, \tau_i \in \cfObs]$
\end{center}
\end{definition}
We can define the contract memory trace $\MemAux{\beta}{\alpha}(p, \sigma_0)$ in the same way, where it consists of only memory observations. 
Since we target constant-time cryptographic programs, the $\ControlAux{\beta}{\alpha}(p, \sigma_0)$ trace is independent from $\sigma_0$, and we use $\ControlAux{\beta}{\alpha}(p)$ for brevity.
Note, that  $\ControlAux{\beta}{\alpha}(p)(i)$ refers to the $i^{th}$ observation of the crypto control flow trace of contract $\CustomInterf{\cdot}{\beta}{\alpha}$.

\textbf{Hardware configuration}. Hardware configuration $\omega$ in \name{} consists of (1) the architectural state $\sigma$ with the memory $m$ and register assignment $a$, (2) a global counter $\cMuarch$ that counts the number of contract-level control flow observations, (3) reorder buffer $\buf$ with size $\wMuarch$, and (4) the microarchitectural context $\mu$. Microarchitectural context is the part of the microarchitecture that the attacker can observe or influence.
We use an abstract model for caches, pipeline scheduler, and branch predictor, similar to~\cite{guarnieri2021hardware}, and also add the trace cache (specific to the \name{} semantics, representing the \BTU{}). Table~\ref{table:uarch-signatures} shows the interface of each component. 
\begin{itemize}[leftmargin=0.5cm]
    \item \textit{Cache}: The $\CacheAccess$ function results in a $\CacheHit$ or $\CacheMiss$ based on a given cache state $\CacheState$ and memory address $\lbl$, and the $\CacheUpdate$ function generates a new cache state based on the input cache state and memory address;
    \item \textit{Scheduler}: The $\SchedNext$ function determines the next processor step ($\FetchDir$, $\ExeDir$, or $\CommitDir$) given the scheduler state $\SchedState$, and the $\SchedUpdate$ function updates the scheduler's state based on the reorder buffer state;
    \item \textit{Branch Predictor}: The $\BpPredict$ function predicts the outcome based on a given predictor state $\BpState$ and input branch $\lbl$. The $\BpUpdate$ function updates the predictor's state based on a branch and its resolved outcome.
    \item \textit{Trace Cache}: The $\BtAccess$ function results in a $\BtHit$ or $\BtMiss$ based on a given trace cache state $\BtState$ and branch PC $\lbl$. 
    The $\BtUpdate$ function updates the trace cache state if needed (e.g., fetching traces for a given branch PC $\lbl$ that misses in the trace cache).
\end{itemize}
Moreover, a \textit{reorder buffer} records the state of in-flight instructions. Expressions in a reorder buffer are initially unresolved, and they can transform to a resolved state after execution. 
A data-independent projection of a reorder buffer is shown as
$\BufProject{\buf}$ where resolved expressions are replaced with $\resolved$ and unresolved expressions with $\unresolved$. In addition to \cite{guarnieri2021hardware}, we define an examine function that outputs $\resolved$ for a given $\BufProject{\buf}$ if all expressions are resolved:

$\mathsf{examine}(\BufProject{\buf}) = 
   \begin{cases}
       \resolved       & \text{if all expressions in }  \BufProject{\buf} \text{ are } \resolved\\
    \unresolved  & \text{otherwise}
   \end{cases}$

\name{} semantics uses a binary relation ($\stepMapN{\nameAbbr{}}$) that maps hardware configurations to their successors:
\begin{mathpar}[\small]
\inferrule[(Step-\name{})]
{
    d = \SchedNext(\SchedState)\\
	\SchedState' = \SchedUpdate(\SchedState,\buf')\\
	\tup{m,a, \cMuarch, \buf,  \CacheState, \BpState, \BtState } ~\stepArrow{d}~ \tup{m',a', \cMuarch', \buf', \CacheState', \BpState', \BtState'  }
}
{
	\tup{m,a, \cMuarch, \buf,  \CacheState, \BpState, \BtState, \SchedState } ~\stepMap{\nameAbbr{}}~ \\\tup{m',a', \cMuarch', \buf',  \CacheState', \BpState', \BtState', \SchedState' }
}
\end{mathpar}

Given the current hardware state $\omega = \tup{m,a,\cMuarch, \buf,  \CacheState, \BpState, \BtState, \SchedState}$, the rule \textsc{Step-\name{}} finds the next directive $d$ via the $\SchedNext(sc)$ and takes an appropriate step (formalized through the fetch, execution, and commit rules) that produces the new state $\omega' =\tup{m',a', \cMuarch',$ $\buf',  \CacheState', \BtState', \SchedState'}$.

Most transition rules of \name{} are standard and the same as the baseline in~\cite{guarnieri2021hardware} and not presented here for brevity.
The difference between the \name{} rules compared to the baseline occurs in the fetch stage when a branch hits the cache. While the baseline semantics use the \textsc{Fetch-Branch-Hit} rule to determine the fetch direction via prediction, \name{} replaces this with a new set of rules that determines the next PCs based on the crypto tags and the $\CtSeqInterf{\cdot}$ contract.
For tagged branches, \name{} uses the contract traces to determine the fetch direction.
The first rule handles the case that the branch traces miss in the
\textit{Trace Cache} (our additions are highlighted):
\begin{mathpar}[\small]
    \inferrule[(Fetch-Hit-TaggedBranch-Trace-Miss)]
    {
    a' = \apply{\buf}{a}\\
    i = a'(\pc)\\ 
    p(i) = \pjz{x}{\lbl}@\cryptotag \mid \pcall{f}@\cryptotag  \mid \pret{}@\cryptotag   \\
    |\buf| < \wMuarch \\
    \CacheAccess(\CacheState, i) = \CacheHit\\
    \CacheUpdate(\CacheState, i) = \CacheState{}'\\
    \highlightBox{\BtAccess(\BtState, i) = \BtMiss}\\
    \highlightBox{\BtUpdate(\BtState, i) = \BtState'}
    }
    {
        \tup{m,a,\cMuarch, \buf,  \CacheState, \BpState, \BtState} 
        ~\stepArrow{fetch}~
        \tup{m,a,\cMuarch,\buf,\CacheState',\BpState,\BtState'}	
    }
\end{mathpar}
In this rule, the \textit{Trace Cache} is updated to bring the missed traces to hit later\footnote{In all rules, $\apply{\buf}{a}$ obtains the new register assignment $a'$ after applying the changes of resolved instructions in $\buf$~\cite{guarnieri2021hardware}.}.
The second rule handles the case that branch traces hit in the \textit{Trace Cache}:
 \begin{mathpar}[\small]
    \inferrule[(Fetch-Hit-TaggedBranch-Trace-Hit)]
    {
    a' = \apply{\buf}{a}\\
    i = a'(\pc)\\ 
    p(i) = \pjz{x}{\lbl}@\cryptotag  \mid \pcall{f}@\cryptotag  \mid \pret{}@\cryptotag   \\
    |\buf| < \wMuarch \\
    \CacheAccess(\CacheState, i) = \CacheHit\\
    \CacheUpdate(\CacheState, i) = \CacheState{}'\\
    \highlightBox{\BtAccess(\BtState, i) = \BtHit}\\
    \highlightBox{\BtUpdate(\BtState, i) = \BtState'}\\
    \highlightBox{\lbl' = \ControlAux{\interfStyle{ct}}{\interfStyle{seq}}(p)(\cMuarch)}
    }
    {
        \tup{m,a,\cMuarch,\buf,  \CacheState, \BpState, \BtState}
         ~\stepArrow{fetch}~\\
        \tup{m,a,\cMuarch+1,\buf \concat \passign{\pc}{\lbl'},\CacheState',\BpState,\BtState'}	
    }
    \end{mathpar}
In this rule, the next PC is determined through contract-level observations.

\begin{table}[t]
  \centering
  \caption{Signatures of \name{} microarchitecture.}
  \resizebox{\linewidth}{!}{%
  \begin{tabular}{l l l}
  \toprule
  \multicolumn{1}{c}{\textbf{Component}} & \multicolumn{1}{c}{\textbf{States}} & \multicolumn{1}{c}{\textbf{Functions}} \\ 
  \midrule
  
   \multirow{2}{*}{{\emph{Cache}}} & \multirow{2}{*}{{$\CacheStates$}} & $\CacheAccess: \CacheStates \times \Val \to \{\CacheHit,\CacheMiss\}$ \\
   & &  $\CacheUpdate: \CacheStates  \times \Val \to \CacheStates$\\
   
   \hline
   
   \multirow{2}{*}{{\emph{Scheduler}}} & \multirow{2}{*}{{$\SchedStates$}} & $\SchedNext: \SchedStates \to \{\FetchDir, \ExeDir, \CommitDir\}$ \\
   & &  $\SchedUpdate: \SchedStates  \times \ReorderBuffers \to \SchedStates$\\

   \hline 

   \multirow{2}{*}{{\emph{Branch Predictor}}} & \multirow{2}{*}{{$\BpStates$}} & $\BpPredict: \BpStates \times \Val \to \Val$ \\
   & &  $\BpUpdate: \BpStates \times \Val \times \Val \to \BpStates$\\
   
   \hline
   
   \multirow{2}{*}{{\emph{Trace Cache}}} & \multirow{2}{*}{{$\BtStates$}} & $\BtAccess: \BtStates \times \Val \to \{\BtHit, \BtMiss\}$\\
   & &  $\BtUpdate: \BtStates \times \Val \to \BtStates$\\
  \bottomrule
  \end{tabular}
  }
  \label{table:uarch-signatures}
\end{table}

While prior rules handle branches that are tagged by \name{} (i.e., crypto branches), the remaining branches (i.e., non-crypto branches) need to be handled differently. Non-crypto branches use the branch predictor in our design, however, \name{} performs integrity checks to avoid speculative fetch redirection if the predicted target is tagged (i.e., the target is within the crypto code). So, \name{} stalls fetch until the crypto target resolves. The next two rules handle the \name{} integrity checks:
\begin{mathpar}[\small]
    \inferrule[(Fetch-Hit-UntaggedBranch-TaggedTarget-UR)]
    {
    a' = \apply{\buf}{a}\\
    i = a'(\pc)\\ 
    p(i) = \pjz{x}{\lbl}@\varepsilon \mid \pcall{f}@\varepsilon \mid \pret{}@\varepsilon  \\
    |\buf| < \wMuarch \\
    \CacheAccess(\CacheState, i) = \CacheHit\\
    \CacheUpdate(\CacheState, i) = \CacheState{}'\\
    \lbl'=\BpPredict(\BpState, i)\\
    \highlightBox{p(\lbl')=j@\cryptotag}\\
    \highlightBox{\mathsf{examine}(\BufProject{\buf})=\unresolved}
    }
    {
        \tup{m,a,\cMuarch, \buf,  \CacheState, \BpState, \BtState} 
        ~\stepArrow{fetch}~
        \tup{m,a,\cMuarch,\buf,  \CacheState', \BpState, \BtState}	
    }
\end{mathpar}
Once the reorder buffer is resolved, we determine the next PC to fetch based on the specific branch type we are handling. Here, we only show a selected rule for conditional branches. The rules for calls and returns also use resolved information to find the next, sequential direction.
\begin{mathpar}[\small]
    \inferrule[(Fetch-Hit-UntaggedBranch-TaggedTarget-R-Beqz)]
    {
    a' = \apply{\buf}{a}\\
    i = a'(\pc)\\ 
    p(i) =\pjz{x}{\lbl}@\varepsilon \\
    |\buf| < \wMuarch \\
    \CacheAccess(\CacheState, i) = \CacheHit\\
    \CacheUpdate(\CacheState, i) = \CacheState{}'\\
    \highlightBox{\mathsf{examine}(\BufProject{\buf})=\resolved}\\
    \lbl'=\BpPredict(\BpState, i)\\
    \highlightBox{p(\lbl')=j@\cryptotag}\\
    \highlightBox{\lbl'' = {\begin{cases} i + 1 & \text{if $a'(x) \neq 0$} \\
                        \lbl & \text{if $a'(x) = 0$}
     \end{cases}
     }}
    }
    {
        \tup{m,a,\cMuarch,\buf,  \CacheState, \BpState, \BtState}
         ~\stepArrow{fetch}~
        \tup{m,a,\cMuarch,\buf \concat \passign{\pc}{\lbl''},  \CacheState', \BpState, \BtState}		
    }
\end{mathpar}

Finally, we allow normal fetch and prediction for untagged branches if the predicted target is untagged as well (see Figure~\ref{fig:cassandra-scenarios} for all control flow transitions of a \name{}-enabled processor):
\begin{mathpar}[\small]
    \inferrule[(Fetch-Hit-UntaggedBranch-UntaggedTarget)]
    {
    a' = \apply{\buf}{a}\\
    i = a'(\pc)\\ 
    p(i) = \pjz{x}{\lbl}@\varepsilon \mid \pcall{f}@\varepsilon \mid \pret{}@\varepsilon  \\
    |\buf| < \wMuarch \\
    \CacheAccess(\CacheState, i) = \CacheHit\\
    \CacheUpdate(\CacheState, i) = \CacheState{}'\\
    \lbl'=\BpPredict(\BpState, i)\\
    \highlightBox{p(\lbl')=j@\varepsilon}
    }
    {
        \tup{m,a,\cMuarch, \buf,  \CacheState, \BpState, \BtState} 
        ~\stepArrow{fetch}~
        \tup{m,a,\cMuarch,\buf \concat \passign{\pc}{\lbl'},  \CacheState',\BpState, \BtState'}	
    }
\end{mathpar}

Since \name{} exploits contract-level observations of the $\CtSeqInterf{\cdot}$, it is guaranteed that no branch mispredictions happen and there is no need to recover the $\cMuarch$ state for squashing branches. In addition, we assume that data flow speculation is disabled and they cannot cause squashes as well (our experiments and prior work~\cite{mosier2024serberus} show that disabling or naively addressing data flow speculation for cryptographic programs incurs negligible overheads).

\subsection{Definitions and Theorems}
\label{sec:theorems}

\textbf{Adversary model}. We define the adversary as a projection function $\mathcal{A}$ that specifies observations from a microarchitectural context. 
For a given hardware semantics $\muarchSem{\cdot}$ and program $p$, hardware run $\omega_0 \hwTrans{} \omega_1 \hwTrans{} \dots \hwTrans{} \omega_n$ produces the hardware observations:
    $\muarchSem{p}(\sigma_0) = [\mathcal{A}(\omega_0) \mathcal{A}(\omega_1) \dots \mathcal{A}(\omega_n)]$.
\begin{definition}[$\omega \approx \omega'$]\label{def:hw-indistinguishability}
	Two hardware configurations $\omega = \tup{m,a, \cMuarch,$ $\buf, \CacheState,\BtState, \SchedState}$ and $\omega' = \tup{m',a',\cMuarch,\buf', \CacheState',\BtState', \SchedState'}$ are \textit{indistinguishable}, iff $\mathcal{A}(\omega) = \mathcal{A}(\omega')$.
\end{definition}
We consider an adversary that observes the entire microarchitectural context, including the reorder buffer, the cache (which only contains the addresses, not the values), the trace cache, and the branch predictor.

To express security guarantees of hardware semantics $\muarchSem{\cdot}$ against a contract $\interfSem{\cdot}$, we use Definition~\ref{def:hni} from~\cite{guarnieri2021hardware}.

\begin{definition}[$\hsni{\interfSem{\cdot}}{\muarchSem{\cdot}}$]\label{def:hni}
A hardware semantics $\muarchSem{\cdot}$ satisfies a contract  $\interfSem{\cdot}$ if for an arbitrary program $p$ and arbitrary initial architectural states $\sigma,\sigma'$:
\begin{center}
    $\interfSem{\Prg}(\sigma) = \interfSem{\Prg}(\sigma') \Rightarrow \muarchSem{\Prg}(\sigma)= \muarchSem{\Prg}(\sigma')$.
\end{center}
\end{definition}
Note, that we require the initial microarchitectural components to be the same for this definition.

\begin{restatable}{thm}{allCtSpec}
    \label{theorem:hni:all}
    For all cryptographic programs that branches are tagged: 
    $\hsni{\CtSeqInterf{\cdot}}{\ProposedMuarchSem{\cdot}}$.
\end{restatable}

\begin{proof}
Let $p$ be an arbitrary crypto program. Moreover, let $\sigma_0 = \tup{m,a}$ and $\sigma_0' = \tup{m',a'}$ be two arbitrary initial architectural states. Two possible cases are:
\begin{enumerate}[leftmargin=0.5cm]
    \item $\CtSeqInterf{p}(\sigma_0) \neq \CtSeqInterf{p}(\sigma_0')$: which trivially holds $\interfSem{\Prg}(\sigma_0) = \interfSem{\Prg}(\sigma_0')$ $\Rightarrow \muarchSem{\Prg}(\sigma_0)= \muarchSem{\Prg}(\sigma_0')$.
    \item $\CtSeqInterf{p}(\sigma_0) = \CtSeqInterf{p}(\sigma_0')$: By unrolling $\CtSeqInterf{p}(\sigma)$, two contract runs are obtained that agree on all observations ($\forall 0 \leq i \leq n: \tau_i = \tau_i'$):
\end{enumerate}
\begin{center}
    $\crun := \sigma_0 \contractTrans{\tau_1} \sigma_1 \contractTrans{\tau_2} \dots \contractTrans{\tau_n} \sigma_n$  \hspace{0.3cm}$\crunp := \sigma_0' \contractTrans{\tau_1'} \sigma_1' \contractTrans{\tau_2'} \dots \contractTrans{\tau_n'} \sigma_n'$
\end{center}
and produced hardware runs of $\ProposedMuarchSem{p}(\sigma_0)$ and $\ProposedMuarchSem{p}(\sigma_0')$:
\begin{center}
    $\hrun := \omega_0 \stepMap{\nameAbbr{}} \omega_1 \stepMap{\nameAbbr{}} \dots \stepMap{\nameAbbr{}} \omega_m$\\
     $\hrunp := \omega_0' \stepMap{\nameAbbr{}} \omega_1' \stepMap{\nameAbbr{}} \dots \stepMap{\nameAbbr{}} \omega_m'$
\end{center}
where $\hrun(i) = \omega_i$ and $\crun(i) = \sigma_i$. We prove by induction that $\ProposedMuarchSem{\Prg}(\sigma_0)= \ProposedMuarchSem{\Prg}(\sigma_0')$, i.e., $\forall 0 \leq i \leq m: \hrun(i) \approx \hrunp(i)$.

(\textit{Induction basis}): the initial hardware configurations $\hrun(0)$ and $\hrunp(0)$ are indistinguishable by definition as they agree on their microarchitectural components.

(\textit{Inductive step}): assume that after $i$ steps in $\ProposedMuarchSem{\Prg}$: $\hrun(i) \approx \hrunp(i)$. 
Since in our hardware semantics observations are informed by the contract $\CtSeqInterf{\cdot}$,
the corresponding contract runs of $\crun$ and $\crunp$ take the same steps $k$. In other words, the corresponding contract state of $\hrun(i)$ is $\crun(k)$, and the corresponding contract state of $\hrunp(i)$ is $\crunp(k)$. 
Based on our assumptions, (a) $\hrun(i)$ and $\hrunp(i)$ agree on all microarchitectural components and (b) the next steps of $\ProposedMuarchSem{\Prg}$ to obtain $\hrun(i+1)$ and $\hrunp(i+1)$ are determined by the $\crun(k+1)$ and $\crunp(k+1)$ observations, which are the same by assumption. Hence, based on (a) and (b): $\hrun(i+1) \approx \hrunp(i+1)$.
\end{proof}

Note, that Theorem~\ref{theorem:hni:all} only holds for crypto programs. Processors would have to deploy another defense for non-crypto branches to guarantee software isolation as well (e.g., Guarnieri et al.~\cite{guarnieri2021hardware} prove that STT~\cite{yu2019speculative} and NDA~\cite{weisse2019nda} satisfy the $\ArchSeqInterf{\cdot}$ contract which is sufficient for software isolation).

\bibliographystyle{ACM-Reference-Format}
\bibliography{refs}


\begin{thebibliography}{85}


\ifx \showCODEN    \undefined \def \showCODEN     #1{\unskip}     \fi
\ifx \showDOI      \undefined \def \showDOI       #1{#1}\fi
\ifx \showISBNx    \undefined \def \showISBNx     #1{\unskip}     \fi
\ifx \showISBNxiii \undefined \def \showISBNxiii  #1{\unskip}     \fi
\ifx \showISSN     \undefined \def \showISSN      #1{\unskip}     \fi
\ifx \showLCCN     \undefined \def \showLCCN      #1{\unskip}     \fi
\ifx \shownote     \undefined \def \shownote      #1{#1}          \fi
\ifx \showarticletitle \undefined \def \showarticletitle #1{#1}   \fi
\ifx \showURL      \undefined \def \showURL       {\relax}        \fi
\providecommand\bibfield[2]{#2}
\providecommand\bibinfo[2]{#2}
\providecommand\natexlab[1]{#1}
\providecommand\showeprint[2][]{arXiv:#2}

\bibitem[Agrawal and Srikant(1995)]%
        {agrawal1995mining}
\bibfield{author}{\bibinfo{person}{Rakesh Agrawal} {and} \bibinfo{person}{Ramakrishnan Srikant}.} \bibinfo{year}{1995}\natexlab{}.
\newblock \showarticletitle{Mining sequential patterns}. In \bibinfo{booktitle}{\emph{IEEE International Conference on Data Engineering (ICDE)}}.
\newblock


\bibitem[Almeida et~al\mbox{.}(2017)]%
        {almeida2017jasmin}
\bibfield{author}{\bibinfo{person}{Jos{\'e}~Bacelar Almeida}, \bibinfo{person}{Manuel Barbosa}, \bibinfo{person}{Gilles Barthe}, \bibinfo{person}{Arthur Blot}, \bibinfo{person}{Benjamin Gr{\'e}goire}, \bibinfo{person}{Vincent Laporte}, \bibinfo{person}{Tiago Oliveira}, \bibinfo{person}{Hugo Pacheco}, \bibinfo{person}{Benedikt Schmidt}, {and} \bibinfo{person}{Pierre-Yves Strub}.} \bibinfo{year}{2017}\natexlab{}.
\newblock \showarticletitle{Jasmin: High-assurance and high-speed cryptography}. In \bibinfo{booktitle}{\emph{ACM Conference on Computer and Communications Security (CCS)}}.
\newblock


\bibitem[Almeida et~al\mbox{.}(2016)]%
        {almeida2016verifying}
\bibfield{author}{\bibinfo{person}{Jos{\'e}~Bacelar Almeida}, \bibinfo{person}{Manuel Barbosa}, \bibinfo{person}{Gilles Barthe}, \bibinfo{person}{Fran{\c{c}}ois Dupressoir}, {and} \bibinfo{person}{Michael Emmi}.} \bibinfo{year}{2016}\natexlab{}.
\newblock \showarticletitle{Verifying Constant-Time Implementations}. In \bibinfo{booktitle}{\emph{USENIX Security Symposium}}.
\newblock


\bibitem[Barthe et~al\mbox{.}(2024)]%
        {barthe2024testing}
\bibfield{author}{\bibinfo{person}{Gilles Barthe}, \bibinfo{person}{Marcel B{\"o}hme}, \bibinfo{person}{Sunjay Cauligi}, \bibinfo{person}{Chitchanok Chuengsatiansup}, \bibinfo{person}{Daniel Genkin}, \bibinfo{person}{Marco Guarnieri}, \bibinfo{person}{David~Mateos Romero}, \bibinfo{person}{Peter Schwabe}, \bibinfo{person}{David Wu}, {and} \bibinfo{person}{Yuval Yarom}.} \bibinfo{year}{2024}\natexlab{}.
\newblock \showarticletitle{Testing side-channel security of cryptographic implementations against future microarchitectures}. In \bibinfo{booktitle}{\emph{ACM Conference on Computer and Communications Security (CCS)}}.
\newblock


\bibitem[Barthe et~al\mbox{.}(2021)]%
        {barthe2021high}
\bibfield{author}{\bibinfo{person}{Gilles Barthe}, \bibinfo{person}{Sunjay Cauligi}, \bibinfo{person}{Benjamin Gr{\'e}goire}, \bibinfo{person}{Adrien Koutsos}, \bibinfo{person}{Kevin Liao}, \bibinfo{person}{Tiago Oliveira}, \bibinfo{person}{Swarn Priya}, \bibinfo{person}{Tamara Rezk}, {and} \bibinfo{person}{Peter Schwabe}.} \bibinfo{year}{2021}\natexlab{}.
\newblock \showarticletitle{High-assurance cryptography in the {Spectre} era}. In \bibinfo{booktitle}{\emph{IEEE Symposium on Security and Privacy (SP)}}.
\newblock


\bibitem[Benson(1999)]%
        {benson1999tandem}
\bibfield{author}{\bibinfo{person}{Gary Benson}.} \bibinfo{year}{1999}\natexlab{}.
\newblock \showarticletitle{Tandem repeats finder: a program to analyze DNA sequences}.
\newblock \bibinfo{journal}{\emph{Nucleic Acids Research}} (\bibinfo{year}{1999}).
\newblock


\bibitem[Bernstein et~al\mbox{.}(2012)]%
        {bernstein2012security}
\bibfield{author}{\bibinfo{person}{Daniel~J Bernstein}, \bibinfo{person}{Tanja Lange}, {and} \bibinfo{person}{Peter Schwabe}.} \bibinfo{year}{2012}\natexlab{}.
\newblock \showarticletitle{The security impact of a new cryptographic library}. In \bibinfo{booktitle}{\emph{International Conference on Cryptology and Information Security in Latin America (LATINCRYPT)}}.
\newblock


\bibitem[Bruening and Amarasinghe(2004)]%
        {bruening2004dynamorio}
\bibfield{author}{\bibinfo{person}{Derek Bruening} {and} \bibinfo{person}{Saman Amarasinghe}.} \bibinfo{year}{2004}\natexlab{}.
\newblock \showarticletitle{Efficient, transparent, and comprehensive runtime code manipulation}.
\newblock \bibinfo{journal}{\emph{Ph.D. Dissertation, Massachusetts Institute of Technology, Department of Electrical Engineering}} (\bibinfo{year}{2004}).
\newblock


\bibitem[Canella et~al\mbox{.}(2019a)]%
        {canella2019fallout}
\bibfield{author}{\bibinfo{person}{Claudio Canella}, \bibinfo{person}{Daniel Genkin}, \bibinfo{person}{Lukas Giner}, \bibinfo{person}{Daniel Gruss}, \bibinfo{person}{Moritz Lipp}, \bibinfo{person}{Marina Minkin}, \bibinfo{person}{Daniel Moghimi}, \bibinfo{person}{Frank Piessens}, \bibinfo{person}{Michael Schwarz}, \bibinfo{person}{Berk Sunar}, {et~al\mbox{.}}} \bibinfo{year}{2019}\natexlab{a}.
\newblock \showarticletitle{Fallout: Leaking data on meltdown-resistant {CPUs}}. In \bibinfo{booktitle}{\emph{ACM Conference on Computer and Communications Security (CCS)}}.
\newblock


\bibitem[Canella et~al\mbox{.}(2019b)]%
        {canella2019systematic}
\bibfield{author}{\bibinfo{person}{Claudio Canella}, \bibinfo{person}{Jo Van~Bulck}, \bibinfo{person}{Michael Schwarz}, \bibinfo{person}{Moritz Lipp}, \bibinfo{person}{Benjamin Von~Berg}, \bibinfo{person}{Philipp Ortner}, \bibinfo{person}{Frank Piessens}, \bibinfo{person}{Dmitry Evtyushkin}, {and} \bibinfo{person}{Daniel Gruss}.} \bibinfo{year}{2019}\natexlab{b}.
\newblock \showarticletitle{A systematic evaluation of transient execution attacks and defenses}. In \bibinfo{booktitle}{\emph{USENIX Security Symposium}}.
\newblock


\bibitem[Cauligi et~al\mbox{.}(2020)]%
        {cauligi2020constant}
\bibfield{author}{\bibinfo{person}{Sunjay Cauligi}, \bibinfo{person}{Craig Disselkoen}, \bibinfo{person}{Klaus~v Gleissenthall}, \bibinfo{person}{Dean Tullsen}, \bibinfo{person}{Deian Stefan}, \bibinfo{person}{Tamara Rezk}, {and} \bibinfo{person}{Gilles Barthe}.} \bibinfo{year}{2020}\natexlab{}.
\newblock \showarticletitle{Constant-time foundations for the new {Spectre} era}. In \bibinfo{booktitle}{\emph{ACM Conference on Programming Language Design and Implementation (PLDI)}}.
\newblock


\bibitem[Cauligi et~al\mbox{.}(2022)]%
        {cauligi2022sok}
\bibfield{author}{\bibinfo{person}{Sunjay Cauligi}, \bibinfo{person}{Craig Disselkoen}, \bibinfo{person}{Daniel Moghimi}, \bibinfo{person}{Gilles Barthe}, {and} \bibinfo{person}{Deian Stefan}.} \bibinfo{year}{2022}\natexlab{}.
\newblock \showarticletitle{{SoK}: Practical Foundations for Software {Spectre} Defenses}. In \bibinfo{booktitle}{\emph{IEEE Symposium on Security and Privacy (SP)}}.
\newblock


\bibitem[Chen et~al\mbox{.}(2024)]%
        {chen2024gadgetspinner}
\bibfield{author}{\bibinfo{person}{Yun Chen}, \bibinfo{person}{Ali Hajiabadi}, {and} \bibinfo{person}{Trevor~E Carlson}.} \bibinfo{year}{2024}\natexlab{}.
\newblock \showarticletitle{GadgetSpinner: A new transient execution primitive using the Loop Stream Detector}. In \bibinfo{booktitle}{\emph{IEEE International Symposium on High-Performance Computer Architecture (HPCA)}}.
\newblock


\bibitem[Choudhary et~al\mbox{.}(2023)]%
        {choudhary2023declassiflow}
\bibfield{author}{\bibinfo{person}{Rutvik Choudhary}, \bibinfo{person}{Alan Wang}, \bibinfo{person}{Zirui~Neil Zhao}, \bibinfo{person}{Adam Morrison}, {and} \bibinfo{person}{Christopher~W Fletcher}.} \bibinfo{year}{2023}\natexlab{}.
\newblock \showarticletitle{Declassiflow: A Static Analysis for Modeling Non-Speculative Knowledge to Relax Speculative Execution Security Measures}. In \bibinfo{booktitle}{\emph{ACM Conference on Computer and Communications Security (CCS)}}.
\newblock


\bibitem[Choudhary et~al\mbox{.}(2021)]%
        {choudhary2021speculative}
\bibfield{author}{\bibinfo{person}{Rutvik Choudhary}, \bibinfo{person}{Jiyong Yu}, \bibinfo{person}{Christopher Fletcher}, {and} \bibinfo{person}{Adam Morrison}.} \bibinfo{year}{2021}\natexlab{}.
\newblock \showarticletitle{Speculative Privacy Tracking ({SPT}): Leaking Information From Speculative Execution Without Compromising Privacy}. In \bibinfo{booktitle}{\emph{IEEE/ACM International Symposium on Microarchitecture (MICRO)}}.
\newblock


\bibitem[curve25519-donna(2008)]%
        {curve25519-donna}
curve25519-donna \bibinfo{year}{2008}\natexlab{}.
\newblock \bibinfo{title}{curve25519-donna}.
\newblock \bibinfo{howpublished}{\url{https://code.google.com/archive/p/curve25519-donna/}}.
\newblock
\newblock
\shownote{Accessed 05-04-2024}.


\bibitem[Daniel et~al\mbox{.}(2021)]%
        {daniel2021hunting}
\bibfield{author}{\bibinfo{person}{Lesly-Ann Daniel}, \bibinfo{person}{S{\'e}bastien Bardin}, {and} \bibinfo{person}{Tamara Rezk}.} \bibinfo{year}{2021}\natexlab{}.
\newblock \showarticletitle{Hunting the haunter-efficient relational symbolic execution for {Spectre} with haunted relse}. In \bibinfo{booktitle}{\emph{Network and Distributed Systems Security (NDSS)}}.
\newblock


\bibitem[Daniel et~al\mbox{.}(2023)]%
        {daniel2023prospect}
\bibfield{author}{\bibinfo{person}{Lesly-Ann Daniel}, \bibinfo{person}{Marton Bognar}, \bibinfo{person}{Job Noorman}, \bibinfo{person}{S{\'e}bastien Bardin}, \bibinfo{person}{Tamara Rezk}, {and} \bibinfo{person}{Frank Piessens}.} \bibinfo{year}{2023}\natexlab{}.
\newblock \showarticletitle{{ProSpeCT}: Provably Secure Speculation for the Constant-Time Policy}. In \bibinfo{booktitle}{\emph{USENIX Security Symposium}}.
\newblock


\bibitem[Evtyushkin et~al\mbox{.}(2018)]%
        {evtyushkin2018branchscope}
\bibfield{author}{\bibinfo{person}{Dmitry Evtyushkin}, \bibinfo{person}{Ryan Riley}, \bibinfo{person}{Nael~CSE Abu-Ghazaleh}, \bibinfo{person}{ECE}, {and} \bibinfo{person}{Dmitry Ponomarev}.} \bibinfo{year}{2018}\natexlab{}.
\newblock \showarticletitle{{BranchScope}: A new side-channel attack on directional branch predictor}. In \bibinfo{booktitle}{\emph{ACM International Conference on Architectural Support for Programming Languages and Operating Systems (ASPLOS)}}.
\newblock


\bibitem[Fustos et~al\mbox{.}(2019)]%
        {fustos2019spectreguard}
\bibfield{author}{\bibinfo{person}{Jacob Fustos}, \bibinfo{person}{Farzad Farshchi}, {and} \bibinfo{person}{Heechul Yun}.} \bibinfo{year}{2019}\natexlab{}.
\newblock \showarticletitle{{SpectreGuard}: An efficient data-centric defense mechanism against {Spectre} attacks}. In \bibinfo{booktitle}{\emph{ACM/IEEE Design Automation Conference (DAC)}}.
\newblock


\bibitem[Guarnieri et~al\mbox{.}(2020)]%
        {guarnieri2020spectector}
\bibfield{author}{\bibinfo{person}{Marco Guarnieri}, \bibinfo{person}{Boris K{\"o}pf}, \bibinfo{person}{Jos{\'e}~F Morales}, \bibinfo{person}{Jan Reineke}, {and} \bibinfo{person}{Andr{\'e}s S{\'a}nchez}.} \bibinfo{year}{2020}\natexlab{}.
\newblock \showarticletitle{Spectector: Principled detection of speculative information flows}. In \bibinfo{booktitle}{\emph{IEEE Symposium on Security and Privacy (SP)}}.
\newblock


\bibitem[Guarnieri et~al\mbox{.}(2021)]%
        {guarnieri2021hardware}
\bibfield{author}{\bibinfo{person}{Marco Guarnieri}, \bibinfo{person}{Boris K{\"o}pf}, \bibinfo{person}{Jan Reineke}, {and} \bibinfo{person}{Pepe Vila}.} \bibinfo{year}{2021}\natexlab{}.
\newblock \showarticletitle{Hardware-software contracts for secure speculation}. In \bibinfo{booktitle}{\emph{IEEE Symposium on Security and Privacy (SP)}}.
\newblock


\bibitem[Guo et~al\mbox{.}(2020)]%
        {guo2020specusym}
\bibfield{author}{\bibinfo{person}{Shengjian Guo}, \bibinfo{person}{Yueqi Chen}, \bibinfo{person}{Peng Li}, \bibinfo{person}{Yueqiang Cheng}, \bibinfo{person}{Huibo Wang}, \bibinfo{person}{Meng Wu}, {and} \bibinfo{person}{Zhiqiang Zuo}.} \bibinfo{year}{2020}\natexlab{}.
\newblock \showarticletitle{{SpecuSym}: Speculative symbolic execution for cache timing leak detection}. In \bibinfo{booktitle}{\emph{ACM/IEEE International Conference on Software Engineering (ICSE)}}.
\newblock


\bibitem[Hajiabadi et~al\mbox{.}(2024)]%
        {hajiabadi2024levioso}
\bibfield{author}{\bibinfo{person}{Ali Hajiabadi}, \bibinfo{person}{Archit Agarwal}, \bibinfo{person}{Andreas Diavastos}, {and} \bibinfo{person}{Trevor~E Carlson}.} \bibinfo{year}{2024}\natexlab{}.
\newblock \showarticletitle{Levioso: Efficient Compiler-Informed Secure Speculation}. In \bibinfo{booktitle}{\emph{ACM/IEEE Design Automation Conference (DAC)}}.
\newblock


\bibitem[Hajiabadi and Carlson(2024)]%
        {hajiabadi2024conjuring}
\bibfield{author}{\bibinfo{person}{Ali Hajiabadi} {and} \bibinfo{person}{Trevor~E Carlson}.} \bibinfo{year}{2024}\natexlab{}.
\newblock \showarticletitle{Conjuring: Leaking Control Flow via Speculative Fetch Attacks}. In \bibinfo{booktitle}{\emph{ACM/IEEE Design Automation Conference (DAC)}}.
\newblock


\bibitem[Hajiabadi et~al\mbox{.}(2021)]%
        {hajiabadi2021noreba}
\bibfield{author}{\bibinfo{person}{Ali Hajiabadi}, \bibinfo{person}{Andreas Diavastos}, {and} \bibinfo{person}{Trevor~E Carlson}.} \bibinfo{year}{2021}\natexlab{}.
\newblock \showarticletitle{{NOREBA}: a compiler-informed non-speculative out-of-order commit processor}. In \bibinfo{booktitle}{\emph{ACM International Conference on Architectural Support for Programming Languages and Operating Systems (ASPLOS)}}.
\newblock


\bibitem[Hertogh et~al\mbox{.}(2023)]%
        {hertogh2023quarantine}
\bibfield{author}{\bibinfo{person}{Math{\'e} Hertogh}, \bibinfo{person}{Manuel Wiesinger}, \bibinfo{person}{Sebastian {\"O}sterlund}, \bibinfo{person}{Marius Muench}, \bibinfo{person}{Nadav Amit}, \bibinfo{person}{Herbert Bos}, {and} \bibinfo{person}{Cristiano Giuffrida}.} \bibinfo{year}{2023}\natexlab{}.
\newblock \showarticletitle{Quarantine: Mitigating Transient Execution Attacks with Physical Domain Isolation}. In \bibinfo{booktitle}{\emph{International Symposium on Research in Attacks, Intrusions and Defenses (RAID)}}.
\newblock


\bibitem[Horn(2018)]%
        {horn2018speculative}
\bibfield{author}{\bibinfo{person}{Jann Horn}.} \bibinfo{year}{2018}\natexlab{}.
\newblock \bibinfo{title}{speculative execution, variant 4: speculative store bypass}.
\newblock
\newblock


\bibitem[intel-affected-cpus({[n.\,d.]})]%
        {intel-affected-cpus}
intel-affected-cpus \bibinfo{year}{[n.\,d.]}\natexlab{}.
\newblock \bibinfo{title}{Affected Processors: Guidance for Security Issues on Intel Processors}.
\newblock \bibinfo{howpublished}{\url{https://software.intel.com/content/www/us/en/develop/articles/software-security-guidance/secure-coding/mitigate-timing-side-channel-crypto-implementation.html}}.
\newblock
\newblock
\shownote{Accessed 20-11-2023}.


\bibitem[Jim{\'e}nez et~al\mbox{.}(2001)]%
        {jimenez2001boolean}
\bibfield{author}{\bibinfo{person}{Daniel~A Jim{\'e}nez}, \bibinfo{person}{Heather~L Hanson}, {and} \bibinfo{person}{Calvin Lin}.} \bibinfo{year}{2001}\natexlab{}.
\newblock \showarticletitle{Boolean formula-based branch prediction for future technologies}. In \bibinfo{booktitle}{\emph{International Conference on Parallel Architectures and Compilation Techniques (PACT)}}.
\newblock


\bibitem[KConfig(2024)]%
        {linuxKconfig}
KConfig \bibinfo{year}{2024}\natexlab{}.
\newblock \bibinfo{title}{Linux Timer Interrupt Frequency Configuration}.
\newblock \bibinfo{howpublished}{\url{https://github.com/torvalds/linux/blob/5be63fc19fcaa4c236b307420483578a56986a37/kernel/Kconfig.hz}}.
\newblock
\newblock
\shownote{Accessed 20-08-2024}.


\bibitem[Khan et~al\mbox{.}(2022)]%
        {khan2022whisper}
\bibfield{author}{\bibinfo{person}{Tanvir~Ahmed Khan}, \bibinfo{person}{Muhammed Ugur}, \bibinfo{person}{Krishnendra Nathella}, \bibinfo{person}{Dam Sunwoo}, \bibinfo{person}{Heiner Litz}, \bibinfo{person}{Daniel~A Jim{\'e}nez}, {and} \bibinfo{person}{Baris Kasikci}.} \bibinfo{year}{2022}\natexlab{}.
\newblock \showarticletitle{Whisper: Profile-guided branch misprediction elimination for data center applications}. In \bibinfo{booktitle}{\emph{IEEE/ACM International Symposium on Microarchitecture (MICRO)}}.
\newblock


\bibitem[Khasawneh et~al\mbox{.}(2019)]%
        {khasawneh2019safespec}
\bibfield{author}{\bibinfo{person}{Khaled~N Khasawneh}, \bibinfo{person}{Esmaeil~Mohammadian Koruyeh}, \bibinfo{person}{Chengyu Song}, \bibinfo{person}{Dmitry Evtyushkin}, \bibinfo{person}{Dmitry Ponomarev}, {and} \bibinfo{person}{Nael Abu-Ghazaleh}.} \bibinfo{year}{2019}\natexlab{}.
\newblock \showarticletitle{{SafeSpec}: Banishing the spectre of a meltdown with leakage-free speculation}. In \bibinfo{booktitle}{\emph{ACM/IEEE Design Automation Conference (DAC)}}.
\newblock


\bibitem[Kocher et~al\mbox{.}(2019)]%
        {Spectre2019Kocher}
\bibfield{author}{\bibinfo{person}{Paul Kocher}, \bibinfo{person}{Jann Horn}, \bibinfo{person}{Anders Fogh}, \bibinfo{person}{Daniel Genkin}, \bibinfo{person}{Daniel Gruss}, \bibinfo{person}{Werner Haas}, \bibinfo{person}{Mike Hamburg}, \bibinfo{person}{Moritz Lipp}, \bibinfo{person}{Stefan Mangard}, \bibinfo{person}{Thomas Prescher}, \bibinfo{person}{Michael Schwarz}, {and} \bibinfo{person}{Yuval Yarom}.} \bibinfo{year}{2019}\natexlab{}.
\newblock \showarticletitle{{Spectre} Attacks: Exploiting Speculative Execution}. In \bibinfo{booktitle}{\emph{IEEE Symposium on Security and Privacy (SP)}}.
\newblock


\bibitem[Koruyeh et~al\mbox{.}(2018)]%
        {koruyeh2018spectre}
\bibfield{author}{\bibinfo{person}{Esmaeil~Mohammadian Koruyeh}, \bibinfo{person}{Khaled~N. Khasawneh}, \bibinfo{person}{Chengyu Song}, {and} \bibinfo{person}{Nael Abu-Ghazaleh}.} \bibinfo{year}{2018}\natexlab{}.
\newblock \showarticletitle{{Spectre} Returns! Speculation Attacks using the Return Stack Buffer}. In \bibinfo{booktitle}{\emph{USENIX Workshop on Offensive Technologies (WOOT)}}.
\newblock


\bibitem[Kyber(2020)]%
        {kyber}
Kyber \bibinfo{year}{2020}\natexlab{}.
\newblock \bibinfo{title}{{Kyber} - Cryptographic suite for algebraic lattices}.
\newblock \bibinfo{howpublished}{\url{https://pq-crystals.org/kyber/index.shtml}}.
\newblock
\newblock
\shownote{Accessed 20-08-2024}.


\bibitem[Li et~al\mbox{.}(2013)]%
        {li2013mcpat}
\bibfield{author}{\bibinfo{person}{Sheng Li}, \bibinfo{person}{Jung~Ho Ahn}, \bibinfo{person}{Richard~D Strong}, \bibinfo{person}{Jay~B Brockman}, \bibinfo{person}{Dean~M Tullsen}, {and} \bibinfo{person}{Norman~P Jouppi}.} \bibinfo{year}{2013}\natexlab{}.
\newblock \showarticletitle{The {McPAT} framework for multicore and manycore architectures: Simultaneously modeling power, area, and timing}.
\newblock \bibinfo{journal}{\emph{ACM Transactions on Architecture and Code Optimization (TACO)}} (\bibinfo{year}{2013}).
\newblock


\bibitem[Li et~al\mbox{.}(2011)]%
        {li2011cacti}
\bibfield{author}{\bibinfo{person}{Sheng Li}, \bibinfo{person}{Ke Chen}, \bibinfo{person}{Jung~Ho Ahn}, \bibinfo{person}{Jay~B Brockman}, {and} \bibinfo{person}{Norman~P Jouppi}.} \bibinfo{year}{2011}\natexlab{}.
\newblock \showarticletitle{{CACTI-P}: Architecture-level modeling for SRAM-based structures with advanced leakage reduction techniques}. In \bibinfo{booktitle}{\emph{IEEE/ACM International Conference on Computer-Aided Design (ICCAD)}}.
\newblock


\bibitem[Lipp et~al\mbox{.}(2018)]%
        {Lipp2018meltdown}
\bibfield{author}{\bibinfo{person}{Moritz Lipp}, \bibinfo{person}{Michael Schwarz}, \bibinfo{person}{Daniel Gruss}, \bibinfo{person}{Thomas Prescher}, \bibinfo{person}{Werner Haas}, \bibinfo{person}{Anders Fogh}, \bibinfo{person}{Jann Horn}, \bibinfo{person}{Stefan Mangard}, \bibinfo{person}{Paul Kocher}, \bibinfo{person}{Daniel Genkin}, \bibinfo{person}{Yuval Yarom}, {and} \bibinfo{person}{Mike Hamburg}.} \bibinfo{year}{2018}\natexlab{}.
\newblock \showarticletitle{{Meltdown}: Reading Kernel Memory from User Space}. In \bibinfo{booktitle}{\emph{USENIX Security Symposium}}.
\newblock


\bibitem[Lock Elision(2021)]%
        {intelLockElision}
Lock Elision \bibinfo{year}{2021}\natexlab{}.
\newblock \bibinfo{title}{Hardware Lock Elision Overview}.
\newblock \bibinfo{howpublished}{\url{https://www.intel.com/content/www/us/en/docs/cpp-compiler/developer-guide-reference/2021-8/hardware-lock-elision-overview.html}}.
\newblock
\newblock
\shownote{Accessed 23-11-2023}.


\bibitem[Loughlin et~al\mbox{.}(2021)]%
        {loughlin2021dolma}
\bibfield{author}{\bibinfo{person}{Kevin Loughlin}, \bibinfo{person}{Ian Neal}, \bibinfo{person}{Jiacheng Ma}, \bibinfo{person}{Elisa Tsai}, \bibinfo{person}{Ofir Weisse}, \bibinfo{person}{Satish Narayanasamy}, {and} \bibinfo{person}{Baris Kasikci}.} \bibinfo{year}{2021}\natexlab{}.
\newblock \showarticletitle{{DOLMA}: Securing Speculation with the Principle of Transient {Non-Observability}}. In \bibinfo{booktitle}{\emph{USENIX Security Symposium}}.
\newblock


\bibitem[Luk et~al\mbox{.}(2005)]%
        {luk2005pin}
\bibfield{author}{\bibinfo{person}{Chi-Keung Luk}, \bibinfo{person}{Robert Cohn}, \bibinfo{person}{Robert Muth}, \bibinfo{person}{Harish Patil}, \bibinfo{person}{Artur Klauser}, \bibinfo{person}{Geoff Lowney}, \bibinfo{person}{Steven Wallace}, \bibinfo{person}{Vijay~Janapa Reddi}, {and} \bibinfo{person}{Kim Hazelwood}.} \bibinfo{year}{2005}\natexlab{}.
\newblock \showarticletitle{Pin: building customized program analysis tools with dynamic instrumentation}. In \bibinfo{booktitle}{\emph{ACM Conference on Programming Language Design and Implementation (PLDI)}}.
\newblock


\bibitem[Mar{\c{c}}ais and Kingsford(2011)]%
        {marccais2011fast}
\bibfield{author}{\bibinfo{person}{Guillaume Mar{\c{c}}ais} {and} \bibinfo{person}{Carl Kingsford}.} \bibinfo{year}{2011}\natexlab{}.
\newblock \showarticletitle{A fast, lock-free approach for efficient parallel counting of occurrences of k-mers}.
\newblock \bibinfo{journal}{\emph{Bioinformatics}} (\bibinfo{year}{2011}).
\newblock


\bibitem[Mosier et~al\mbox{.}(2024)]%
        {mosier2024serberus}
\bibfield{author}{\bibinfo{person}{Nicholas Mosier}, \bibinfo{person}{Hamed Nemati}, \bibinfo{person}{John~C Mitchell}, {and} \bibinfo{person}{Caroline Trippel}.} \bibinfo{year}{2024}\natexlab{}.
\newblock \showarticletitle{Serberus: Protecting Cryptographic Code from {Spectres} at Compile-Time}. In \bibinfo{booktitle}{\emph{IEEE Symposium on Security and Privacy (SP)}}.
\newblock


\bibitem[Nir and Langley(2018)]%
        {nir2018chacha20}
\bibfield{author}{\bibinfo{person}{Yoav Nir} {and} \bibinfo{person}{Adam Langley}.} \bibinfo{year}{2018}\natexlab{}.
\newblock \bibinfo{title}{ChaCha20 and Poly1305 for IETF Protocols}.
\newblock \bibinfo{howpublished}{\url{https://www.rfc-editor.org/rfc/rfc8439}}.
\newblock


\bibitem[Oleksenko et~al\mbox{.}(2020)]%
        {oleksenko2020specfuzz}
\bibfield{author}{\bibinfo{person}{Oleksii Oleksenko}, \bibinfo{person}{Bohdan Trach}, \bibinfo{person}{Mark Silberstein}, {and} \bibinfo{person}{Christof Fetzer}.} \bibinfo{year}{2020}\natexlab{}.
\newblock \showarticletitle{{SpecFuzz}: Bringing spectre-type vulnerabilities to the surface}. In \bibinfo{booktitle}{\emph{USENIX Security Symposium}}.
\newblock


\bibitem[OpenSSL(2024)]%
        {openssl}
OpenSSL \bibinfo{year}{2024}\natexlab{}.
\newblock \bibinfo{title}{{OpenSSL} - TLS/SSL and crypto library v3.2.2}.
\newblock \bibinfo{howpublished}{\url{https://github.com/openssl/openssl/tree/openssl-3.2.2}}.
\newblock
\newblock
\shownote{Accessed 20-08-2024}.


\bibitem[Padmanabha et~al\mbox{.}(2017)]%
        {padmanabha2017mirage}
\bibfield{author}{\bibinfo{person}{Shruti Padmanabha}, \bibinfo{person}{Andrew Lukefahr}, \bibinfo{person}{Reetuparna Das}, {and} \bibinfo{person}{Scott Mahlke}.} \bibinfo{year}{2017}\natexlab{}.
\newblock \showarticletitle{Mirage cores: The illusion of many out-of-order cores using in-order hardware}. In \bibinfo{booktitle}{\emph{IEEE/ACM International Symposium on Microarchitecture (MICRO)}}.
\newblock


\bibitem[Pashrashid et~al\mbox{.}(2022)]%
        {pashrashid2022fast}
\bibfield{author}{\bibinfo{person}{Arash Pashrashid}, \bibinfo{person}{Ali Hajiabadi}, {and} \bibinfo{person}{Trevor~E Carlson}.} \bibinfo{year}{2022}\natexlab{}.
\newblock \showarticletitle{Fast, robust and accurate detection of cache-based spectre attack phases}. In \bibinfo{booktitle}{\emph{IEEE/ACM International Conference on Computer-Aided Design (ICCAD)}}.
\newblock


\bibitem[Pashrashid et~al\mbox{.}(2023)]%
        {pashrashid2023hidfix}
\bibfield{author}{\bibinfo{person}{Arash Pashrashid}, \bibinfo{person}{Ali Hajiabadi}, {and} \bibinfo{person}{Trevor~E Carlson}.} \bibinfo{year}{2023}\natexlab{}.
\newblock \showarticletitle{{HidFix}: Efficient mitigation of cache-based {Spectre} attacks through hidden rollbacks}. In \bibinfo{booktitle}{\emph{IEEE/ACM International Conference on Computer Aided Design (ICCAD)}}.
\newblock


\bibitem[Pescosta et~al\mbox{.}(2021)]%
        {pescosta2021bounded}
\bibfield{author}{\bibinfo{person}{Emmanuel Pescosta}, \bibinfo{person}{Georg Weissenbacher}, {and} \bibinfo{person}{Florian Zuleger}.} \bibinfo{year}{2021}\natexlab{}.
\newblock \showarticletitle{Bounded model checking of speculative non-interference}. In \bibinfo{booktitle}{\emph{IEEE/ACM International Conference On Computer Aided Design (ICCAD)}}.
\newblock


\bibitem[Popping the Hood on Golden Cove(2021)]%
        {golden-cove}
Popping the Hood on Golden Cove \bibinfo{year}{2021}\natexlab{}.
\newblock \bibinfo{title}{Popping the Hood on Golden Cove}.
\newblock \bibinfo{howpublished}{\url{https://chipsandcheese.com/2021/12/02/popping-the-hood-on-golden-cove/}}.
\newblock


\bibitem[Pornin(2018)]%
        {bearssl}
\bibfield{author}{\bibinfo{person}{Thomas Pornin}.} \bibinfo{year}{2018}\natexlab{}.
\newblock \bibinfo{title}{{BearSSL} - Constant-Time Crypto Library}.
\newblock \bibinfo{howpublished}{\url{https://www.bearssl.org}}.
\newblock
\newblock
\shownote{Accessed 22-11-2023}.


\bibitem[Puddu et~al\mbox{.}(2021)]%
        {puddu2021frontal}
\bibfield{author}{\bibinfo{person}{Ivan Puddu}, \bibinfo{person}{Moritz Schneider}, \bibinfo{person}{Miro Haller}, {and} \bibinfo{person}{Srdjan {\v{C}}apkun}.} \bibinfo{year}{2021}\natexlab{}.
\newblock \showarticletitle{Frontal Attack: Leaking {Control-Flow} in {SGX} via the {CPU} Frontend}. In \bibinfo{booktitle}{\emph{USENIX Security Symposium}}.
\newblock


\bibitem[Qi et~al\mbox{.}(2021)]%
        {qi2021spectaint}
\bibfield{author}{\bibinfo{person}{Zhenxiao Qi}, \bibinfo{person}{Qian Feng}, \bibinfo{person}{Yueqiang Cheng}, \bibinfo{person}{Mengjia Yan}, \bibinfo{person}{Peng Li}, \bibinfo{person}{Heng Yin}, {and} \bibinfo{person}{Tao Wei}.} \bibinfo{year}{2021}\natexlab{}.
\newblock \showarticletitle{{SpecTaint}: Speculative Taint Analysis for Discovering Spectre Gadgets}. In \bibinfo{booktitle}{\emph{The Network and Distributed System Security Symposium (NDSS)}}.
\newblock


\bibitem[Reis et~al\mbox{.}(2019)]%
        {reis2019site}
\bibfield{author}{\bibinfo{person}{Charles Reis}, \bibinfo{person}{Alexander Moshchuk}, {and} \bibinfo{person}{Nasko Oskov}.} \bibinfo{year}{2019}\natexlab{}.
\newblock \showarticletitle{Site isolation: Process separation for websites within the browser}. In \bibinfo{booktitle}{\emph{USENIX Security Symposium}}.
\newblock


\bibitem[Rotenberg et~al\mbox{.}(1996)]%
        {rotenberg1996trace}
\bibfield{author}{\bibinfo{person}{Eric Rotenberg}, \bibinfo{person}{Steve Bennett}, {and} \bibinfo{person}{James~E Smith}.} \bibinfo{year}{1996}\natexlab{}.
\newblock \showarticletitle{Trace cache: a low latency approach to high bandwidth instruction fetching}. In \bibinfo{booktitle}{\emph{IEEE/ACM International Symposium on Microarchitecture (MICRO)}}.
\newblock


\bibitem[Rotenberg et~al\mbox{.}(1997)]%
        {rotenberg1997trace}
\bibfield{author}{\bibinfo{person}{Eric Rotenberg}, \bibinfo{person}{Quinn Jacobson}, \bibinfo{person}{Yiannakis Sazeides}, {and} \bibinfo{person}{Jim Smith}.} \bibinfo{year}{1997}\natexlab{}.
\newblock \showarticletitle{Trace processors}. In \bibinfo{booktitle}{\emph{IEEE/ACM International Symposium on Microarchitecture (MICRO)}}.
\newblock


\bibitem[Saileshwar and Qureshi(2019)]%
        {saileshwar2019cleanupspec}
\bibfield{author}{\bibinfo{person}{Gururaj Saileshwar} {and} \bibinfo{person}{Moinuddin~K Qureshi}.} \bibinfo{year}{2019}\natexlab{}.
\newblock \showarticletitle{{CleanupSpec}: An "undo" approach to safe speculation}. In \bibinfo{booktitle}{\emph{IEEE/ACM International Symposium on Microarchitecture (MICRO)}}.
\newblock


\bibitem[Sakalis et~al\mbox{.}(2019)]%
        {DOM}
\bibfield{author}{\bibinfo{person}{Christos Sakalis}, \bibinfo{person}{Stefanos Kaxiras}, \bibinfo{person}{Alberto Ros}, \bibinfo{person}{Alexandra Jimborean}, {and} \bibinfo{person}{Magnus Sj{\"a}lander}.} \bibinfo{year}{2019}\natexlab{}.
\newblock \showarticletitle{Efficient invisible speculative execution through selective delay and value prediction}. In \bibinfo{booktitle}{\emph{ACM/IEEE International Symposium on Computer Architecture (ISCA)}}.
\newblock


\bibitem[Schwarz et~al\mbox{.}(2020)]%
        {schwarz2020context}
\bibfield{author}{\bibinfo{person}{Michael Schwarz}, \bibinfo{person}{Moritz Lipp}, \bibinfo{person}{Claudio Canella}, \bibinfo{person}{Robert Schilling}, \bibinfo{person}{Florian Kargl}, {and} \bibinfo{person}{Daniel Gruss}.} \bibinfo{year}{2020}\natexlab{}.
\newblock \showarticletitle{{ConTExT}: A Generic Approach for Mitigating {Spectre}.}. In \bibinfo{booktitle}{\emph{The Network and Distributed System Security Symposium (NDSS)}}.
\newblock


\bibitem[Schwarz et~al\mbox{.}(2019)]%
        {schwarz2019zombieload}
\bibfield{author}{\bibinfo{person}{Michael Schwarz}, \bibinfo{person}{Moritz Lipp}, \bibinfo{person}{Daniel Moghimi}, \bibinfo{person}{Jo Van~Bulck}, \bibinfo{person}{Julian Stecklina}, \bibinfo{person}{Thomas Prescher}, {and} \bibinfo{person}{Daniel Gruss}.} \bibinfo{year}{2019}\natexlab{}.
\newblock \showarticletitle{{ZombieLoad}: Cross-privilege-boundary data sampling}. In \bibinfo{booktitle}{\emph{ACM Conference on Computer and Communications Security (CCS)}}.
\newblock


\bibitem[Schwarzl et~al\mbox{.}(2022)]%
        {schwarzl2022robust}
\bibfield{author}{\bibinfo{person}{Martin Schwarzl}, \bibinfo{person}{Pietro Borrello}, \bibinfo{person}{Andreas Kogler}, \bibinfo{person}{Kenton Varda}, \bibinfo{person}{Thomas Schuster}, \bibinfo{person}{Michael Schwarz}, {and} \bibinfo{person}{Daniel Gruss}.} \bibinfo{year}{2022}\natexlab{}.
\newblock \showarticletitle{Robust and scalable process isolation against {Spectre} in the cloud}. In \bibinfo{booktitle}{\emph{European Symposium on Research in Computer Security (ESORICS)}}.
\newblock


\bibitem[scikit-bio(2014)]%
        {scikit}
scikit-bio \bibinfo{year}{2014}\natexlab{}.
\newblock \bibinfo{title}{scikit-bio Python Library}.
\newblock \bibinfo{howpublished}{\url{https://scikit.bio/docs/latest/index.html}}.
\newblock
\newblock
\shownote{Accessed 23-11-2023}.


\bibitem[Sherwood et~al\mbox{.}(2002)]%
        {sherwood2002automatically}
\bibfield{author}{\bibinfo{person}{Timothy Sherwood}, \bibinfo{person}{Erez Perelman}, \bibinfo{person}{Greg Hamerly}, {and} \bibinfo{person}{Brad Calder}.} \bibinfo{year}{2002}\natexlab{}.
\newblock \showarticletitle{Automatically characterizing large scale program behavior}. In \bibinfo{booktitle}{\emph{ACM International Conference on Architectural Support for Programming Languages and Operating Systems (ASPLOS)}}.
\newblock


\bibitem[Shivakumar et~al\mbox{.}(2023a)]%
        {shivakumar2023spectre}
\bibfield{author}{\bibinfo{person}{Basavesh~Ammanaghatta Shivakumar}, \bibinfo{person}{Jack Barnes}, \bibinfo{person}{Gilles Barthe}, \bibinfo{person}{Sunjay Cauligi}, \bibinfo{person}{Chitchanok Chuengsatiansup}, \bibinfo{person}{Daniel Genkin}, \bibinfo{person}{Sioli O’Connell}, \bibinfo{person}{Peter Schwabe}, \bibinfo{person}{Rui~Qi Sim}, {and} \bibinfo{person}{Yuval Yarom}.} \bibinfo{year}{2023}\natexlab{a}.
\newblock \showarticletitle{Spectre declassified: Reading from the right place at the wrong time}. In \bibinfo{booktitle}{\emph{IEEE Symposium on Security and Privacy (SP)}}.
\newblock


\bibitem[Shivakumar et~al\mbox{.}(2023b)]%
        {shivakumar2023typing}
\bibfield{author}{\bibinfo{person}{Basavesh~Ammanaghatta Shivakumar}, \bibinfo{person}{Gilles Barthe}, \bibinfo{person}{Benjamin Gr{\'e}goire}, \bibinfo{person}{Vincent Laporte}, \bibinfo{person}{Tiago Oliveira}, \bibinfo{person}{Swarn Priya}, \bibinfo{person}{Peter Schwabe}, {and} \bibinfo{person}{Lucas Tabary-Maujean}.} \bibinfo{year}{2023}\natexlab{b}.
\newblock \showarticletitle{Typing High-Speed Cryptography against Spectre v1}. In \bibinfo{booktitle}{\emph{IEEE Symposium on Security and Privacy (SP)}}.
\newblock


\bibitem[SPHINCS+(2020)]%
        {sphincs}
SPHINCS+ \bibinfo{year}{2020}\natexlab{}.
\newblock \bibinfo{title}{{SPHINCS+} - Stateless hash-based signatures}.
\newblock \bibinfo{howpublished}{\url{https://sphincs.org/}}.
\newblock
\newblock
\shownote{Accessed 20-08-2024}.


\bibitem[Taram et~al\mbox{.}(2019)]%
        {taram2019context}
\bibfield{author}{\bibinfo{person}{Mohammadkazem Taram}, \bibinfo{person}{Ashish Venkat}, {and} \bibinfo{person}{Dean Tullsen}.} \bibinfo{year}{2019}\natexlab{}.
\newblock \showarticletitle{Context-sensitive fencing: Securing speculative execution via microcode customization}. In \bibinfo{booktitle}{\emph{ACM International Conference on Architectural Support for Programming Languages and Operating Systems (ASPLOS)}}.
\newblock


\bibitem[Van~Bulck et~al\mbox{.}(2018)]%
        {van2018foreshadow}
\bibfield{author}{\bibinfo{person}{Jo Van~Bulck}, \bibinfo{person}{Marina Minkin}, \bibinfo{person}{Ofir Weisse}, \bibinfo{person}{Daniel Genkin}, \bibinfo{person}{Baris Kasikci}, \bibinfo{person}{Frank Piessens}, \bibinfo{person}{Mark Silberstein}, \bibinfo{person}{Thomas~F Wenisch}, \bibinfo{person}{Yuval Yarom}, {and} \bibinfo{person}{Raoul Strackx}.} \bibinfo{year}{2018}\natexlab{}.
\newblock \showarticletitle{Foreshadow: Extracting the keys to the Intel SGX kingdom with transient out-of-order execution}. In \bibinfo{booktitle}{\emph{USENIX Security Symposium}}.
\newblock


\bibitem[Van~Schaik et~al\mbox{.}(2019)]%
        {van2019ridl}
\bibfield{author}{\bibinfo{person}{Stephan Van~Schaik}, \bibinfo{person}{Alyssa Milburn}, \bibinfo{person}{Sebastian {\"O}sterlund}, \bibinfo{person}{Pietro Frigo}, \bibinfo{person}{Giorgi Maisuradze}, \bibinfo{person}{Kaveh Razavi}, \bibinfo{person}{Herbert Bos}, {and} \bibinfo{person}{Cristiano Giuffrida}.} \bibinfo{year}{2019}\natexlab{}.
\newblock \showarticletitle{{RIDL}: Rogue in-flight data load}. In \bibinfo{booktitle}{\emph{IEEE Symposium on Security and Privacy (SP)}}.
\newblock


\bibitem[Vassena et~al\mbox{.}(2021)]%
        {vassena2021blade}
\bibfield{author}{\bibinfo{person}{Marco Vassena}, \bibinfo{person}{Craig Disselkoen}, \bibinfo{person}{Klaus~von Gleissenthall}, \bibinfo{person}{Sunjay Cauligi}, \bibinfo{person}{Rami~G{\"o}khan K{\i}c{\i}}, \bibinfo{person}{Ranjit Jhala}, \bibinfo{person}{Dean Tullsen}, {and} \bibinfo{person}{Deian Stefan}.} \bibinfo{year}{2021}\natexlab{}.
\newblock \showarticletitle{Automatically eliminating speculative leaks from cryptographic code with blade}.
\newblock \bibinfo{journal}{\emph{ACM Symposium on Principles of Programming Languages (POPL)}} (\bibinfo{year}{2021}).
\newblock


\bibitem[Wang et~al\mbox{.}(2020)]%
        {wang2020kleespectre}
\bibfield{author}{\bibinfo{person}{Guanhua Wang}, \bibinfo{person}{Sudipta Chattopadhyay}, \bibinfo{person}{Arnab~Kumar Biswas}, \bibinfo{person}{Tulika Mitra}, {and} \bibinfo{person}{Abhik Roychoudhury}.} \bibinfo{year}{2020}\natexlab{}.
\newblock \showarticletitle{KLEESpectre: Detecting information leakage through speculative cache attacks via symbolic execution}.
\newblock \bibinfo{journal}{\emph{ACM Transactions on Software Engineering and Methodology (TOSEM)}} (\bibinfo{year}{2020}).
\newblock


\bibitem[Wang et~al\mbox{.}(2019)]%
        {wang2019oo7}
\bibfield{author}{\bibinfo{person}{Guanhua Wang}, \bibinfo{person}{Sudipta Chattopadhyay}, \bibinfo{person}{Ivan Gotovchits}, \bibinfo{person}{Tulika Mitra}, {and} \bibinfo{person}{Abhik Roychoudhury}.} \bibinfo{year}{2019}\natexlab{}.
\newblock \showarticletitle{oo7: Low-overhead defense against spectre attacks via program analysis}.
\newblock \bibinfo{journal}{\emph{IEEE Transactions on Software Engineering}} (\bibinfo{year}{2019}).
\newblock


\bibitem[Weisse et~al\mbox{.}(2019)]%
        {weisse2019nda}
\bibfield{author}{\bibinfo{person}{Ofir Weisse}, \bibinfo{person}{Ian Neal}, \bibinfo{person}{Kevin Loughlin}, \bibinfo{person}{Thomas~F Wenisch}, {and} \bibinfo{person}{Baris Kasikci}.} \bibinfo{year}{2019}\natexlab{}.
\newblock \showarticletitle{{NDA}: Preventing speculative execution attacks at their source}. In \bibinfo{booktitle}{\emph{IEEE/ACM International Symposium on Microarchitecture (MICRO)}}.
\newblock


\bibitem[Wikner and Razavi(2022)]%
        {wikner2022retbleed}
\bibfield{author}{\bibinfo{person}{Johannes Wikner} {and} \bibinfo{person}{Kaveh Razavi}.} \bibinfo{year}{2022}\natexlab{}.
\newblock \showarticletitle{{RETBLEED}: Arbitrary Speculative Code Execution with Return Instructions}. In \bibinfo{booktitle}{\emph{USENIX Security Symposium}}.
\newblock


\bibitem[Wu and Wang(2019)]%
        {wu2019abstract}
\bibfield{author}{\bibinfo{person}{Meng Wu} {and} \bibinfo{person}{Chao Wang}.} \bibinfo{year}{2019}\natexlab{}.
\newblock \showarticletitle{Abstract interpretation under speculative execution}. In \bibinfo{booktitle}{\emph{ACM Conference on Programming Language Design and Implementation (PLDI)}}.
\newblock


\bibitem[Yan et~al\mbox{.}(2018)]%
        {yan2018invisispec}
\bibfield{author}{\bibinfo{person}{Mengjia Yan}, \bibinfo{person}{Jiho Choi}, \bibinfo{person}{Dimitrios Skarlatos}, \bibinfo{person}{Adam Morrison}, \bibinfo{person}{Christopher Fletcher}, {and} \bibinfo{person}{Josep Torrellas}.} \bibinfo{year}{2018}\natexlab{}.
\newblock \showarticletitle{{InvisiSpec}: Making speculative execution invisible in the cache hierarchy}. In \bibinfo{booktitle}{\emph{IEEE/ACM International Symposium on Microarchitecture (MICRO)}}.
\newblock


\bibitem[Yavarzadeh et~al\mbox{.}(2024)]%
        {yavarzadeh2024pathfinder}
\bibfield{author}{\bibinfo{person}{Hosein Yavarzadeh}, \bibinfo{person}{Archit Agarwal}, \bibinfo{person}{Max Christman}, \bibinfo{person}{Christina Garman}, \bibinfo{person}{Daniel Genkin}, \bibinfo{person}{Andrew Kwong}, \bibinfo{person}{Daniel Moghimi}, \bibinfo{person}{Deian Stefan}, \bibinfo{person}{Kazem Taram}, {and} \bibinfo{person}{Dean Tullsen}.} \bibinfo{year}{2024}\natexlab{}.
\newblock \showarticletitle{Pathfinder: High-Resolution Control-Flow Attacks Exploiting the Conditional Branch Predictor}. In \bibinfo{booktitle}{\emph{ACM International Conference on Architectural Support for Programming Languages and Operating Systems (ASPLOS)}}.
\newblock


\bibitem[Yu et~al\mbox{.}(2019a)]%
        {yu2019data}
\bibfield{author}{\bibinfo{person}{Jiyong Yu}, \bibinfo{person}{Lucas Hsiung}, \bibinfo{person}{Mohamad El'Hajj}, {and} \bibinfo{person}{Christopher~W Fletcher}.} \bibinfo{year}{2019}\natexlab{a}.
\newblock \showarticletitle{Data Oblivious ISA Extensions for Side Channel-Resistant and High Performance Computing}. In \bibinfo{booktitle}{\emph{The Network and Distributed System Security Symposium (NDSS)}}.
\newblock


\bibitem[Yu et~al\mbox{.}(2019b)]%
        {yu2019speculative}
\bibfield{author}{\bibinfo{person}{Jiyong Yu}, \bibinfo{person}{Mengjia Yan}, \bibinfo{person}{Artem Khyzha}, \bibinfo{person}{Adam Morrison}, \bibinfo{person}{Josep Torrellas}, {and} \bibinfo{person}{Christopher~W Fletcher}.} \bibinfo{year}{2019}\natexlab{b}.
\newblock \showarticletitle{Speculative taint tracking ({STT}) a comprehensive protection for speculatively accessed data}. In \bibinfo{booktitle}{\emph{IEEE/ACM International Symposium on Microarchitecture (MICRO)}}.
\newblock


\bibitem[Zangeneh et~al\mbox{.}(2020)]%
        {zangeneh2020branchnet}
\bibfield{author}{\bibinfo{person}{Siavash Zangeneh}, \bibinfo{person}{Stephen Pruett}, \bibinfo{person}{Sangkug Lym}, {and} \bibinfo{person}{Yale~N Patt}.} \bibinfo{year}{2020}\natexlab{}.
\newblock \showarticletitle{{BranchNet}: A convolutional neural network to predict hard-to-predict branches}. In \bibinfo{booktitle}{\emph{IEEE/ACM International Symposium on Microarchitecture (MICRO)}}.
\newblock


\bibitem[Zhang et~al\mbox{.}(2023)]%
        {zhang2023ultimate}
\bibfield{author}{\bibinfo{person}{Zhiyuan Zhang}, \bibinfo{person}{Gilles Barthe}, \bibinfo{person}{Chitchanok Chuengsatiansup}, \bibinfo{person}{Peter Schwabe}, {and} \bibinfo{person}{Yuval Yarom}.} \bibinfo{year}{2023}\natexlab{}.
\newblock \showarticletitle{Ultimate {SLH}: Taking Speculative Load Hardening to the Next Level}. In \bibinfo{booktitle}{\emph{USENIX Security Symposium}}.
\newblock


\bibitem[Zhao et~al\mbox{.}(2020)]%
        {invarspec}
\bibfield{author}{\bibinfo{person}{Zirui~Neil Zhao}, \bibinfo{person}{Houxiang Ji}, \bibinfo{person}{Mengjia Yan}, \bibinfo{person}{Jiyong Yu}, \bibinfo{person}{Christopher~W Fletcher}, \bibinfo{person}{Adam Morrison}, \bibinfo{person}{Darko Marinov}, {and} \bibinfo{person}{Josep Torrellas}.} \bibinfo{year}{2020}\natexlab{}.
\newblock \showarticletitle{Speculation invariance ({InvarSpec}): Faster safe execution through program analysis}. In \bibinfo{booktitle}{\emph{IEEE/ACM International Symposium on Microarchitecture (MICRO)}}.
\newblock


\bibitem[Zinzindohou{\'e} et~al\mbox{.}(2017)]%
        {zinzindohoue2017hacl}
\bibfield{author}{\bibinfo{person}{Jean-Karim Zinzindohou{\'e}}, \bibinfo{person}{Karthikeyan Bhargavan}, \bibinfo{person}{Jonathan Protzenko}, {and} \bibinfo{person}{Benjamin Beurdouche}.} \bibinfo{year}{2017}\natexlab{}.
\newblock \showarticletitle{{HACL*}: A verified modern cryptographic library}. In \bibinfo{booktitle}{\emph{ACM Conference on Computer and Communications Security (CCS)}}.
\newblock


\end{thebibliography}

\end{document}